\documentclass[english,twocolumn,transaction,10pt]{IEEEtran}

\usepackage{comment}
\usepackage[T1]{fontenc}
\usepackage{float}
\usepackage{amsmath}
\usepackage{amsthm}
\usepackage{amssymb}
\usepackage{stackrel}
\usepackage{graphicx}
\usepackage{setspace}
\usepackage{url}%in order to show the web link in the reference
\usepackage{bm}
\usepackage{caption}
\usepackage{subfigure}
\usepackage{xcolor}
\usepackage{cases}
\usepackage{algorithm}
\usepackage{algorithmic}
\usepackage{stfloats}
\makeatletter

\floatstyle{ruled}
\newfloat{algorithm}{tbp}{loa}
\providecommand{\algorithmname}{Algorithm}
\floatname{algorithm}{\protect\algorithmname}

\theoremstyle{plain}

\theoremstyle{definition}

\theoremstyle{plain}

\theoremstyle{definition}

\theoremstyle{plain}

\newtheorem{define}{Definition}
\newtheorem{assumption}{Assumption}

\newtheorem{theo}{Theorem}

\newtheorem{proposition}{Proposition}
\IEEEoverridecommandlockouts
\usepackage{ifpdf}
\usepackage{cite}
\ifCLASSOPTIONcompsoc
\usepackage[caption=false,font=normalsize,labelfont=sf,textfont=sf]{subfig}
\else
\usepackage[caption=false,font=footnotesize]{subfig}
\fi
\usepackage{amsmath}
\usepackage{mathtools}
\usepackage{booktabs}
\usepackage{multirow}
\usepackage{setspace}
\usepackage{lettrine}
\usepackage{setspace}
\usepackage{array}
\@ifundefined{showcaptionsetup}{}{
	\PassOptionsToPackage{caption=false}{subfig}}
\usepackage{subfig}
\makeatother
\allowdisplaybreaks[4]
\usepackage{diagbox}
\usepackage{babel}

\newtheorem{lemma}{Lemma}

\begin{document}
	\captionsetup[figure]{font={small}, name={Fig.}, labelsep=period}
	
	\title{\huge{Decentralized Federated Learning Over Imperfect Communication Channels}}
		\author{Weicai~Li,~Tiejun~Lv, \textit{Senior Member, IEEE},~Wei~Ni, \textit{Fellow, IEEE},  
 Jingbo~Zhao,\\
 Ekram~Hossain, \textit{Fellow, IEEE}, and H. Vincent Poor, \textit{Life Fellow, IEEE}

\thanks{W. Li, T. Lv and J. Zhao are with the School of Information and Communication Engineering, Beijing University of Posts and Telecommunications (BUPT), Beijing 100876, China (e-mail: \{liweicai, lvtiejun, zhjb\}@bupt.edu.cn). 

W.~Ni is with Data61, Commonwealth Science and Industrial Research Organisation (CSIRO), Sydney, New South Wales, 2122, Australia (e-mail: wei.ni@data61.csiro.au). 

E. Hossain is with the Department of Electrical and Computer Engineering, University of Manitoba, Canada (email: ekram.hossain@umanitoba.ca). 

H. V. Poor is with the Department of Electrical and Computer Engineering, Princeton University, Princeton, NJ 08544, USA (email: poor@princeton.edu).}	
	}

	\maketitle
\begin{abstract}
 This paper analyzes the impact of imperfect communication channels on decentralized federated learning (D-FL) and subsequently determines the optimal number of local aggregations per training round, adapting to the network topology and imperfect channels. We start by deriving the bias of locally aggregated D-FL models under imperfect channels from the ideal global models requiring perfect channels and aggregations. The bias reveals that excessive local aggregations can accumulate communication errors and degrade convergence. Another important aspect is that we analyze a convergence upper bound of D-FL based on the bias. By minimizing the bound, the optimal number of local aggregations is identified to balance a trade-off with accumulation of communication errors in the absence of knowledge of the channels. With this knowledge, the impact of communication errors can be alleviated, allowing the convergence upper bound to decrease throughout aggregations.
Experiments validate our convergence analysis and also identify the optimal number of local aggregations on two widely considered image classification tasks. It is seen that D-FL, with an optimal number of local aggregations, can outperform its potential alternatives
by over 10\% in training accuracy. 	
\end{abstract}
	

	\section{Introduction}
Different federated learning (FL) architectures, i.e., centralized FL (C-FL)~\cite{9562559}, decentralized FL (D-FL)~\cite{1638541,10038786}, and semi-D-FL~\cite{9838880,9772065} have been proposed to cope with different application scenarios. It is generally difficult for C-FL to scale due to the limited communication capacity and range of a central aggregator~\cite{9562559}. C-FL is also susceptible to a single point of failure. Semi-D-FL integrates hybrid FL, a device-to-server communication paradigm, and device-to-device communication within local clusters, e.g., using the gossip protocol~\cite{1238221} for intra/inter-cluster communications~\cite{9838880} or hierarchical C-FL for democratized machine learning \cite{9772065}. These architectures~\cite{9838880,9772065} may still face a single point of failure at the central server. On the other hand, many distributed systems, e.g., mobile ad-hoc networks, can have dynamic topology and do not have a center capable of communicating with all devices. D-FL can enable local model aggregations at each device and eliminate the need for a central server. It is flexible and scalable for large-scale or disconnected networks with low bandwidth or low computing powers.

A critical challenge arising from D-FL is the impact of local model inconsistency resulting from imperfect/unreliable communication channels and aggregation strategies among the agents on the convergence of D-FL has not been accounted for in the literature. While the convergence and conformity of D-FL could benefit from more aggregations~\cite{XIAO200465}, existing works, e.g.,~\cite{9563232,9772390}, often considered only one aggregation per communication round and could produce models differing substantially among the clients. On the other hand, errors resulting from model transmissions in imperfect channels can accumulate and hinder model convergence. It is important to specify the number of local aggregations with a balanced consideration between model aggregation and error accumulation. 

The quantification of the impact of local model inconsistency on the convergence of D-FL is non-trivial. Specifically, the inputs to the local training of different nodes can differ substantially, resulting from their different positions in the network topology and imperfect channels. As a consequence, the local models of D-FL can further differ, not only hindering the model convergence but making it challenging to analyze the convergence. Existing studies have not considered the joint impact of imperfect channels and aggregation strategies on D-FL. Several studies have only considered imperfect channels on C-FL~\cite{9726793,9435350,9515709}. However, the analysis of C-FL cannot apply to D-FL due to their different network protocols and aggregation methods. Other studies \cite{9772390,10032555,9563232} considered D-FL protocols under assumptions of imperfect channels, but none analyzed the impact of the accumulation of communication errors resulting from multiple local aggregations.

This paper investigates the impact of local model inconsistency resulting from imperfect/unreliable channels and local model aggregation strategies on the convergence of D-FL, where parts of models can be corrupted due to transmission errors and precluded from local model aggregations. The convergence upper bound of D-FL is rigorously analyzed. Insights are drawn to help optimally decide the number of local aggregations per local training round, adapting to the topology and channel conditions of D-FL.

The key contributions of the paper are as follows.
\begin{itemize}
    \item We consider a new scenario of D-FL under imperfect channels, where the initialization of local training is designed to allow different devices' models to converge after training.
    
    \item We derive the bias of locally aggregated D-FL models under imperfect channels from the ideal global model, and analyze the convergence upper bound of D-FL under imperfect channels based on the bias.
    
    \item  We reveal that increasing local iterations facilitates convergence in imperfect channels, while increasing local aggregations entails a trade-off with the accumulation of communication errors within the convergence bound.
    
    \item 
    We minimize the convergence upper bound by optimizing the number of local aggregations to restrain the impact of communication errors when the channel conditions are unknown \textit{a-priori}. On the other hand, the \textit{a-priori} knowledge of the channels can help mitigate the impact so that more local aggregations improve convergence.
    \end{itemize}
Extensive experiments validate the new convergence analysis and accordingly identify the optimal number of local aggregations. Two image classification tasks, using a convolutional neural network (CNN) on the Federated Extended MNIST (F-MNIST) dataset and using the 18-layer residential network (ResNet-18) on the Federated CIFAR100 (F-CIFAR100) dataset, are performed. When the channel conditions are unknown \textit{a-priori}, with the optimal number of local aggregations, D-FL can outperform its potential alternatives, i.e., C-FL and D-FL without optimizing the local aggregation number, by 12.5\% and 10\% in training accuracy, respectively. 

The remainder of this paper is organized as follows. The related works are discussed in Section \ref{section: related work}, followed by the system model in Section \ref{section:system}. In Section \ref{section:convergence analysis}, we analyze the convergence upper bound of D-FL under imperfect channels. In Section \ref{section:problem formulation}, we optimize the number of local aggregations for effective convergence. Experimental results are provided in Section \ref{section:results}, followed by conclusions in Section~\ref{section:con}. 

\textit{Notation:}  $(\cdot)^\dagger$ denotes transpose; $\Vert\cdot\Vert$ is the 2-norm of a vector or a matrix; $\lceil\cdot\rceil$ takes ceiling;  $|\cdot|$ takes cardinality; $\circ$ takes the element-wise multiplication of two vectors or matrices; $\mathbb{E}_{\mathbf{e}}(\cdot)$ and $\mathbb{D}_{\mathbf{e}}(\cdot)$ take element-wise expectation and variance over communication errors, respectively; $\mathbb{E}_{\xi}(\cdot)$ and $\mathbb{E}_{\{\xi_1,\cdots,\xi_K\}}(\cdot)$ take expectations over a mini-batch $\xi$ and $K$ mini-batches $\xi_1,\cdots,\xi_K$, respectively; $\mathbb{E}_{\xi,\mathbf{e}}(\cdot)$ takes expectation over a mini-batch and communication errors; $\mathbb{E}_{\{\xi_1,\cdots,\xi_K\},\mathbf{e}}(\cdot)$ takes expectation over $K$ mini-batches and communication errors.

\section{Related Work}\label{section: related work}

Some studies have focused on the model aggregation of D-FL, typically under ideal communication channels. A fast-linear iteration approach for decentralized averaging under the gossip-based protocol is the most widely used synchronous D-FL scheme to help the locally aggregated parameters among all users approach the global model of C-FL~\cite{XIAO200465}. Several D-FL studies~\cite{9709678,bornstein2022swift,8950073,9838891} allowed clients to train and aggregate models asynchronously. In general, considerably more iterations are required for D-FL to converge.

Some other studies, e.g.,~\cite{9726793,9435350,9515709}, have investigated the adverse effects of unreliable communications on C-FL, but no consideration has been given to D-FL. In \cite{9726793}, the aggregation coefficients were designed for C-FL to reduce the model bias resulting from transmission errors. In~\cite{9435350}, the authors proved that the user transmission success probability, affected by unreliable and resource-constrained wireless channels, is a vital factor for the convergence of C-FL. An optimal resource allocation strategy was proposed to improve the transmission success probability and accelerate convergence.

Other studies have designed communication strategies for D-FL systems. But the impact of {accumulated communication errors resulting from multiple local aggregations} on the convergence of D-FL has been overlooked. In~\cite{9772390}, the transmission between two devices is considered unreliable when their communication rate is lower than a certain threshold. The bandwidth is allocated to minimize a weighted sum of the local model biases caused by the imperfect channels. However, the convergence bound of D-FL is analyzed under the assumption that the transmission errors were negligible. In~\cite{10032555}, a user datagram protocol with erasure coding and retransmission was employed to recover content lost during transmissions at the cost of delays and signaling overhead. The authors of~\cite{9660377} designed a model compression method to lower the requirements for bandwidth and storage resources of edge devices, and hence improve the system communication efficiency. However, the model compression method would distort the local models and compromise training accuracy.

\section{System Model and Assumptions}\label{section:system}
We consider a D-FL system, where individual devices aggregate their models in a fully decentralized manner. In what follows, we describe its workflow and channel model. 

\subsection{D-FL Model}

\begin{algorithm} %[ht]
\caption{The proposed D-FL protocol}
\textbf{Input:} $\{\mathcal{D}_n, \forall n\}$, $\eta$, $I$, $J$.

\textbf{Initialization:} $\{\boldsymbol{\omega}_{n,0}^1, \forall n\}$, $t=1$.

\begin{algorithmic}[1]
 \FOR{$t=1,2,\cdots$}
    \STATE Update $\mathcal{G}^t$ and $\mathbf{C}^t$.
    \FOR{each device $n\in\mathcal{N}$} 
    \STATE \% \textit{Execute $I$ local training iterations}
    \FOR{$i=1,\cdots,I$}
    \STATE Randomly select a mini-batch $\xi_{n,i}^t$  to train and update its local model by $\boldsymbol{\omega}^{t}_{n,i}=\boldsymbol{\omega}^{t}_{n,i-1}-\eta\mathbf{g}_{\xi,n,i-1}^t$.
    \ENDFOR
    \STATE \% \textit{Execute $J$ local model aggregations}
    \STATE Set $\mathbf{w}_{n,0}^t=Np_n\boldsymbol{\omega}_{n,I}^t$.   
    \FOR{$j=1,\cdots,J$}
    \STATE Transmit ${\mathbf{w}}_{n,j-1}^{t}$ and receive ${\mathbf{w}}_{m,j-1}^{t}$ from its neighboring devices $m$ with $(n,m)\in\mathcal{E}^t$. 
    \STATE  Conduct an imperfect local aggregation to obtain $ {\mathbf{w}}_{n,j}^{t}=\sum_{\forall m}c^{t}_{n,m}(\mathbf{e}^{t}_{n,m,j}\circ{\mathbf{w}}_{m,j-1}^{t}).$
    \ENDFOR
    \STATE Set $\boldsymbol{\omega}_{n,0}^{t+1}=\mathbf{w}_{n,J}^t$ for the next training round.
     \ENDFOR
      \ENDFOR
\end{algorithmic}
\end{algorithm}

\begin{figure}[t]
		\centering{}
		\includegraphics[scale=0.39]{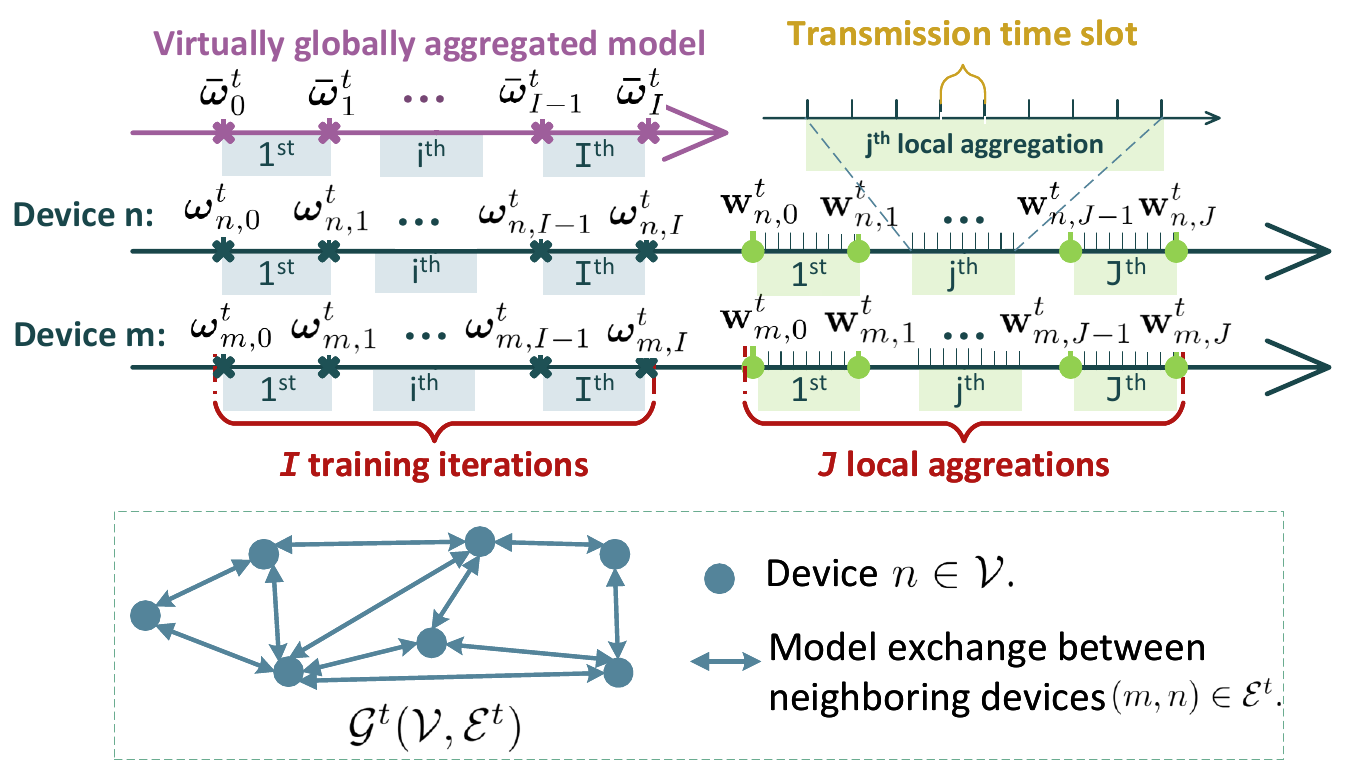}
		\caption{The timeline of the $t$-th training round of D-FL.}
		\label{fig:timeline1}
	\end{figure}
The considered D-FL system has $N$ edge devices and no central point. The topology of the network may change between training rounds, e.g., due to the mobility of the devices, but remains unchanged in a round. The topology can be given by an undirected graph $\mathcal{G}^t(\mathcal{V}, \mathcal{E}^t)$ with $\mathcal{V} = \{1,\cdots,N\}$ being the set of devices and  $\mathcal{E}^t= \{(n, m),\forall n,m \in\mathcal{V}\}$ being the set of edges if devices $n$ and $m$ are directly connected~\cite{10044204}. Assume that there is no self-loop, i.e., $(n, n) \notin \mathcal{E}^t,\, \forall n\in \mathcal{V}$. Device $ n$ owns a dataset ${{\cal D}}_n$ with $D_n=|{{\cal D}}_n|$ data points. The local loss function of the device is  
\begin{align}
 F_n(\boldsymbol{\omega})=\frac{1}{D_n}{\sum}_{d\in{\cal D}_{n}} {f}(d,\boldsymbol{\omega}),
\end{align}
where $\boldsymbol{\omega}\in {\mathbb{R}^{M\times1}}$ denotes the model parameter, $d \in {{\cal D}}_n$ is a data sample, and ${f}(d,\boldsymbol{\omega})$ is the loss function of $d$ and $\boldsymbol{\omega}$. 

The objective of the D-FL system is to  minimize the global loss function $\underset{\boldsymbol{\omega}\in \mathbb{R}^{M\times1}}{\min} F(\boldsymbol{\omega})$, and find the optimal model parameter $\boldsymbol{\omega}^*=\underset {\boldsymbol{\omega}\in \mathbb{R}^{M\times1}}{ \arg\, \min}\,F(\boldsymbol{\omega})$, where the global loss function $F(\boldsymbol{\omega})$ is the weighted loss of all devices, as given by 
 \begin{align}
      F(\boldsymbol{\omega})={\sum}_{n\in \mathcal{V}} p_nF_n(\boldsymbol{\omega}),\label{FW}
 \end{align} 
where $p_n=\frac{D_n}{{\sum}_{\forall n} D_n}$ is the weighting coefficient of device $n$. Let $F_n^*$ denote the minimum value of $F_n(\boldsymbol{\omega})$ and $F^*$ denote the minimum value of $F(\boldsymbol{\omega})$.

The D-FL system operates on the basis of synchronized training rounds. Each device conducts local training, model transmission, and model aggregation in a round. This is different from C-FL, which requires a central point to aggregate the local models and distribute the global models~\cite{pmlr-v54-mcmahan17a,DBLP}. The workflow in a round of D-FL is summarized in \textbf{Algorithm~1} and depicted in Fig. \ref{fig:timeline1}.

\subsubsection{Local Model Training}
In the $t$-th ($t = 1,2,\cdots$) round, each device $n$ trains $I$ iterations locally based on its local dataset ${\cal D}_n$ using {stochastic gradient descent (SGD)}. The local training starts with its locally aggregated model of the $(t-1)$-th training round. During the local training, the local model of the $i$-th iteration, denoted by $\boldsymbol{\omega}^{t}_{n,i}$, is updated by 
\begin{align}\label{epoch_imperfect}
\boldsymbol{\omega}^{t}_{n,i}=\boldsymbol{\omega}^{t}_{n,i-1}-\eta\mathbf{g}_{\xi,n,i-1}^t,\forall i=1,\cdots,I,  
\end{align}
where $\eta$ is the learning rate, $\mathbf{g}_{\xi,n,i}^t=\nabla F_n(\boldsymbol{\omega}^{t}_{n,i},\xi_{n,i}^{t})$ denotes the stochastic gradient of device $n$ with $\xi_{n,i}^t$ being a mini-batch sampled from the local dataset of the device in the $i$-th iteration of the $t$-th round, and $\nabla F_n(\boldsymbol{\omega}^{t}_{n,i},\xi^{t}_{n,i})=\frac{1}{\left|\xi^{t}_{n,i}\right|} \sum_{d\in \xi^{t}_{n,i}} \nabla f\left(d, \boldsymbol{\omega}^{t}_{n,i}\right),\,\forall  \xi^{t}_{n,i}\subseteq\mathcal{D}_n$. $\boldsymbol{\omega}^{t}_{n,0}$ is the starting local training model for the $t$-th round. The output of the local training in the round is $\boldsymbol{\omega}^{t}_{n,I}$. $\boldsymbol{\omega}^{t}_{n,0}=\mathbf{w}_{n,J}^{t-1}$, where $\mathbf{w}_{n,J}^{t-1}$ is the output of the local aggregations in the $(t-1)$-th round. 

Let $\boldsymbol{\Omega}_i^{t}\in \mathbb{R}^{N\times M}$ collect the training results of all devices at the $i$-th iteration of the $t$-th round:
\begin{align}
\boldsymbol{\Omega}_i^{t}=\left[\boldsymbol{\omega}^{t}_{1,i},\cdots,\boldsymbol{\omega}^{t}_{N,i}\right]^{\dagger}.\label{aver_modeL_parameter}
\end{align}
If these results could be globally aggregated like C-FL, the (virtually aggregated) ``global'' model of D-FL is written as
\begin{align}\left(\boldsymbol{\bar{\omega}}_i^{t}\right)^{\dagger}={\sum}_{\forall n}p_{n}\left(\boldsymbol{\omega}_{n,i}^{t}\right)^{\dagger}=\mathbf{p}\boldsymbol{\Omega}_i^{t},
\label{imprecise x}
\end{align}
where $\mathbf{p}=[p_1,\cdots,p_N]$. With reference to C-FL, we wish the aggregated model of D-FL to be consistent across all devices at the end of any round, i.e., $\boldsymbol{\bar{\omega}}_I^{t}=\mathbf{p}\boldsymbol{\Omega}_I^{t}$. This is difficult due to the distributed nature of D-FL and imperfect channels.

\subsubsection{Local Model Aggregation}\label{section3}
After local training, each device conducts $J$ local aggregations in the round. Every device is allocated $J$ time slots for broadcasting its local (or locally aggregated) models without collision; see Fig. \ref{fig:timeline1}. After every local aggregation, each device aggregates its own model and models received in the slots. The device proceeds to broadcast its locally aggregated model if there are more slots in the round, or resumes to train its model based on the aggregated model in the next round. The impact of models learned at a device multiple hops away can be captured through successive aggregations at the devices hop-by-hop from the distant device. Instead of directly aggregating models from distant devices, each device along the path performs its own aggregation with its neighboring devices.

Let $\mathbf{x}^{t}_{n,j}$ denote the locally aggregated model of device $n$ after the $j$-th local aggregation in the $t$-th training round, and 
\begin{align}\label{eq: consensus aggregation}
    \mathbf{x}_{n,j}^{t}={\sum}_{\forall m}c^{t}_{n,m}\mathbf{x}_{m,j-1}^{t}, \, j=1,\cdots,J,
\end{align}
where $c^{t}_{n,m}$ denotes the coefficient that device $n$ uses to combine the model received from device $m$ in the $t$-th round. $\mathbf{x}^{t}_{n,0}$ initializes the local aggregation of device $n$ in the round. The output of device $n$ in the round is $\mathbf{x}^{t}_{n,J}$.

Define $\mathbf{X}^{t}_{j}=\left[\mathbf{x}^{t}_{1,j},\cdots,\mathbf{x}^{t}_{N,j}\right]^{\dagger}\in\mathbb{R}^{N\times M},\, j=0,1,\cdots,J$, and define $\mathbf{C}^{t}=\{c_{n,m}^t\}_{N\times N}$ as the consensus matrix that collects all coefficients for local model aggregations in the $t$-th round. It is important to achieve consistent local models across all devices upon convergence. The consensus matrix $\mathbf{C}^t$ is designed to achieve the consistency~\cite{XIAO200465}. In the case where the channel conditions of the network are unknown \textit{a-priori}, according to~\cite{XIAO200465}, the $(n,m)$-th element of $\mathbf{C}^{t}$ is given by
\begin{equation}\label{C_r_definition}
\! c_{n,m}^t\!\! = \!\!  \begin{cases}
     \alpha^t,&\text{if $(n,m)\in\mathcal{E}^t$;}\\
    1\! -\! d^t_n\alpha^t,\!\!&\text{if } n=m;\\
    0,&\text{otherwise},
\end{cases}
\end{equation}%
where $d_n^t$ is the number of neighbors that device $n$ has in the $t$-th round. Given the topology in the round, and $\alpha^t$ is a constant and can be found by minimizing the spectral norm of $\mathbf{C}^t-\frac{{\mathbf{1}_N^{\dagger}\mathbf{1}_N}}{N}$; please refer to~\cite{XIAO200465}. Then, 
\begin{align}		\mathbf{X}_{j}^{t}=\mathbf{C}^{t}\mathbf{X}_{j-1}^{t}=\cdots=(\mathbf{C}^{t})^{j}\mathbf{X}_{0}^{t}.\label{local_c_parameter_i}
\end{align}
The matrix $\mathbf{C}^{t}$ depends on the network topology and satisfies the two constraints~\cite{XIAO200465}:
\begin{align}
\mathbf{1}_N\mathbf{C}^{t}=\mathbf{1}_N ,\;\;\mathbf{C}^{t}\mathbf{1}_N^{\dagger}=\mathbf{1}_N^{\dagger},\label{Cr_constraints}
\end{align}
which indicates the sum of any row or column of $\mathbf{C}^{t}$ is~$1$; i.e., $\mathbf{C}^{t}$ is doubly stochastic~\cite{mirsky1963results}. Hence, the $j$-th power of $\mathbf{C}^{t}$, i.e., $\left(\mathbf{C}^{t}\right)^{j}, j\in \mathbb{N}_{+}$, can converge as $j\rightarrow \infty$~\cite{XIAO200465}:
  \begin{align}
       \underset{j\rightarrow\infty,\,j\in \mathbb{N}_{+}}{\lim}&\left(\mathbf{C}^{t}\right)^{j}=\frac{\mathbf{1}_N^{\dagger}{\mathbf{1}_N}}{N}.\label{consensus_matrix:a}
  \end{align}

The initialization of D-FL is also critical for the convergence of each round and, subsequently, the convergence of D-FL.  

\begin{proposition}
To allow the locally aggregated models of D-FL to converge to $ \boldsymbol{\bar{\omega}}_I^{t}$, the local model of device $n$ used for local aggregations in the $t$-th round is initialized by $\mathbf{x}_{n,0}^{t}=Np_{n}\boldsymbol{\omega}_{n,I}^{t} $; or in other words, 
\begin{align}
\notag\mathbf{X}_{0}^{t}&=\left[\mathbf{x}_{1,0}^{t},\cdots,\mathbf{x}_{N,0}^{t}\right]^{\dagger}=\left[Np_{1}\boldsymbol{\omega}_{1,I}^{t},\cdots,Np_{N}\boldsymbol{\omega}_{N,I}^{t}\right]^{\dagger}\\&=N\mathrm{diag}(\mathbf{p})\boldsymbol{\Omega}_I^{t},\label{wr0}
\end{align}
where $\mathrm{diag}(\mathbf{p})$ is the diagonal matrix with $\mathbf{p}$ along diagonal. 
\end{proposition}

\begin{proof} 
By multiplying $\mathbf{X}_{0}^{t}$ to both sides of \eqref{consensus_matrix:a}, we have
\begin{align}
     \underset{j\rightarrow\infty}{\lim}\left(\mathbf{C}^{t}\right)^{j}\!\mathbf{X}^{t}_{0}\!\!=\!\!\frac{\mathbf{1}_N^{\dagger}{\mathbf{1}_N}}{N}\mathbf{X}^{t}_{0}\!\!=\!\!\mathbf{1}_N^{\dagger}\mathbf{p}\boldsymbol{\Omega}_I^{t}\!\!=\!\!\left[\boldsymbol{\bar{\omega}}_I^{t},\cdots,\boldsymbol{\bar{\omega}}_I^{t}\right]^{\dagger}.\end{align}
\end{proof}

\subsection{Communication Model} 
The devices send their model parameters, which are discretized and stored in bits, and packetized and modulated for transmissions. The size of a model comprising $M$ real-valued elements is $32 M$ bits (i.e., using ``float32''). We segment the model into $\left\lceil \frac{ M}{L_\mathrm{p}} \right\rceil $ packets with $L_\mathrm{p}$ elements per packet. Given the bit error rate (BER) $\epsilon_{\mathrm{B},n,m}^{t}$ over the direct link $(n,m)\in\mathcal{E}^t$, the packet error rate (PER), $\epsilon^t_{\mathrm{P},n,m}=1-(1-\epsilon_{\mathrm{B},n,m}^{t})^{32 L_{\mathrm{p}}}$, is the probability of losing a packet over the link. This is due to the fact that a cyclic redundancy check (CRC) is typically performed at a packet level. When a packet fails to pass CRC, it is often dropped in its entirety. The $j$-th locally aggregated model of device $n$ in the $t$-th round in \eqref{eq: consensus aggregation} is updated as 
 \begin{align}
    {\mathbf{w}}_{n,j}^{t}={\sum}_{\forall m}c^{t}_{n,m}(\mathbf{e}^{t}_{n,m,j}\circ{\mathbf{w}}_{m,j-1}^{t}), \, j=1,\cdots,J,\label{imperfect_wic}
 \end{align}
where $\mathbf{e}^{t}_{n,m,j}\in \mathbb{R}^{M\times1} $ collects the reception qualities of device $m$'s packets at device $n$, and is defined as
\begin{align}\label{error1}
    \mathbf{e}^{t}_{n,m,j}&\!\big[(l\!-\!1)L_{\mathrm{p}}\!+\!1:lL_{\mathrm{p}}\big]\! \!=\! \!
    \begin{cases}
       \! \multirow{2}{*}{$\mathbf{1}_{ L_{\mathrm{p}}}$,} \!&\text{if the $l$-th packet}\\
      & \text{is error-free;}\\
       \!  \mathbf{0}_{ L_{\mathrm{p}}},\!&\text{otherwise.}
    \end{cases}%
\end{align}%
Here, the probability of device $n$ correctly receiving the $l$-th packet of device $m$'s model in the $t$-th round is $1-\epsilon_{\mathrm{P},n,m}^t$. $\mathbf{e}^{t}_{n,m,j}$ and ${\mathbf{w}}_{m,j-1}^{t}$ have the same dimension. $\mathbf{e}^{t}_{n,m,j}\circ{\mathbf{w}}_{m,j-1}^{t}$ gives the correctly received model parameters from device $m$ to device $n$. Define $\mathbf{W}^{t}_{j}=\left[\mathbf{w}^{t}_{1,j},\cdots,\mathbf{w}^{t}_{N,j}\right]^{\dagger}\in\mathbb{R}^{N\times M},\, j=0,1,\cdots,J$. The initial local model for aggregations is ${\mathbf{w}}_{n,0}^{t}={\mathbf{x}}_{n,0}^{t}$ in the $t$-th round. $\mathbf{W}_{0}^{t}=\mathbf{X}_{0}^{t}$. 

When the channel conditions of the network in the $t$-th round, denoted by $\mathbf{T}^{t}$, are known \textit{a-priori}, the consensus matrix $\mathbf{C}^{t}$ can be adjusted, so that the models prone to errors are aggregated with larger aggregation coefficients, while the mean of the sum of the aggregation coefficients remains one at each device. The impact of imperfect channels on the local aggregation is compensated for. Here, $\mathbf{T}^{t}$ is defined as
\begin{small}
 \begin{align}\label{T^r}
\mathbf{T}^{t} =\left[\begin{array}{ccc}
1-\epsilon_{1,1}^{t} & \cdots & 1-\epsilon_{1,N}^{t}\\
\vdots & \ddots & \vdots\\
1-\epsilon_{N,1}^{t} & \cdots & 1-\epsilon_{N,N}^{t}
\end{array}\right].%
\end{align}%
\end{small}%
This can be achieved by making $\mathbf{C}^{t}\circ\mathbf{T}^{t}$ doubly stochastic with the $(n,m)$-th element of $\mathbf{C}^{t}$ specified by 
\begin{equation}\label{CoT_r_definition}
\! c_{n,m}^{t}\!\! = \!\!  \begin{cases}
\frac{\alpha^{t}}{1-\epsilon_{\mathrm{P},n,m}^{t}},&\text{if $(n,m)\in\mathcal{E}^{t}$;}\\
{  1\! -\! d_n^{t}\alpha^{t}},\!\!&\text{if $n=m$};\\
0,&\text{otherwise},
\end{cases}
\end{equation}%
where $\alpha^{t}$ can be found by minimizing the spectral norm of $\big(\mathbf{C}^{t}\circ\mathbf{T}^{t}-\frac{{\mathbf{1}_N^{\dagger}\mathbf{1}_N}}{N}\big)$.

  \section{Convergence analysis of D-FL under imperfect channel}\label{section:convergence analysis}
In this section, we analyze the convergence of D-FL under imperfect channels, starting with the following assumptions and definitions.

\begin{assumption}\label{assumption}
Widely considered assumptions~\cite{9264742,boyd2004convex,9770266,pmlr-v130-ruan21a,stich2018local,pmlr-v119-koloskova20a} for analyzing the bounds of FL are also considered here for analyzing the convergence bound of D-FL:
\begin{enumerate}
\item $\exists \,L>0$, the local FL objective $F_n(\boldsymbol{\omega}),\forall n$ is $L$-smooth, i.e., $\left\Vert \nabla F_{n}\left(\boldsymbol{\omega}_{1}\right)\!-\!\nabla F_{n}\left(\boldsymbol{\omega}_{2}\right)\right\Vert\! \leq\! L\left\Vert \boldsymbol{\omega}_{1}\!-\!\boldsymbol{\omega}_{2}\right\Vert $, $\forall \boldsymbol{\omega}_{1},\boldsymbol{\omega}_{2}\in \mathbb{R}^{M\times1}$. The global FL objective $F(\boldsymbol{\omega})$ is also $L$-smooth.

\item $\exists \,\mu>0$, the local FL objective $F_n(\boldsymbol{\omega}),\forall n$ is $\mu$-strongly convex, i.e., $\left\Vert \nabla F_{n}\left(\boldsymbol{\omega}_{1}\right)\!-\!\nabla F_{n}\left(\boldsymbol{\omega}_{2}\right)\right\Vert\! \geq\! \mu\left\Vert \boldsymbol{\omega}_{1}\!-\!\boldsymbol{\omega}_{2}\right\Vert $, $\forall \boldsymbol{\omega}_{1},\boldsymbol{\omega}_{2}\in \mathbb{R}^{M\times1}$. The global FL objective $F(\boldsymbol{\omega})$ is $\mu$-strongly convex.

\item $\exists \, G>0$, $\mathbb{E}_{\xi_{n,i}^{t}}\big(\Vert\mathbf{g}_{\xi,n,i}^t\Vert^2\big)\leq G^2,\,\forall i, n, t$~\cite{stich2018local}.

\item The variances of stochastic gradients are bounded, i.e., ${\mathbb{E}_{\xi_{n,i}^{t}}}\big(\Vert \mathbf{g}_{\xi,n,i}^t- \mathbf{g}_{n,i}^t\Vert^2\big)\leq \sigma^2_n, \forall  n,i,t$,~\cite{9770266,pmlr-v130-ruan21a}, where $\mathbf{g}^{t}_{n,i}=\nabla F_{n}(\boldsymbol{\omega}^{t}_{n,i})$ is the full-batch gradient of device $n$ in the $i$-th iteration of round $t$. $\mathbf{g}^{t}_{n,i}={\mathbb{E}_{\xi_{n,i}^{t}}}( \mathbf{g}_{\xi,n,i}^t)$.
\end{enumerate}
In practice, we can estimate the minimum and maximum of $\frac{\left\Vert \nabla F_{n}\left(\boldsymbol{\omega}_{1}\right)\!-\!\nabla F_{n}\left(\boldsymbol{\omega}_{2}\right)\right\Vert }{\left\Vert \boldsymbol{\omega}_{1}\!-\!\boldsymbol{\omega}_{2}\right\Vert }$ empirically, given a dataset and ML model~\cite{8664630}. $L$ is the maximum. $\mu$ is the minimum.
\end{assumption}

\begin{define}\label{def-divergence}
Under imperfect channels:
\begin{enumerate}
\item Define the distance between the ``global'' model of D-FL at the $i$-th iteration of the $t$-th round, i.e., $ \boldsymbol{\bar{\omega}}_i^{t}$, $i=1,\cdots,I$, and the global optimum, $\boldsymbol{\omega}^{*}$, as
\begin{align}
\Delta \boldsymbol{{\omega}}_i^{t}=\left\Vert  \boldsymbol{\bar{\omega}}_i^{t}-\boldsymbol{\omega}^{*}\right\Vert ^{2}.\label{delta_r}
\end{align}

\item Define the bias between device $n$'s locally aggregated model, $\mathbf{w}_{n,J}^{t}$, and $\boldsymbol{\bar{\omega}}_I^{t}$ in the $t$-th training round as
\begin{align}	\boldsymbol{\varpi}^{t}_{n}=\boldsymbol{\bar{\omega}}_I^{t}-	\mathbf{w}_{n,J}^{t}.\label{Eq:consensus error noise}
\end{align}
\end{enumerate}
Let $\boldsymbol{\Pi}^{t}\in\mathbb{R}^{N\times M}$ collect the biases of all devices in the $t$-th round. From \eqref{wr0},
\begin{align}
\boldsymbol{\Pi}^{t}&=\left[\boldsymbol{\varpi}^{t}_{1},\cdots,\boldsymbol{\varpi}^{t}_{N}\right]^{\dagger}=\left[\boldsymbol{\bar{\omega}}_I^{t},\cdots,\boldsymbol{\bar{\omega}}_I^{t}\right]^{\dagger}-{\mathbf{W}}_{J}^{t}\notag\\&=\frac{\mathbf{1}_N^{\dagger}\mathbf{1}_N}{N}N\mathrm{diag}(\mathbf{p})\boldsymbol{\Omega}_I^{t}-{\mathbf{W}}_{J}^{t}.\label{noise_consensus:c}
\end{align} %

\end{define}

Note that the virtually aggregated global model of D-FL, $ \boldsymbol{\bar{\omega}}_I^{t}$, is different from the global model of C-FL because each round of D-FL starts with the locally aggregated models of individual devices in D-FL, i.e., $\mathbf{w}_{n,J}^{t-1}$ (By contrast, C-FL would start with $\boldsymbol{\bar{\omega}}_I^{t-1}$ across all devices). To prove the convergence of D-FL is, in essence, to prove that $\Delta \boldsymbol{{\omega}}_I^{t}$ (i.e., the distance between $\bar{\boldsymbol{\omega}}_I^t$ and $\boldsymbol{\omega}^*$) decreases over $t$. This is non-trivial under imperfect local aggregations arising from the different initial model parameters $\mathbf{w}_{n,J}^{t-1}, \forall n$ of the devices per round $t$ and communication errors. 
Under \textbf{Assumption \ref{assumption}}, the following one-round convergence upper bound of D-FL is established in imperfect channels.       
\begin{lemma}\label{theo1}
 With random mini-batch sampling and communication errors, the expectation of the distance between the global model of D-FL in the $t$-th round, i.e., $ \boldsymbol{\bar{\omega}}_I^{t}$, and the global optimum of D-FL, i.e., $\boldsymbol{\omega}^*$, is given by
   \begin{subequations}\small
       \begin{align}
    &\mathbb{E}_{\{\xi_{n,i}^{t},{\forall n,i}\},\mathbf{e}}(\Delta\boldsymbol{\omega}_I^{t})\!\!\leq\!\!  \zeta_1\Delta\boldsymbol{{\omega}}_I^{t\!-\!1}\!\!\!+\!\!\zeta_{2}G^2\!\!+\!\!\zeta_3\mathbb{E}_{\mathbf{e}}\!\Big(\!\big\Vert\!{\sum}_{\forall n}\!p_{n}\!\boldsymbol{\varpi}^{t\!-\!1}_{n}\!\big\Vert^{2}\!\Big)\!\!\notag\\&\;+\!\!
\zeta_4\!{\sum}_{\forall n}\!\!p_{n}\!\mathbb{E}_{\mathbf{e}}\!\big(\!\big\Vert\!\boldsymbol{\varpi}_{n}^{t\!-\!1}\!\!\!-\!\!\!{\sum}_{\forall m}\!p_{m}\!\boldsymbol{\varpi}_{m}^{t-\!1}\!\big\Vert^{2}\!\big)\!\!+\!\!\zeta_5\!{\sum}_{\forall n}\!p_{n}\!{\mathbb{E}_{\mathbf{e}}}\big\Vert\mathbf{w}_{n,J}^{t\!-\!1}\!\!-\!\!\mathbf{x}_{n,J}^{t\!-\!1}\!\big\Vert^{2}\notag\\&\;+\!\zeta_6{\sum}_{\forall n}\!\!\eta p_{n}^{2}\!\sigma^2_n\!+\!\zeta_7{\sum}_{\forall n}p_{n}\left\Vert \bar{\boldsymbol{\omega}}_{I}^{t\!-\!1}\!-\!\mathbf{x}_{n,J}^{t\!-\!1}\right\Vert ^{2}\!\!\label{eq:new_expectation}\\
 &\notag\leq\!\!\zeta_{1}\Delta\boldsymbol{{\omega}}_I^{t\!-\!1}\!\!\!+\!\!\zeta_{2}G^{2}\!\!+\!\zeta_{3}\left(\big\Vert\mathbf{p}\mathbb{E}_{\mathbf{e}}(\boldsymbol{\Pi}^{t-1})\big\Vert^{2}\!\!+\!\!\mathbf{p}\mathrm{diag}(\mathbf{p})\mathbb{D}_{\mathbf{e}}(\mathbf{W}_{J}^{t-1})\mathbf{1}_M^{\dagger}\right)\\
 &\notag+\!\zeta_{4}\!\left(\left\Vert \mathrm{diag}(\sqrt{\mathbf{p}}\!\!-\!\!\mathbf{p})\mathbb{E}_{\mathbf{e}}(\boldsymbol{\Pi}^{t\!-\!1})\right\Vert ^{2}\!\!+\!\!(\mathbf{p}\!\!-\!\!\mathbf{p}\mathrm{diag}(\mathbf{p}))\mathbb{D}_{\mathbf{e}}(\mathbf{W}_{J}^{t\!-\!1})\mathbf{1}_M^{\dagger}\right)\\
 &\notag+\zeta_{5}\!\left(\left\Vert \mathrm{diag}(\sqrt{\mathbf{p}})\mathbb{E}_{\mathbf{e}}\!\!\left(\mathbf{W}_{J}^{t\!-\!1}-\mathbf{X}_{J}^{t\!-\!1}\right)\!\right\Vert ^{2}\!\!\!+\!\!\mathbf{p}\mathbb{D}_{\mathbf{e}}(\mathbf{W}_{J}^{t\!-\!1})\mathbf{1}_M^{\dagger}\right)\\
 &+\!\!\zeta_{6}\!\!\sum_{\forall n}\!\!\eta p_{n}^{2}\!\sigma_{n}^{2}\!\!+\!\!\zeta_{7}\Big\Vert \mathrm{diag}(\!\sqrt{\mathbf{p}})\!\Big(\!\frac{\mathbf{1}_N^{\dagger}\mathbf{1}_N}{N}N\mathrm{diag}(\mathbf{p})\boldsymbol{\Omega}_I^{t\!-\!1}\!\!-\!\!{\mathbf{X}}_{J}^{t\!-\!1}\Big)\!\Big\Vert ^{2}\!\!,\!\!\!\label{eq:new_expectationb}
\end{align}
   \end{subequations}%
where $\mathbb{D}_{\mathbf{e}}(\mathbf{W}_{J}^{t\!-\!1}) \in \mathbb{R}^{N\times M}$  collects the element-wise variances of $\mathbf{W}_{n,J}^{t\!-\!1}$, 
$\zeta_1\!=\!(1+\tau_{\epsilon})\left(1-\frac{\mu\eta}{2}\right)^{I}$,
$\zeta_{2}\!=\!(\!2\eta L^{2}\!\!+\!\!L\!\!+\!\!\mu)\!(1\!\!+\!\!\eta)\!\big(\!\frac{(1\!+\!\eta)^{I\!+\!1}\!-\!(1\!+\!\eta)^{2}(1\!-\!\frac{\mu\eta}{2})^{I\!-\!1}}{1\!\!+\!\!\frac{\mu}{2}}\!-\!\frac{\!2\!-\!2(1\!-\!\frac{\mu\eta}{2})^{I\!-\!1}\!}{\mu}\!\big)$,
$\zeta_3\!=\!(1+\frac{1}{\tau_{\epsilon}})(1+\tau_{\eta})\left(1-\frac{\mu\eta}{2}\right)^{I-1}$,
$\zeta_4\! =\! (\!2\eta L^{2}\!\!+\!\!L\!\!+\!\!\mu)(1\!+\!\eta)^{3}\frac{\!(\!1\!+\!\eta\!)^{I\!-\!1}\!\!-\!\!(\!1\!-\!\frac{\mu\eta}{2}\!)^{I\!-\!1}\!}{1\!+\!\frac{\mu}{2}}$,
$\zeta_5\!=\!\frac{\zeta_{3}\eta^{2}\!L^{2}}{\tau_{\eta}}$,
$\zeta_6\!=\!\frac{2+((1+\tau_{\epsilon})\mu\eta-2)(1-\frac{\mu\eta}{2})^{I-1}}{\mu}$, and
$\zeta_7\!=\!(1\!+\!\tau_{\epsilon}\!)\left(1\!\!-\!\!\frac{\mu\eta}{2}\right)^{I-1}\left(\!2\eta^{2}L^{2}\!+\!(L\!+\!\mu)\eta\right).$
 Here, $\tau_{\epsilon}$ is the average PER across the network. $\tau_{\eta}=\eta L/\mu$. 
\end{lemma}
\begin{proof}
    See \textbf{Appendix \ref{appendix:theorem proof}}.
\end{proof}

As revealed in \textbf{Lemma \ref{theo1}}, when communication errors are negligible (i.e., $\tau_{\epsilon}\rightarrow0$), $\zeta_1\rightarrow \left(1-\frac{\mu\eta}{2}\right)^{I}$. By setting $I=1$ local training iteration per round, the one-round convergence upper bound in \eqref{eq:new_expectation} can conform to, and is even tighter than, the one developed in \cite{pmlr-v119-koloskova20a} under perfect channels and $I=1$ local training iteration per round.

We proceed to define four matrices, $\mathbf{M}_1^{t}$, $\mathbf{M}_2^{t}$, $\mathbf{M}_3^{t}$, and $\mathbf{M}_4^{t}$, per round $t$. With the matrices, we can present (20b) in a more concise way with the impact of imperfect channels and local aggregations decoupled from the (initial) model parameters of the round, facilitating comprehension of the impact.

\begin{lemma}\label{lemma1}
In the $t$-th round, the expectation of the bias $\boldsymbol{\Pi}^{t}$~is
\begin{align}
    \mathbb{E}_{\mathbf{e}}(\boldsymbol{\Pi}^{t})=\mathbf{M}_1^{t}\mathbf{W}_{0}^{t},\label{Expectation_consensus}
\end{align}
where $\mathbf{M}_1^{t}=\frac{{\mathbf{1}_N^{\dagger}\mathbf{1}_N}}{N}-(\mathbf{C}^{t}\!\circ \!\mathbf{T}^{t})^J$.
\end{lemma} 
\begin{proof}
See \textbf{Appendix \ref{proof of lemma1}}.
\end{proof}

\begin{lemma}\label{lemma_variance}
In the $J$-th local aggregation of the $t$-th round, the variances of the locally aggregated models can be written in a matrix form, as given by
\begin{subequations}
\begin{align}
    \mathbb{D}_{\mathbf{e}}(\mathbf{W}_J^{t})&=\mathbf{M}_2^{t} {(\mathbf{W}_{0}^{t}\!\circ\!\mathbf{W}_{0}^{t})}, \label{D_J}\\
    \text{where }\quad \mathbf{M}_2^{t}&={\sum}_{k=1}^{J}(\mathbf{A}_{1}^{t})^{k-1}\!\mathbf{A}_{2}^{t}\mathbf{A}_{3}^{t}(J-k);\label{define_M2}\\
    \mathbf{A}_{1}^{t}&=\!\mathbf{C}^{t}\circ\mathbf{C}^{t}\circ\mathbf{T}^{t};\\
    \mathbf{A}_{2}^{t}&=\!\mathbf{C}^{t}\circ\mathbf{C}^{t}\circ\mathbf{T}^{t}\circ({\mathbf{1}_{N}^{\dagger}\!\mathbf{1}_{N}}\!-\!\mathbf{T}^{t});\\
    \mathbf{A}_{3}^{t}(j)&=(\mathbf{C}^{t}\circ\mathbf{T}^{t})^{j\!-\!1}\circ(\mathbf{C}^{t}\circ\mathbf{T}^{t})^{j\!-\!1}.
\end{align}
\end{subequations}
   \end{lemma}	
   \begin{proof}
      See \textbf{Appendix \ref{proof_of_lemma_variance}}.
   \end{proof}

Likewise, we rewrite $\mathbb{E}_{\mathbf{e}}(\mathbf{W}_{J}^{t}-\mathbf{X}_{J}^{t})$, the expectation of the difference between the locally aggregated models under imperfect and perfect channels, ${\mathbf{W}}_{J}^{t}$ and ${\mathbf{X}}_{J}^{t}$, as $$ \mathbb{E}_{\mathbf{e}}(\mathbf{W}_{J}^{t}-\mathbf{X}_{J}^{t})=\mathbf{M}_3^{t}\mathbf{W}_{0}^{t},$$ where $\mathbf{M}_3^{t}=(\mathbf{C}^{t})^J-(\mathbf{C}^{t}\!\circ \!\mathbf{T}^{t})^J$ when $\mathbf{C}^t$ is designed using \eqref{C_r_definition}, or $\mathbf{M}_3^{t}=\mathbf{0}_{N\times M}$ when $\mathbf{C}^t$ is designed using \eqref{CoT_r_definition}.

Moreover, we define $\mathbf{M}_4^{t}$ such that $$\frac{\mathbf{1}_N^{\dagger}\mathbf{1}_N}{N}N\mathrm{diag}(\mathbf{p})\boldsymbol{\Omega}_I^{t}-{\mathbf{X}}_{J}^{t}=\mathbf{M}_4^{t}\mathbf{W}_{0}^{t}.$$ Then, $\mathbf{M}_4^{t}=\!\frac{\mathbf{1}_N^{\dagger}\mathbf{1}_N}{N}\!-\!(\mathbf{C}^{t}\!)^J$ under \eqref{C_r_definition}, or $\mathbf{M}_4^{t}=\!\frac{\mathbf{1}_N^{\dagger}\mathbf{1}_N}{N}\!-\!(\mathbf{C}^{t}\!\circ \!\mathbf{T}^{t})^J$ under \eqref{CoT_r_definition}.

By using $\mathbf{M}_1^{t}$, $\mathbf{M}_2^{t}$, $\mathbf{M}_3^{t}$, $\mathbf{M}_4^{t}$, and $\mathbf{W}_0^t$, $\forall t$, the one-round convergence upper bound established in \textbf{Lemma \ref{theo1}} can be updated with the impact of imperfect local aggregations captured in $\mathbf{M}_1^{t}$, $\mathbf{M}_2^{t}$, $\mathbf{M}_3^{t}$, and $\mathbf{M}_4^{t}$ decoupled from the initial local models $\mathbf{W}_0^t$. This facilitates developing the accumulated convergence upper bound of D-FL.

\begin{theo}\label{theo2}
For D-FL under imperfect channels, a one-round convergence upper bound is given by
{\small
\begin{align}
&{\mathbb{E}_{\{\xi_{n,i}^{t},{\forall n,i}\},\mathbf{e}}}(\Delta \boldsymbol{{\omega}}_I^{t})\leq \! \zeta_1\Delta\boldsymbol{{\omega}}_I^{t-1}\!+\!\zeta_{2}G^2\!+\!\zeta_6{\sum}_{\forall n}\eta p_{n}^{2}\!\sigma^2_n\notag\\
&\!+\!\zeta_5 \Phi(\mathbf{M}_1^{t\!-\!1}\!,\mathbf{M}_2^{t\!-\!1}\!,\mathbf{M}_3^{t\!-\!1},\mathbf{M}_4^{t\!-\!1})p_{\max}\big\Vert\mathbf{W}_{0}^{t\!-\!1}\!\big\Vert^{2};\label{theorem_expectation2}%
\end{align}%
}%
and an accumulated convergence upper bound of $t$ rounds is 
{\small
\begin{align}
  \notag & {\mathbb{E}_{\{\xi_{n,i}^{t'},{\forall n,i,t'<t}\},\mathbf{e}}}(\Delta\boldsymbol{{\omega}}_I^{t})\!\!\leq \!\!(\zeta_1)^{\!t}\!\Delta\boldsymbol{\omega}^{0}\!\!+\!\!\frac{(\zeta_1)^{t\!+\!1}\!\!-\!\!1}{\zeta_{1}\!-\!\!1}\Big(\!\zeta_{2}G^{2}\!\!+\!\!\zeta_6\!\sum_{\forall n}\!\!\eta p_{n}^{2}\!\sigma_{n}^{2}\Big)
\\&\;+\!\!{\sum}_{t'\!=\!1}^{t}\!(\zeta_1)^{t\!-\!t'}\!\zeta_5\Phi(\mathbf{M}_1^{t'\!-\!1}\!\!,\!\mathbf{M}_2^{t'\!-\!1}\!\!,\!\mathbf{M}_3^{t'\!-\!1}\!\!,\!\mathbf{M}_4^{t'\!-\!1}\!) p_{\max}\!\Vert\!\mathbf{W}_{0}^{t'\!-\!1}\!\Vert^{2}\!\!,\label{eq:convergence_t_to_0}
\end{align}}%
where $p_{\max}\!=\!\Vert \mathrm{diag}(\mathbf{p})\Vert\!=\!\underset{\forall n}{\max}\; p_n$ is the maximum weighting coefficient; and $\forall t=1,2,\cdots$,
\begin{align}\notag
&\!\!\!\!\Phi(\mathbf{M}_1^{t}\!,\mathbf{M}_2^{t}\!,\mathbf{M}_3^{t}\!,\mathbf{M}_4^{t})\!\!=\!\!\frac{\zeta_{3}p_{\max} }{\zeta_5}\Vert\!\mathbf{1}_{N}\mathbf{M}_{1}^{t}\!\Vert^{2}\!\!+\!\!\Vert\mathbf{M}_{3}^{t}\!\Vert^{2}\!\!+\!\!\frac{\zeta_7}{\zeta_5}\!\Vert\mathbf{M}_{4}^{t}\!\Vert^{2}\\&\!\!+\!\!\frac{\zeta_4}{\zeta_5}\!\left(1\!\!-\!\!\sqrt{p_{\max}}\!\right)^{2}\!\Vert\mathbf{M}_{1}^{t}\!\Vert^{2}\!\!+\!\!\Big(\!\frac{\zeta_{3}p_{\max}\!\!+\!\!\zeta_{4}\!(\!1\!\!-\!\!p_{\max}\!)\!}{\zeta_5}\!\!+\!\!1\Big)\!\big\Vert\mathbf{M}_{2}^{t}\!\big\Vert.\label{problem1}%
\end{align}%
\end{theo}%
\begin{proof}%
Based on the definitions of $\mathbf{X}_{n,J}^{t}$  in \eqref{local_c_parameter_i}, $\mathbb{E}_{\mathbf{e}}(\boldsymbol{\Pi}^{t})$ in \eqref{Expectation_consensus}, and $\mathbb{D}_{\mathbf{e}}(\mathbf{W}_{n,J}^{t})$ in \eqref{D_J}, and the homogeneity of the 2-norms of matrices and vectors, we have
\begin{subequations}\label{eq:new_expectationb_upperbound}
    \small%
\begin{align}
\notag 	& \!\!\!\!\!\!  \big\Vert\mathbf{p}\mathbb{E}_{\mathbf{e}}(\boldsymbol{\Pi}^{t})\big\Vert^{2}+\mathbf{p} \mathrm{diag}(\mathbf{p})\mathbb{D}_{\mathbf{e}}\!(\mathbf{W}_{J}^{t})\mathbf{1}_M^{\dagger}
\\& \!\!\!\! \;	\leq \left(\big\Vert\mathbf{1}_{N}\mathbf{M}_{1}^{t}\big\Vert^{2}+\left\Vert \mathbf{M}_{2}^{t}\right\Vert \right)p_{\max}^{2}\left\Vert \mathbf{W}_{0}^{t}\right\Vert ^{2};\label{var2}
\\\notag&\!\!\!\!\!\left\Vert \mathrm{diag}(\!\sqrt{\mathbf{p}}\!\!-\!\!\mathbf{p}\!)\mathbb{E}_{\mathbf{e}}(\boldsymbol{\Pi}^{t})\right\Vert ^{2}\!\!+\!\!(\mathbf{p}\!\!-\!\!\mathbf{p}\mathrm{diag}(\mathbf{p}))\mathbb{D}_{\mathbf{e}}(\mathbf{W}_{J}^{t})\!\mathbf{1}_M^{\dagger}
\\ &\!\!\!\!\;\leq\!\!\left(\!\left\Vert \mathbf{M}_{1}^{t}\right\Vert ^{2}\!\!(1\!\!-\!\!\sqrt{p_{\max}})^{2}\!\!+\!\!(1\!-\!p_{\max})\Vert\mathbf{M}_{2}^{t}\Vert\!\right)p_{\max}\!\Vert\mathbf{W}_{0}^{t}\Vert^{2};\label{mean}
\\&\notag\!\!\!\!\!\left\Vert \mathrm{diag}(\sqrt{\mathbf{p}})\mathbb{E}_{\mathbf{e}}\!\!\left(\mathbf{W}_{J}^{t}\!\!-\!\!\mathbf{X}_{J}^{t}\right)\!\right\Vert ^{2} \! \!+\!\! \mathbf{p}\mathbb{D}_{\mathbf{e}}\!(\mathbf{W}_{J}^{t})\mathbf{1}_M^{\dagger}
\\&\!\!\!\!\;\leq\! \!\left(\left\Vert \mathbf{M}_{3}^{t}\right\Vert ^{2}\!+\left\Vert \mathbf{M}_{2}^{t}\right\Vert \right)p_{\max}\left\Vert \mathbf{W}_{0}^{t}\right\Vert ^{2};\label{mean_sixth_part}
\\&\!\!\!\!\! {\sum}_{\forall n}\!p_{n}\!  \Vert\bar{\boldsymbol{\omega}}_{I}^{t}\!-\!\mathbf{x}_{n,J}^{t}\Vert^{2}\!\leq \left\Vert \mathbf{M}_{4}^{t}\right\Vert^2  p_{\max}\left\Vert \mathbf{W}_{0}^{t}\right\Vert ^{2},\label{mean_fifth_part}%
	 \end{align}%
\end{subequations}
By substituting~\eqref{eq:new_expectationb_upperbound} into \eqref{eq:new_expectationb}, we obtain \eqref{theorem_expectation2}. By mathematical induction, we further obtain \eqref{eq:convergence_t_to_0}.
\end{proof}

As revealed in \textbf{Theorem \ref{theo2}}, the impacts of the network topology and the consensus matrix on the convergence of D-FL are captured in $\Phi(\mathbf{M}_{1}^{t},\mathbf{M}_{2}^{t},\mathbf{M}_{3}^{t},\mathbf{M}_{4}^{t})$ in \eqref{theorem_expectation2}. Since the network topology can change over rounds, $\Phi(\mathbf{M}_{1}^{t},\mathbf{M}_{2}^{t},\mathbf{M}_{3}^{t},\mathbf{M}_{4}^{t}),\forall t$ and subsequently the one-round convergence upper bound changes over time. It is also revealed that communication quality plays a non-negligible role in the convergence of D-FL. Specifically, $\zeta_1=(1+\tau_{\epsilon})\left(1-\frac{\mu\eta}{2}\right)^{I}$ is positively correlated with the average PER,  $\tau_{\epsilon}$, in \eqref{theorem_expectation2}. When communication errors are large, it is possible that $\zeta_1>1$ and hence $(\zeta_1)^{t}\Delta\boldsymbol{\omega}^{0}\rightarrow\infty$ as $t\rightarrow\infty$; in other words, D-FL diverges. 

We note that \textbf{Assumption \ref{assumption}-3)} is widely used for the convenience of theoretical analysis~\cite{9264742,boyd2004convex,9770266,pmlr-v130-ruan21a,stich2018local}. It is adopted in this paper to analyze the impact of communication errors and D-FL architecture on the system. From \textbf{Assumptions \ref{assumption}-2)} and \textbf{\ref{assumption}-3)}, it inherently follows that
\begin{align}
\left\|\boldsymbol{\omega}_{n, i}^{t}-\boldsymbol{\omega}_{n}^{\star}\right\|^{2} \leq \frac{G^2}{\mu^{2}},\,
\forall i,t,
\label{eq:wni_domain}
\end{align}
which leads to the following upper bound on $\Delta\boldsymbol{\omega}_{I}^t$:
\begin{align}
\notag \Delta\boldsymbol{\omega}_{I}^t & =\!\!\left\Vert \! \boldsymbol{\bar{\omega}}_I^{t}\!\!-\!\!\boldsymbol{\omega}^{*}\!\right\Vert ^{2}\!\!=\!\!\big\Vert \! {\sum}_{\forall n}p_{n}\boldsymbol{\omega}_{n,I}^{t}\!\!-\!\!\boldsymbol{\omega}^{*}\big\Vert ^{2} \!\!\overset{(a)}{\leq}\!\! {\sum}_{\forall n}p_{n}\!\!\left\Vert \! \boldsymbol{\omega}_{n,I}^{t}\!\!-\!\!\boldsymbol{\omega}^{*}\!\right\Vert ^{2}\\\notag
&\leq{\sum}_{\forall n}p_{n}\left[\left\Vert  \boldsymbol{\omega}_{n,I}^{t}-\boldsymbol{\omega}_n^{*}\right\Vert +\left\Vert  \boldsymbol{\omega}^{*}-\boldsymbol{\omega}_n^{*}\right\Vert \right]^{2}\\&\overset{(b)}{\leq}{\sum}_{\forall n}p_{n}(1+\kappa)^2\frac{G^2}{\mu^2}=(1+\kappa)^2\frac{G^2}{\mu^2},\label{loose_bound}
\end{align} 
where $(a)$ is based on Jensen's inequality; and $(b)$ is based on \eqref{eq:wni_domain}; since $\left\Vert  \boldsymbol{\omega}^{*}-\boldsymbol{\omega}_n^{*}\right\Vert \ll\frac{G}{\mu},\,\forall n$, there exists $\kappa \in (0,1]$ such that $\left\Vert  \boldsymbol{\omega}^{*}-\boldsymbol{\omega}_n^{*}\right\Vert \leq \kappa\frac{G}{\mu}, \forall n$. 

When $t=0$, we readily have $ \Delta\boldsymbol{\omega}^0 \leq(1+\kappa)^2\frac{G^2}{\mu^2}$; i.e., the upper bound in \eqref{eq:convergence_t_to_0} is tighter than the one in \eqref{loose_bound}. When $ t\geq 1$, \eqref{eq:convergence_t_to_0} is tighter as long as $G^{2} \geq \frac{\frac{(\zeta_1)^{t+1}-1}{\zeta_{1}-1}\big(\zeta_6\sum_{\forall n}\eta p_{n}^{2}\sigma_{n}^{2}\big)+C}{\frac{(1+\kappa)^2(1-(\zeta_{1})^{t})}{\mu^{2}}-\frac{(\zeta_{1})^{t+1}-1}{\zeta_{1}-1}\zeta_{2}}$ and $\frac{(1+\kappa)^2(1-(\zeta_{1})^{t})}{\mu^{2}}-\frac{(\zeta_{1})^{t+1}-1}{\zeta_{1}-1}\zeta_{2}>0$, where, for conciseness we have defined $C\triangleq{\sum}_{t'=1}^{t}(\zeta_1)^{t\!-\!t'}\zeta_5\Phi(\mathbf{M}_1^{t'\!-\!1},\mathbf{M}_2^{t'\!-\!1},\mathbf{M}_3^{t'\!-\!1},\mathbf{M}_4^{t'\!-\!1}) p_{\max}\Vert\mathbf{W}_{0}^{t'-1}\Vert^{2}$. Specifically, under these conditions, it readily follows that 
{\small\begin{align}\notag&\frac{(\zeta_1)^{t\!+\!1}\!\!-\!\!1}{\zeta_{1}\!\!-\!\!1}\big(\zeta_6\!\!\sum_{\forall n}\eta p_{n}^{2}\sigma_{n}^{2}\big)
\!\!+\!\!C\!\!\leq\!\! \big(\frac{(1\!\!+\!\!\kappa)^2(1\!\!-\!\!(\zeta_{1})^{t})}{\mu^{2}}\!\!-\!\!\frac{(\zeta_{1})^{t\!+\!1}\!\!-\!\!1}{\zeta_{1}\!\!-\!\!1}\zeta_{2}\big)G^{2},\notag%
\end{align}%
}%
which, with reorganization and $ \Delta\boldsymbol{\omega}^0 \leq(1+\kappa)^2\frac{G^2}{\mu^2}$, leads to
{\small\begin{align}
(1\!\!+\!\!\kappa)^2\frac{G^2}{\mu^2} &\!\!\geq\!\! (\zeta_1)^{t}(1\!\!+\!\!\kappa)^2\frac{G^2}{\mu^2}\!\!+\!\!\frac{(\zeta_1)^{t\!+\!1}\!\!-\!\!1}{\zeta_{1}\!\!-\!\!1}\big(\zeta_{2}G^{2}\!\!+\!\!\zeta_6\sum_{\forall n}\eta p_{n}^{2}\sigma_{n}^{2}\big)
\!\!+\!\!C \notag\\&\notag{\geq}(\zeta_1)^{t}\Delta\boldsymbol{\omega}^{0}\!\!+\!\!\frac{(\zeta_1)^{t+1}\!\!-\!\!1}{\zeta_{1}\!\!-\!\!1}\big(\zeta_{2}G^{2}\!\!+\!\!\zeta_6\sum_{\forall n}\eta p_{n}^{2}\sigma_{n}^{2}\big)
\!\!+\!\!C.%
\end{align}%
}%
In this sense, when $\frac{(1+\kappa)^2(1-(\zeta_{1})^{t})}{\mu^{2}}-\frac{(\zeta_{1})^{t+1}-1}{\zeta_{1}-1}\zeta_{2}>0$ (and $\zeta_1<1$), the existence of $G<+\infty$ or, in other words, the validity of \textbf{Assumption \ref{assumption}-3)} is confirmed; $G$ can take any value no smaller than $\underset{\forall t>1}{\max}\sqrt{\frac{\frac{(\zeta_1)^{t+1}-1}{\zeta_{1}-1}\big(\zeta_6\sum_{\forall n}\eta p_{n}^{2}\sigma_{n}^{2}\big)+C}{\frac{(1+\kappa)^2(1-(\zeta_{1})^{t})}{\mu^{2}}-\frac{(\zeta_{1})^{t+1}-1}{\zeta_{1}-1}\zeta_{2}}}$. When the channel is error-free (i.e., $\tau_{\epsilon}\rightarrow0$), there always exist adequate $\eta$ and $I$, e.g., $0<\eta\leq\frac{1}{\mu}$ and $I=1$, that can satisfy the conditions $\frac{(1+\kappa)^2(1-(\zeta_{1})^{t})}{\mu^{2}}-\frac{(\zeta_{1})^{t+1}-1}{\zeta_{1}-1}\zeta_{2}>0$ and $\zeta_1<1$. When the channel is imperfect, it is possible that no $\eta$ or $I$ can satisfy the conditions; i.e., \eqref{eq:convergence_t_to_0} is not as tight as \eqref{loose_bound}. The convergence of the D-FL is disrupted.

\begin{theo}\label{theo3}
In the general case with a changing topology over rounds, when $\zeta_1 < 1$ and $\frac{(1+\kappa)^2(1-(\zeta_{1})^{t})}{\mu^{2}}-\frac{(\zeta_{1})^{t+1}-1}{\zeta_{1}-1}\zeta_{2}>0$, a convergence upper bound of D-FL is given by
{\small \begin{align}
  \notag 
  &\!\!{\mathbb{E}_{\{\xi_{n,i}^{t'},{\forall n,i, t'<t}\},\mathbf{e}}}(\Delta \boldsymbol{\omega}_I^t)|_{t\rightarrow\infty}\leq\frac{1}{1-\zeta_{1}}\left(\!\zeta_{2}G^{2}\!+\!\zeta_6{\sum}_{\forall n}\eta p_{n}^{2}\!\sigma_{n}^{2}\right)\\&
  +\!\!\frac{\zeta_5 \lambda_{\max}p_{\max}}{1\!\!-\!\!\zeta_{1}}\underset{t'\in\mathbb{N}}{\max}\!\Big\{\Phi\big(\mathbf{M}_{1}\!^{t'}\!,\mathbf{M}_{2}^{t'},\mathbf{M}_{3}^{t'},\mathbf{M}_{4}^{t'}\!\big)\!\Big\},\!\!\!
  \label{eq: general_bound}%
\end{align}%
}%   
where $\lambda_{\max}=\underset{t\in \mathbb{N}}{\max}\;\{\Vert\mathbf{W}_{0}^{t}\Vert^{2}\}$. 
\end{theo}
\begin{proof}
\eqref{eq: general_bound} is obtained by substituting $\Vert\mathbf{W}_{0}^{t}\Vert^{2}\leq\lambda_{\max},\forall t$ and $\Phi(\mathbf{M}_{1}^{t},\mathbf{M}_{2}^{t},\mathbf{M}_{3}^{t},\mathbf{M}_{4}^{t})\leq\underset{t'\in\mathbb{N}}{\max}\{\Phi(\mathbf{M}_{1}^{t'},\mathbf{M}_{2}^{t'},\mathbf{M}_{3}^{t'},\mathbf{M}_{4}^{t'}\!)\!\},\forall t$ into~\eqref{eq:convergence_t_to_0}, and applying $\underset{t\rightarrow\infty}{\lim}\!(\zeta_1)^t\!\rightarrow\!0$.
\vspace{-1mm }\end{proof}

\noindent \textbf{Corollary 1.} 
\textit{As $t\rightarrow \infty$, the distance between the locally aggregated model $ \mathbf{w}_{n,J}^{t}$ of device $n$ and $\boldsymbol{\omega}^{*}$, i.e., $\mathbb{E}_{\{\xi_{n,i}^{t'},{\forall n,i,t'<t}\},\mathbf{e}}(\left\Vert \mathbf{w}_{n,J}^{t}-\boldsymbol{\omega}^{*}\right\Vert ^{2})\big|_{t\rightarrow \infty}$, is upper bounded.}

\begin{proof}
At any round $t$, we have
\begin{align}
   \Vert \mathbf{w}_{n,J}^{t}\!\!-\!\!\boldsymbol{\omega}^{*}\Vert ^{2}
    \!\!\leq
    &2\Vert \boldsymbol{\bar{\omega}}_I^{t}\!\!-\!\!\mathbf{w}_{n,J}^{t}\Vert ^{2}\!\!+\!\!2\Vert \boldsymbol{\bar{\omega}}_I^{t}\!\!-\!\!\boldsymbol{\omega}^{*}\!\Vert ^{2}\!\!=\!\!2\Vert \boldsymbol{\varpi}_{n}^{t}\Vert ^{2}\!\!+\!\!2\Delta\boldsymbol{\omega}_I^{t},\notag
\end{align}
which leads to 
\begin{align}
\small
\notag   &\!\!\!\!\mathbb{E}_{\{\xi_{n,i}^{t'},{\forall n,i,t'<t}\},\mathbf{e}}({\sum}_{\forall n}p_{n}\left\Vert \mathbf{w}_{n,J}^{t}-\boldsymbol{\omega}^{*}\right\Vert ^{2})|_{t\rightarrow\infty}\\
&\!\!\!\leq\!\!2\!\left(\!\mathbb{E}_{\{\xi_{n,i}^{t'},{\forall n,i,t'<t}\},\mathbf{e}}(\Delta\boldsymbol{\omega}_I^{t})\!\!+\!\!\mathbb{E}_{\mathbf{e}}({\sum}_{\forall n}\!\!p_{n}\!\Vert \boldsymbol{\varpi}_{n}^{t}\Vert ^{2})\!\right)\!|_{t\rightarrow\infty}.\!\label{eq: 29}%
\end{align}%

With $\mathbb{E}_{\mathbf{e}}(\sum_{\forall n}p_{n}\Vert \boldsymbol{\varpi}_{n}^{t}\Vert ^{2})\leq\left(\Vert\mathbf{M}_{1}^{t}\Vert^{2}+\Vert\mathbf{M}_{2}^{t}\Vert\right)p_{\max}\Vert \mathbf{W}_{0}^{t}\Vert ^{2}$ based on \eqref{mean_sixth_part} and \eqref{eq:new_expectation_fourth_part_c}, and \textbf{Theorem 2}, the RHS of \eqref{eq: 29} is upper bounded and so is $\mathbb{E}_{\{\xi_{n,i}^{t'},{\forall n,i,t'<t}\},\mathbf{e}}(\sum_{\forall n}p_{n}\left\Vert \mathbf{w}_{n,J}^{t}-\boldsymbol{\omega}^{*}\right\Vert ^{2})\big|_{t\rightarrow \infty}$.
\end{proof}

In \textbf{Theorems \ref{theo2}} and \textbf{\ref{theo3}}, only $\Vert\mathbf{W}_{0}^{t}\Vert^{2},\forall t$ is assumed to be upper bounded. We do not assume the one-round upper bound \eqref{theorem_expectation2} or $\Delta \boldsymbol{\omega}_I^{t-1}$ in \eqref{theorem_expectation2} is finite, although they are since the overall convergence upper bound \eqref{eq: general_bound} exists. Hence, \eqref{eq: general_bound} is established by adequately identifying the requirement for communication errors to ensure convergence, i.e., $\zeta_1<1$  and $\frac{(1+\kappa)^2(1-(\zeta_{1})^{t})}{\mu^{2}}-\frac{(\zeta_{1})^{t+1}-1}{\zeta_{1}-1}\zeta_{2}>0$. While $\underset{t\in\mathbb{N}}{\max}\{\Phi(\mathbf{M}_{1}^t,\mathbf{M}_{2}^t,\mathbf{M}_{3}^t,\mathbf{M}_{4}^t)\}$ can be minimized to minimize the convergence upper bound in \eqref{eq: general_bound}, one can minimize $\Phi\left(\mathbf{M}_1^t,\mathbf{M}_2^t,\mathbf{M}_3^t,\mathbf{M}_4^t\right)$ per round $t$ to further tighten bound, since ${\sum}_{t'=1}^{t}(\zeta_1)^{t-t'}\Phi\big(\mathbf{M}_1^{t'},\mathbf{M}_2^{t'},\mathbf{M}_3^{t'},\mathbf{M}_4^{t'}\big)\leq{\sum}_{t'=1}^{t}(\zeta_1)^{t-t'}\underset{t'\in\mathbb{N}}{\max}\big\{\Phi(\mathbf{M}_{1}^{t'},\mathbf{M}_{2}^{t'},\mathbf{M}_{3}^{t'},\mathbf{M}_{4}^{t'}\!)\!\big\}\leq\frac{1}{1-\zeta_1}\underset{t'\in\mathbb{N}}{\max}\big\{\Phi(\mathbf{M}_{1}^{t'},\mathbf{M}_{2}^{t'},\mathbf{M}_{3}^{t'},\mathbf{M}_{4}^{t'}\!)\!\big\}$.

\section{Optimal Aggregation Schedule of D-FL}\label{section:problem formulation}
In this section, we accelerate the convergence of D-FL under imperfect channels by minimizing the convergence upper bound. According to \textbf{Theorem~\ref{theo3}}, minimizing $\Phi\left(\mathbf{M}_1^t,\mathbf{M}_2^t,\mathbf{M}_3^t,\mathbf{M}_4^t\right)$ per $t$ contributes to reducing the convergence upper bound.
For this reason, we formulate the following problem per $t$: 
\begin{subequations}\label{P1}\small
\begin{align}	
\!\!\!\!\notag	\textbf{P1}:\, &\underset{\mathbf{M}_1^t,\mathbf{M}_2^t,\mathbf{M}_3^t,\mathbf{M}_4^t}{\min}\, \Phi\left(\mathbf{M}_1^t,\mathbf{M}_2^t,\mathbf{M}_3^t,\mathbf{M}_4^t\right)\\
\text{ s.t.} \; \;   
&\mathbf{M}_1^t=\frac{{\mathbf{1}_N^{\dagger}\mathbf{1}_N}}{N}-(\mathbf{C}^{t}\!\circ \!\mathbf{T}^{t})^J, \label{eq:40a}\\
&\mathbf{M}_2^t={\sum}_{k=1}^{J}(\mathbf{A}_{1}^{t})^{k-1}\!\mathbf{A}_{2}^{t}\mathbf{A}_{3}^{t}(J-k),\label{eq:40b}\\
&\mathbf{M}_3^{t}=\! \!
    \begin{cases}
\! \!(\mathbf{C}^{t})^J\!\!-\!\!(\mathbf{C}^{t}\!\!\circ\! \!\mathbf{T}^{t})^J, \!&\text{when \eqref{C_r_definition} is used;}\\
\mathbf{0}_{N\times M},&\text{when \eqref{CoT_r_definition} is used;}
    \end{cases}\label{eq:40c}\\
\label{eq:40d}&
 \mathbf{M}_4^{t}=\! \!
    \begin{cases}
  \frac{\mathbf{1}_N^{\dagger}\mathbf{1}_N}{N}\!-\!(\mathbf{C}^{t}\!)^J, \!&\text{when \eqref{C_r_definition} is used;}\\
\frac{\mathbf{1}_N^{\dagger}\mathbf{1}_N}{N}\!-\!(\mathbf{C}^{t}\circ\mathbf{T}^t\!)^J,&\text{when \eqref{CoT_r_definition} is used,}
    \end{cases}
	\end{align}%
	\end{subequations} 
For conciseness, we suppress the superscript ``$^t$'' in $\mathbf{M}_1^{t}$, $\mathbf{M}_2^{t}$, $\mathbf{M}_3^{t}$, and $\mathbf{M}_4^{t}$ in the following discussion.
	
Note that $\mathbf{M}_1$, $\mathbf{M}_2$, $\mathbf{M}_3$ and $\mathbf{M}_4$ depend on $\mathbf{C}^{t}$, $\mathbf{T}^{t}$, and $J$, where $\mathbf{C}^{t}$ is designed in prior using \eqref{C_r_definition} or \eqref{CoT_r_definition}, and $\mathbf{T}^{t}$ describes the channel conditions. 
With fewer local aggregations (or fewer time slots for model sharing), the model parameter variance among the devices can be larger, leading to higher variations in model accuracy. As more time slots are used (i.e., $J$ is larger), the transmission errors can accumulate, possibly deteriorating the convergence. 
To this end, we reformulate Problem \textbf{P1} to optimize $J$, as follows. 
 \begin{subequations}
 \begin{align}
\textbf{P2}:  \;  \underset{J\in \mathbb{N}_{+}}{\min}\,\Phi(\mathbf{M}_1,\mathbf{M}_2,\mathbf{M}_3,\mathbf{M}_4)
\; \text{ s.t.} \; \text{
\eqref{eq:40a} -- \eqref{eq:40d}.
}\notag
 \end{align}%
 \end{subequations}

We note that $\Phi(\mathbf{M}_1,\mathbf{M}_2,\mathbf{M}_3,\mathbf{M}_4)$, or $\Phi$ for conciseness, is non-convex with respect to $J$. We analyze the upper and lower bound of $\Phi$ regarding $J$ to narrow the search range of~$J$. To do this, we first establish the following lemma. 
\begin{lemma}\label{lemma 5}
     The upper and lower bounds of $\left\Vert \mathbf{M}_1\right\Vert$, $\left\Vert \mathbf{M}_2\right\Vert$, $\left\Vert \mathbf{M}_3\right\Vert$, $\left\Vert \mathbf{M}_4\right\Vert$, and $\left\Vert \mathbf{1}_N\mathbf{M}_1\right\Vert$ are given as follows.
\begin{itemize}
      \item If $\mathbf{T}^t$ is unknown and $\mathbf{C}^t$ is designed using \eqref{C_r_definition}:
        \begin{subequations}\label{bounds_C_r_definition}
            \begin{align}\label{up-low-of-Ma}
          &  1-\beta_1^J\leq\left\Vert \mathbf{M}_1\right\Vert\leq \frac{\beta_{2}^{J}-\beta_{1}^{J}}{\beta_{2}-\beta_{1}}\beta_{3}+ \beta_{4}^{J};\\
            \label{up-low-of-1Ma}
	& \sqrt{N}(  1-\beta_{1}^{J})\leq\Vert\mathbf{1}_N\mathbf{M}_1\Vert\leq \sqrt{N}(1-\beta_6^{J});\\
 \label{up-low-of-M_2}
   & \frac{\beta_{7}^{2J\!+\!1}(1\!-\!\beta_{8})\beta_{2}^{2J\!+\!2}}{N}\!\leq\!\big\Vert\mathbf{M}_2\big\Vert\leq\frac{\beta_{2}^{J}-\beta_{1}^{J}}{\beta_{2}-\beta_{1}}\beta_{3}\beta_1^J;\\
   & \beta_{2}^{J}-\beta_{1}^{J}\leq\big\Vert\mathbf{M}_3\big\Vert\leq \frac{\beta_{2}^{J}-\beta_{1}^{J}}{\beta_{2}-\beta_{1}}\beta_{3};\label{up-low-of-M3}\\
 \label{up-low-of-M4a}&\big\Vert\mathbf{M}_4\big\Vert= \beta_{4}^{J},\\\label{7_beta_range}
 &0\leq\beta_1,\beta_3,\beta_4,\beta_6 ,\beta_{7},\beta_{8}\leq1,\beta_2=1.
            \end{align}
        \end{subequations}

\item If $\mathbf{T}^t$ is known \textit{a-priori} and $\mathbf{C}^t$ is designed using~\eqref{CoT_r_definition}:
\begin{subequations}\label{bounds_CoT_r_definition}
\begin{align}
&\left\Vert \mathbf{M}_1\right\Vert=\beta_{5}^J;\label{up-low-of-Mb}\\
&\left\Vert \mathbf{1}_N\mathbf{M}_1\right\Vert=0;\label{up-low-of-1Mb}\\
\label{up-low-of-M_2_16}
   & \frac{\beta_{7}^{2J\!+\!1}(1\!-\!\beta_{8})\beta_{2}^{2J\!+\!2}}{N}\!\leq\!\big\Vert\mathbf{M}_2\big\Vert\leq\frac{\beta_{2}^{J}-\beta_{1}^{J}}{\beta_{2}-\beta_{1}}\beta_{3}\beta_1^J;\\
   &\big\Vert\mathbf{M}_3\big\Vert=0;\label{up-low-of-M3_16}\\&\big\Vert\mathbf{M}_4\big\Vert=\beta_{5}^J,\label{up-low-of-M4b}\\&
   0\leq\beta_3,\beta_5 ,\beta_{7},\beta_{8}\leq1, \beta_1=1, \beta_2\geq 1.\label{16_beta_range}%
\end{align}%
\end{subequations}%
\end{itemize}%
Here, $\beta_{1}=\left\Vert\mathbf{C}^{t}\!\circ\!\mathbf{T}^{t}\right\Vert$, $\beta_{2}=\left\Vert\mathbf{C}^{t}\right\Vert$, $ \beta_{3}=\left\Vert \mathbf{C}^{t}-(\mathbf{C}^{t}\!\circ\!\mathbf{T}^{t})\right\Vert $, $\beta_{4}=\left\Vert\frac{\mathbf{1}_N^{\dagger}\mathbf{1}_N}{N}\!-\!\mathbf{C}^{t}\right\Vert$,  $\beta_{5}=\left\Vert\frac{\mathbf{1}_N^{\dagger}\mathbf{1}_N}{N}\!-\!\mathbf{C}^{t}\circ\mathbf{T}^{t}\right\Vert$, $\beta_6=\underset{\forall m}{\min}\Big(\underset{\forall n}{\sum} c_{n,m}^{t}(1-\epsilon_{\mathrm{P},n,m}^{t})\Big)$, $\beta_{7}=\underset{\forall n,m}{\min}\;(1-\epsilon_{\mathrm{P},n,m}^{t})$, and $\beta_{8}=\underset{\forall n,m}{\max}\;(1-\epsilon_{\mathrm{P},n,m}^{t})$. 
\end{lemma}%
\begin{proof}
   See \textbf{Appendices \ref{proof:B}}, \textbf{\ref{proof:C}} and \textbf{\ref{proof:D}}.
\end{proof}
When $\mathbf{T}^t$ is unknown and $\mathbf{C}^t$ is designed using \eqref{C_r_definition}, we obtain the upper and lower bounds of $\Phi$ by substituting \eqref{bounds_C_r_definition} into \eqref{problem1}. The bounds are given by 
    {\small\begin{align}
     \notag &\!\!\!\!\Phi\!\!\geq\!\!\frac{N\zeta_3p_{\max}\!\!+\!\!\zeta_4\!(1\!\!-\!\sqrt{p_{\max}}\!)^{2}}{\zeta_5}(1\!\!-\!\!\beta_{1}^{J})^{2}\!\!+\!\!\frac{\zeta_7\beta_{4}^{2J}}{\zeta_5}+(\beta_{2}^{J}\!\!-\!\!\beta_{1}^{J})^{2}\\&\!\!\!\!+\!\!\frac{\big(\zeta_3p_{\max}\!\!+\!\!\zeta_4(1\!\!-\!\!p_{\max})\!\!+\!\!\zeta_5\big)\beta_{7}^{2J\!+\!1}(1\!\!-\!\!\beta_{8})\beta_{2}^{2J\!+\!2}}{\zeta_5 N}\triangleq\Phi_{\mathrm{L}},\label{phi_lowerbound}\\
      &\!\!\!\! \Phi\!\leq\!\!\frac{\big(\zeta_3p_{\max}\!\!+\!\!\zeta_4(1\!\!-\!\!p_{\max})\!\!+\!\!1\big)(\beta_2^J\!\!-\!\!\beta_1^J)}{\zeta_5(\beta_2\!\!-\!\!\beta_1)}\beta_{3}\beta_1^J\!\!+\!\! \frac{\!N\zeta_3p_{\max}\!(1\!\!-\!\!\beta_6^{J})^{2}\!}{\zeta_5}\notag\\
       &\!\!\!\!+\!\!\frac{\zeta_4\!\!\left(1\!\!-\!\!\sqrt{p_{\max}}\right)^{2}}{\zeta_5}\!\!\Big(\!\frac{\beta_{2}^{J}\!\!-\!\!\beta_{1}^{J}}{\beta_{2}\!\!-\!\!\beta_{1}}\!\beta_{3}\!\!\!+\!\! \beta_{4}^{J}\!\Big)^2 \!\!+\!\!\frac{\zeta_7\beta_{4}^{2J}}{\zeta_5}\!\!+\!\!\Big(\!\frac{\beta_{2}^{J}\!\!-\!\!\beta_{1}^{J}}{\beta_{2}\!\!-\!\!\beta_{1}}\beta_{3}\!\Big)^2\!\!\triangleq\!\!\Phi_{\mathrm{U}}.\!\!\!\label{phi_upperbound}
     \end{align}}%
The lower bound $\Phi_{\mathrm{L}}$ under \eqref{C_r_definition} is made up of terms monotonically increasing with $J$, i.e., $1-\beta_1^J \geq 0$ and $\beta_2^J-\beta_1^J \geq 0$, and terms monotonically decreasing with $J$, i.e., $\beta_4^J \geq 0$ and $\frac{\beta_{7}^{2J\!+\!1}(1\!-\!\beta_{8})\beta_{2}^{2J\!+\!2}}{N} \geq 0$. To this end, $\Phi_{\mathrm{L}}$ first decreases and then increases with the growth of $J$. This can help narrow the search region for $J^*$, as stated in the following theorem.

 \begin{theo}
 \label{pro:optimal_J}
  The optimal number of local aggregations per round, $J^*$, is no greater than {\small\begin{align}\notag
    \!\!\! \!\! J_{\mathrm{TH}}\!\!=\!\!\underset{J\in\mathbb{R}_+}{\arg}\!\Big(\!\Phi_{\mathrm{L}}\!=\!\!\underset{J\in\mathbb{R}_+}{\min}\Phi_{\mathrm{U}}\!\Big).
  \end{align}}
\end{theo}
 \begin{proof}
The minimum of the upper bound in \eqref{phi_lowerbound},  $\min_{J\in\mathbb{R}_+}\Phi_{\mathrm{U}}$, can be obtained by solving $\frac{\partial\Phi_{\mathrm{U}}}{\partial J}=0$, and is larger than the minimum of $\Phi$, i.e., at $J^*$, is always smaller than the minimum of its upper bound. Moreover, at $J^*$, the lower bound in \eqref{phi_lowerbound}, that is, $\Phi_{\mathrm{L}}$, is lower than $\min_{J\in\mathbb{R}_+}\Phi$ and, in turn, lower than $\min_{J\in\mathbb{R}_+}\Phi_{\mathrm{U}}$; and $\Phi_{\mathrm{L}}$ also increases with~$J$ when $J$ is large. Since  $\min_{J\in\mathbb{R}_+}\Phi$ (i.e., at $J^*$) is no greater than $\min_{J\in\mathbb{R}_+}\Phi_{\mathrm{U}}$ while $\Phi_{\mathrm{L}}$ increases with $J$ when $J$ is large, $J^*$ must be in the region satisfying $\Phi_{\mathrm{L}}\leq\min_{J\in\mathbb{R}_+}\Phi$. As a result,  $J^*\leq J_{\mathrm{TH}}=\arg_{J\in\mathbb{R}_+}\big(\Phi_{\mathrm{L}}=\min_{J\in\mathbb{R}_+}\Phi_{\mathrm{U}}\big)$.
 \end{proof}
Based on \textbf{Theorem \ref{pro:optimal_J}}, we rewrite Problem \textbf{P2} as 
\begin{align}
    \textbf{P3:}\;\underset{J\in \mathbb{N}_{+}}{\min}\;\Phi \;\;\mathrm{s.t.} \; J\in \big[0,J_{\mathrm{TH}}\big],
    \text{(\ref{P1}a) -- (\ref{P1}d)}.\notag
 \end{align}%
This problem can be readily solved using a one-dimensional search, with the search range between $1$ and $J_{TH}$ (inclusive), and the step size of 1, as summarized in \textbf{Algorithm 2}.
\begin{algorithm} %[ht]
\caption{Optimization of the number of local aggregations}
\begin{algorithmic}[1]
\STATE Evaluate $p_{\max}$, $\zeta_1,\cdots,\zeta_7$, $\mathbf{M}_1,\cdots,\mathbf{M}_4$ to specify $\Phi$;
\STATE Evaluate $\beta_1,\cdots,\beta_8$ to specify the upper and lower bounds of $\Phi$, $\Phi_{\mathrm{L}}$ and $\Phi_{\mathrm{U}}$, using (33) and (34), respectively;
\STATE Derive $J_\mathrm{TH}$ by \textbf{Theorem 3};
\STATE Execute one-dimensional search for $J^*$ within $[0,J_\mathrm{TH}]$.
\end{algorithmic}
\end{algorithm}

When $\mathbf{T}^t$ is known \textit{a-priori} and $\mathbf{C}^t$ is designed using \eqref{CoT_r_definition}, the upper and lower bounds on $\Phi$ can be obtained by substituting \eqref{bounds_CoT_r_definition} into \eqref{problem1}. The lower bound decreases monotonically as $J$ increases, and it cannot be used to narrow the search region for $J^*$. Nevertheless, as shown numerically, $\Phi$ generally decreases as $J$ decreases. This indicates that more local aggregations always contribute to the convergence of D-FL, when the channel conditions are known \textit{a-priori} and the communication errors can be properly compensated for.

\section{Numerical Results}\label{section:results}
In this section, we conduct extensive experiments to validate the new convergence analysis and algorithm of D-FL. 
    \begin{table*}[]\vspace{-2mm }
    \footnotesize
\centering
\caption{The default coordinates of all devices, generated randomly.}
\label{tab:coordinate}
\begin{tabular}{c|cccccccccccc}
\hline
Device ID & 1 & 2 & 3 & 4 & 5 & 6 & 7 & 8 & 9 & 10 & 11 & 12 \\ \hline
$x$-coordinate (m) & 2196 & 3637 & 2642 & 2884 & 5254 & 1730 & 3572 & 4546 & 4328 & 2534 & 173 & 2050 \\ \hline
$y$-coordinate (m)& 1351 & 3127 & 284 & 848 & 596 & 1923 & 2668 & 5326 & 4001 & 5171 & 575 & 3676 \\ \hline
non-i.i.d. F-MNIST Group& 1 & 2 & 3 & 4 & 5 & 6 & 7 & 8 & 2 & 8 & 4 & 6 \\ \hline
\end{tabular}
\end{table*}
Following a random geometric graph model~\cite{penrose2003random}, a 12-vertex graph with an undirected and connected topology is generated at random. Table \ref{tab:coordinate} lists the coordinates of the 12 vertices. When there are $N\leq 12$ participating devices, we select the first $N$ vertices. The connectivity density of edges (i.e., transmitter-receiver pairs) in the topological graph of the network is $\rho$; in other words, the number of transmitter-receiver pairs in the considered network is $\rho\times\frac{N(N-1)}{2} $. By default, $N=10$ and $\rho=0.5$. All devices operate at the central radio frequency $f_c=2.5$ GHz and bandwidth $B=30$ MHz. The transmit power of each device is $P=20$ dBm. The noise power spectral density is $N_0=-174$ dBm/Hz. {We set the channel gain to $h^t_{m,n} (\text{dB})=20\text{log}(f_c)+20\text{log}(d_{m,n})+32.4$~\cite{rappaport2010wireless}, where $d_{m,n}$ (km) is the distance between neighboring clients $m$ and $n$. Considering BPSK or QPSK, the PERs of the direct links are evaluated accordingly.\footnote{Take BPSK/QPSK in the AWGN channel for example. $\epsilon_{\mathrm{B},n,m}^{t}=Q(\sqrt{2\gamma^{t}_{n,m}})$, where $\gamma^{t}_{n,m}$ is the signal-to-noise ratio (SNR) between the two devices, and $Q(x)=\frac{1}{\sqrt{2\pi}}\int_{x}^{\infty}e^{-\frac{x^{2}}{2}}\textrm{d}x$ is the $Q$ function~\cite{goldsmith_2005}.}}

We consider two image classification tasks as follows.
\begin{itemize}
    \item \textbf{F-MNIST \& CNN:} We consider an image classification task to classify a non-identical and independently distributed (non-i.i.d.) F-MNIST dataset based on a CNN model. The F-MNIST dataset contains the standard MNIST dataset of ten handwritten digits, and 26 upper- and lower-case English characters. Each sample is a grayscale image with $28\times28$ pixels. The dataset has 3,400 subsets and 62 labels. We select 12 subsets at random. We divide the labels into eight independent groups. Each device selects a group with 7 or 8 labels and 20 to 70 training examples. It is estimated that $L=2$ and $\mu =0.0016$. The batch size is 20. 

\hspace{3 mm} The considered CNN model consists of two convolutional layers with 32 or 64 convolutional filters per layer. We add a pooling layer between the two convolutional layers to prevent overfitting. Following the convolutional layers, we insert two fully-connected layers to integrate the local information and generate the output. We use the rectified linear unit (ReLU) in the convolutional and fully-connected layers. By default, the CNN model has $1.21\times 10^6$ parameters in $38.72$ Mbits. The model is divided into $1,600$ packets with $24.2$ kbits per packet; unless otherwise specified. The learning rate is $\eta=0.03$. 

\item \textbf{F-CIFAR100 \& ResNet-18:} We consider another image classification task to train the F-CIFAR100 dataset based on the ResNet-18 model. The CIFAR100 dataset comprises 100 classes, with 500 training images and 100 test images per class. Each image is a $32\times 32$-pixel color image. F-CIFAR100 divides CIFAR100 into 500 training subsets with 100 images per subset. We select 10 subsets from F-CIFAR100 at random.

\hspace{3 mm}The ResNet-18 model consists of 18 layers: 17 convolutional layers and a fully-connected layer.
The model has $11.69\times 10^6$ parameters in $374.08$~Mbits. By default, the model is divided into 800 packets with $467.6$ kbits per packet. The learning rate is~$\eta=0.1$. It is estimated that $L=3$ and $\mu =0.005$. The batch size is~20.
\end{itemize}
The following FL models are compared with D-FL.
\begin{itemize}
    
    \item \textbf{C-FL:} In C-FL, one device is selected as the aggregation center. All devices transmit their local models (hop-by-hop) to the device, where all local models are aggregated to update the global model. The global model is then returned to all devices in the reverse links. We consider different nodes as the aggregation center, and select the best performer in every simulation. 

    \item  \textbf{ Unsegmented D-FL (U-D-FL):} A special case of D-FL was considered in~\cite{9772390}, where a device precludes any models received with errors from its aggregation, i.e., one packet per model. Moreover, only one local aggregation is performed per training round, i.e., $J=1$~\cite{9772390}. 
   
\end{itemize}
We examine both perfect and imperfect channels.

\subsection{CNN on F-MNIST}
\begin{figure*}[t]
	\centering
	\subfigure
	{
		\begin{minipage}[t]{0.2\textwidth}
			\centering
				\includegraphics[width=4cm]{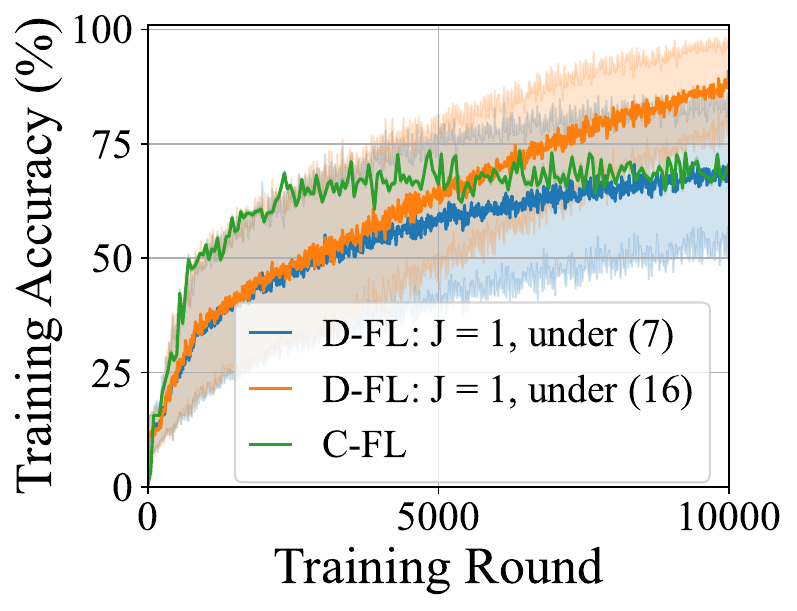}
		\end{minipage}
	}
	\subfigure
	{
		\begin{minipage}[t]{0.2\textwidth}
			\centering
				\includegraphics[width=4cm]{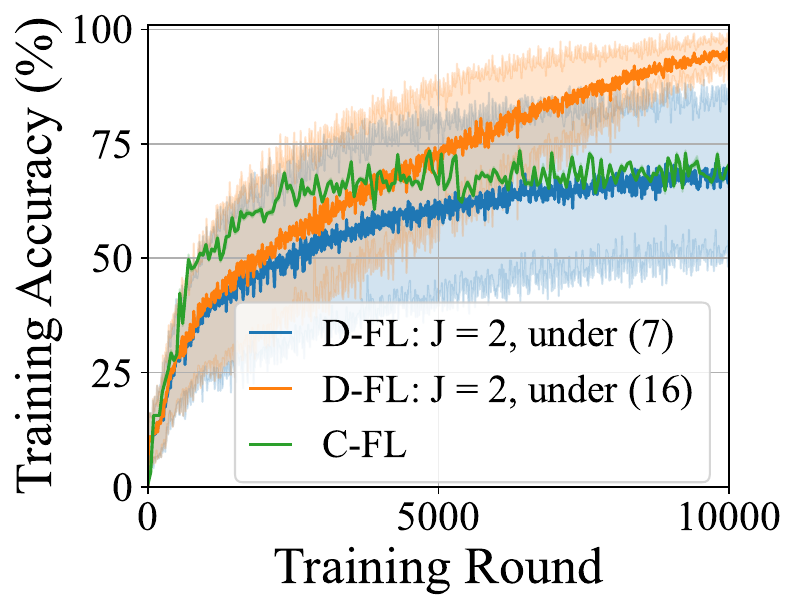}
		\end{minipage}
	} 
 \subfigure
	{
		\begin{minipage}[t]{0.2\textwidth}
			\centering
				\includegraphics[width=4cm]{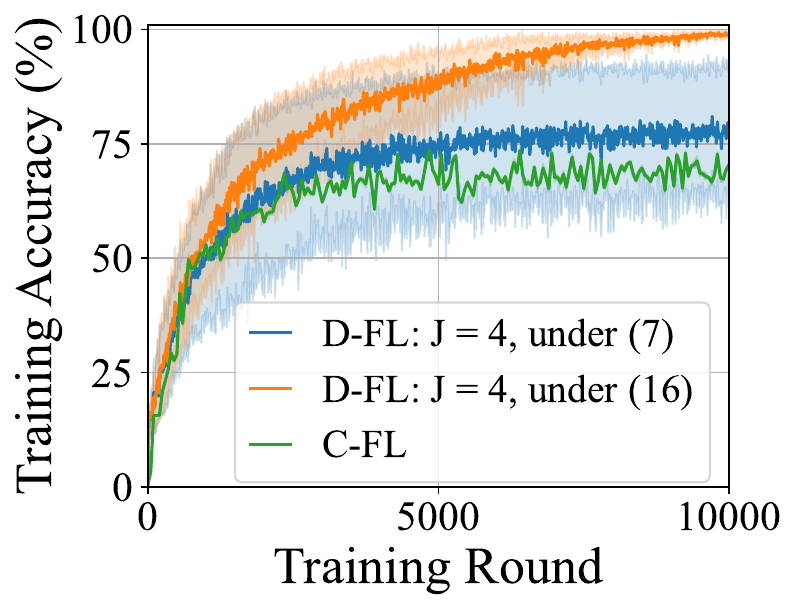}
		\end{minipage}
	}
 \subfigure
	{
		\begin{minipage}[t]{0.2\textwidth}
			\centering
				\includegraphics[width=4cm]{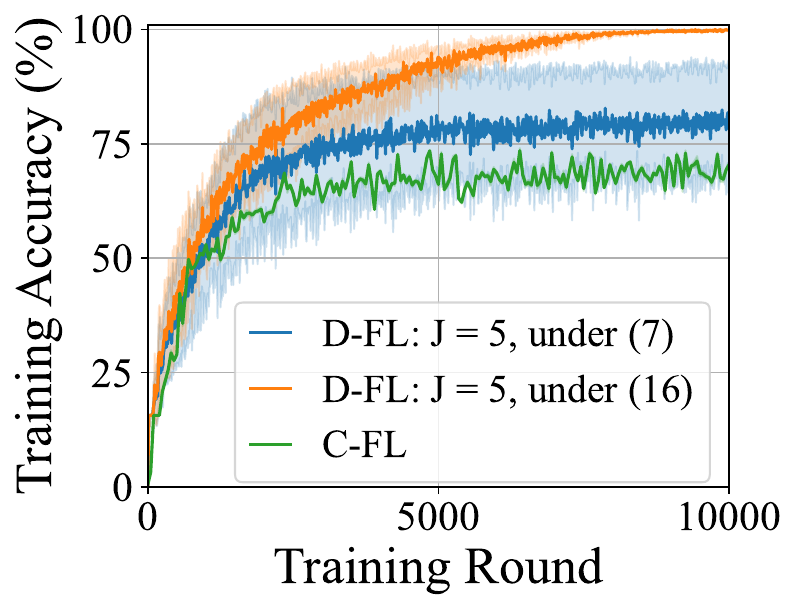}
		\end{minipage}
	}
 \subfigure
	{
		\begin{minipage}[t]{0.2\textwidth}
			\centering
				\includegraphics[width=4cm]{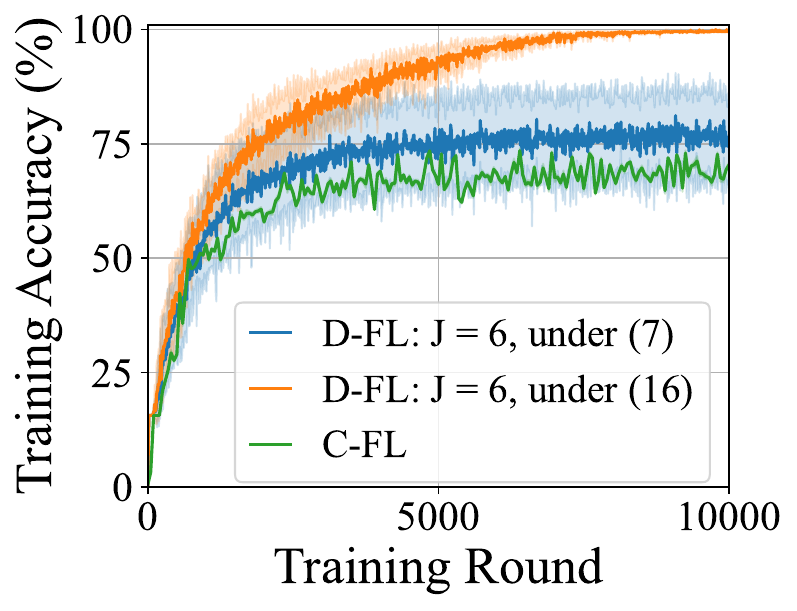}
		\end{minipage}
	}
 \subfigure
	{
		\begin{minipage}[t]{0.2\textwidth}
			\centering
				\includegraphics[width=4cm]{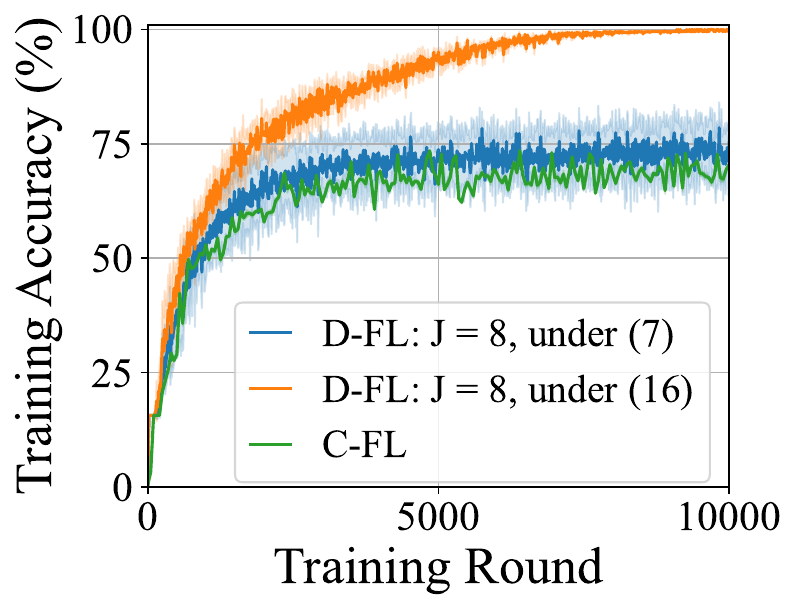}
		\end{minipage}
	}
 \subfigure
	{
		\begin{minipage}[t]{0.2\textwidth}
			\centering
				\includegraphics[width=4cm]{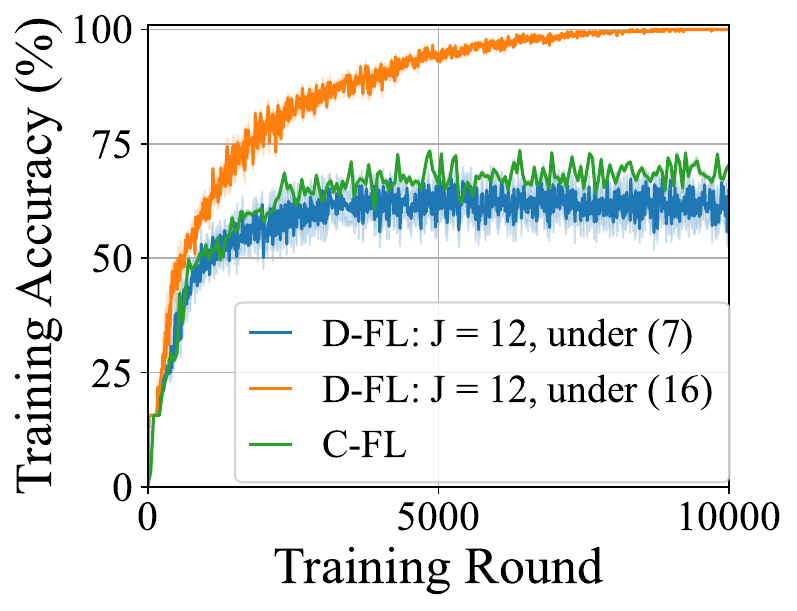}
		\end{minipage}
	}
 \subfigure
	{
		\begin{minipage}[t]{0.2\textwidth}
			\centering
				\includegraphics[width=4cm]{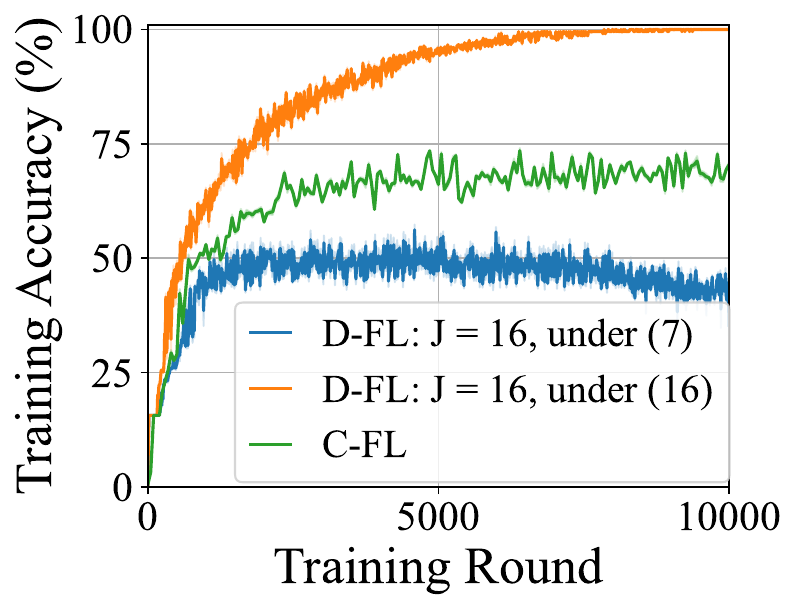}
		\end{minipage}
	}\vspace{-3mm }
\caption{The training accuracy of D-FL versus the training rounds, where the CNN model and non-i.i.d. F-MNIST dataset. $N=10$, $\rho=0.5$, and $\kappa=1$. The C-FL is also considered with device $4$ serving as the central aggregator.}\label{Convergence curves}\vspace{-2mm }
\end{figure*}

\begin{figure} %[h]
\centering
\includegraphics[scale=0.65]{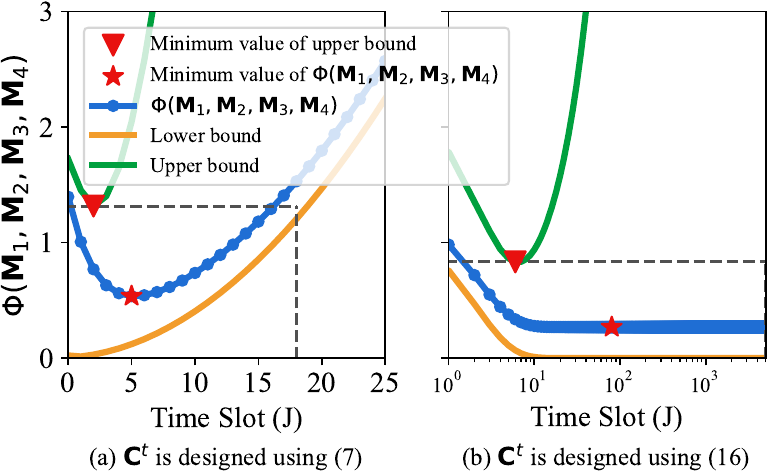}
\caption{The upper and lower bounds of $\Phi$ under the CNN model, where the default setting is considered: $N=10$ and $\rho=0.5$.}
    \label{fig:lower_upper}\vspace{-3mm }
\end{figure}
Fig. \ref{Convergence curves} shows the training accuracy of D-FL with the increasing number of rounds, where $N=10$ and different numbers of local aggregations per round, $J$, are considered. Two consensus matrix designs of D-FL using \eqref{C_r_definition} and \eqref{CoT_r_definition} are tested, and C-FL is also plotted. Each solid line is the average training accuracy of all 10 devices under a given $J$. The shade shows the individual training accuracy of the $10$ devices. 

When the channel conditions $\mathbf{T}^t$ of the network are not known \textit{a-priori} and the consensus matrix $\mathbf{C}^t$ is designed using~\eqref{C_r_definition}, the best convergence of D-FL is achieved under $J=5$ and is even better than C-FL. When $J<5$, the average training accuracy after convergence increases with $J$, benefiting from increased aggregation opportunities. When $J>5$, the average training accuracy after convergence decreases with the further increase in $J$. This is due to the increasingly accumulated impact of transmission errors on accuracy. Upon convergence, the differences among the local models of the devices decrease consistently with the increase in $J$. This is reasonable since increasing transmissions among the devices contribute to the conformity of the local models. Moreover, the optimal number of local aggregations is $J^*=5$ in Fig. \ref{Convergence curves}, which is consistent with the observation in Fig. \ref{fig:lower_upper}(a), where the upper and lower bounds of $\Phi$ are used to narrow the search range of $J$.

On the other hand, when $\mathbf{T}^t$ is known \textit{a-priori} and $\mathbf{C}^t$ is designed using \eqref{CoT_r_definition}, the convergence and the average training accuracy of D-FL improve by increasing $J$ although the improvements slow down as $J$ increases, as shown in Fig.~\ref{Convergence curves}. This is consistent with the observation in Fig. \ref{fig:lower_upper}(b), where $\Phi$ decreases with increasing of $J$ over a wide range of~$J$. Moreover, D-FL under \eqref{CoT_r_definition} substantially outperforms D-FL under \eqref{C_r_definition}, as shown in Fig. \ref{Convergence curves} and also consistently shown in the rest of the figures in this paper. The conclusion drawn is that the \textit{a-priori} knowledge of the imperfect channels can help compensate for the imperfect local aggregations resulting from the communication errors and contribute to improved convergence and accuracy.

Next, we evaluate D-FL statistically, where the average training accuracy of the last 1,000 rounds is considered. We also scale the graph of the considered network with a scaling factor $\kappa$; i.e., the graph (and all its edges) is expanded by $\kappa$. Then, the impacts of $N$, the connectivity density $\rho$, and $\kappa$ are assessed on the average training accuracy and $\Phi$,  when $\mathbf{T}^t$ is unknown \textit{a-priori} and $\mathbf{C}^t$ is designed using \eqref{C_r_definition}.

Fig. \ref{fig:user_num_changes} demonstrates that under different numbers of devices, the optimal numbers of local aggregations per round, i.e., $J^*$, identified in Fig. \ref{fig:user_num_changes}(a) by solving Problem \textbf{P3} are consistent with the optimal numbers experimentally obtained in Fig. \ref{fig:user_num_changes}(b), when $\mathbf{C}^t$ is designed using \eqref{C_r_definition}. Moreover, Fig. \ref{fig:user_num_changes}(a) shows that increasing the participating devices does not always help improve the training accuracy of D-FL in the presence of transmission errors. This is because the training accuracy depends heavily on the topology of the participating devices (not only the number), as well as the SNR of the links. Nevertheless, the highest training accuracy of D-FL obtained under the identified optimal number of local aggregations per training round is better than that of C-FL. The effectiveness of the algorithm is demonstrated.

\begin{figure}\vspace{-2mm }
	\centering
	\subfigure[$\Phi(\mathbf{M}_1,\mathbf{M}_2,\mathbf{M}_3,\mathbf{M}_4)$ vs. $J$.]
	{
		\begin{minipage}[t]{0.45\textwidth}
			\centering
		\includegraphics[width=7cm]{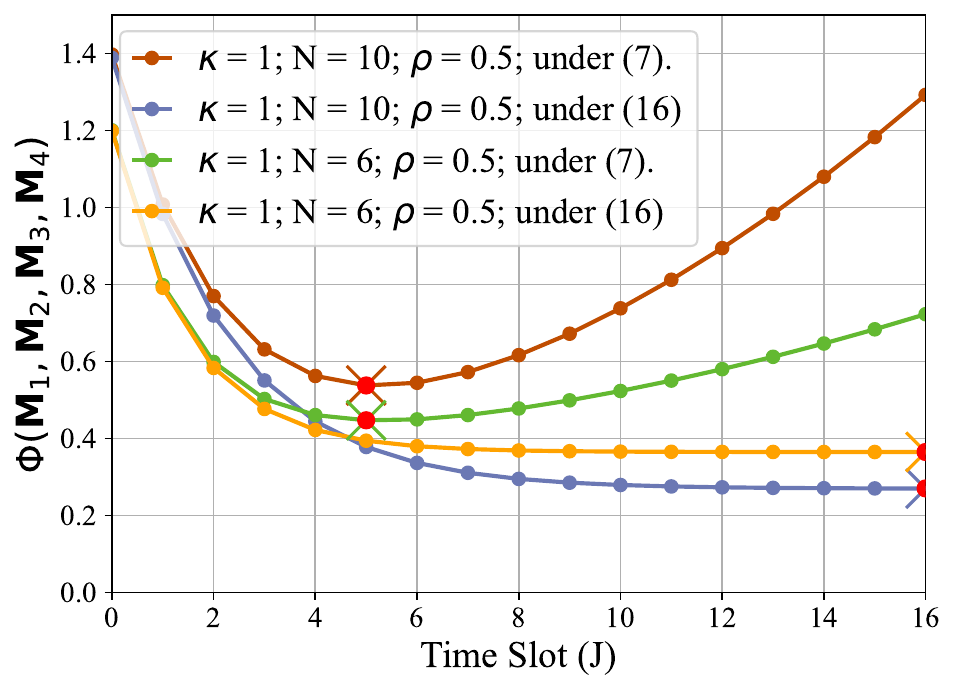}
		\end{minipage}
	}
	\subfigure[Training accuracy vs. $J$, where D-FL (solid line) and C-FL (dotted line) are compared under the same topology.] 
	{
		\begin{minipage}[t]{0.45\textwidth}
			\centering
			\includegraphics[width=7cm]{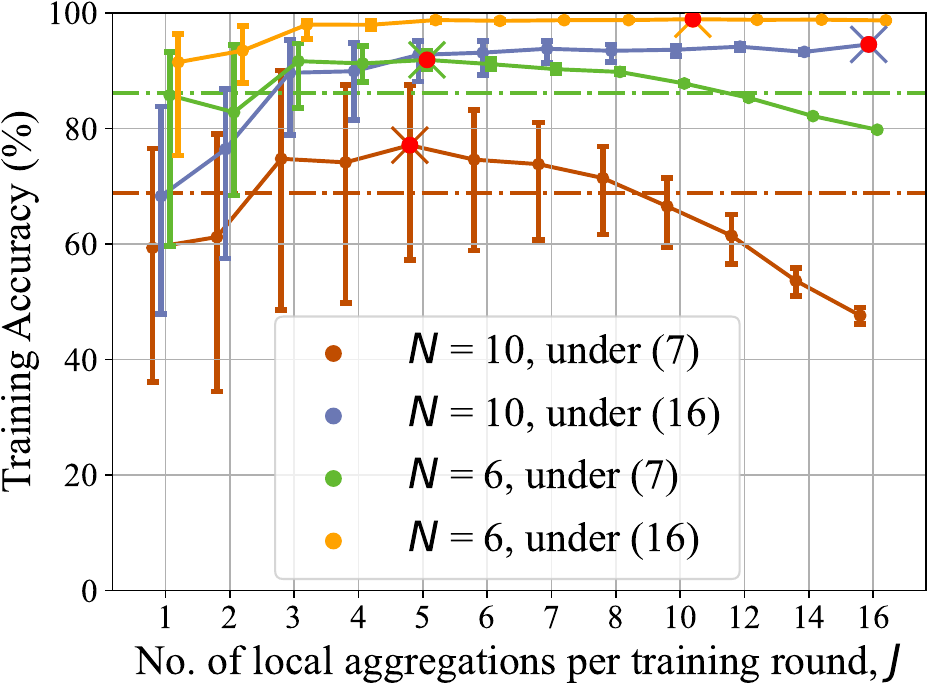}
		\end{minipage}
	} \vspace{-2mm }\caption{The impact of the number of devices  under the CNN model.}\label{fig:user_num_changes}\vspace{-3mm }
\end{figure}
\begin{figure}[h]\vspace{-3mm }
	\centering
	\subfigure[$\Phi(\mathbf{M}_1,\mathbf{M}_2,\mathbf{M}_3,\mathbf{M}_4)$ vs. $J$.]
	{
		\begin{minipage}[t]{0.45\textwidth}
			\centering
				\includegraphics[width=7cm]{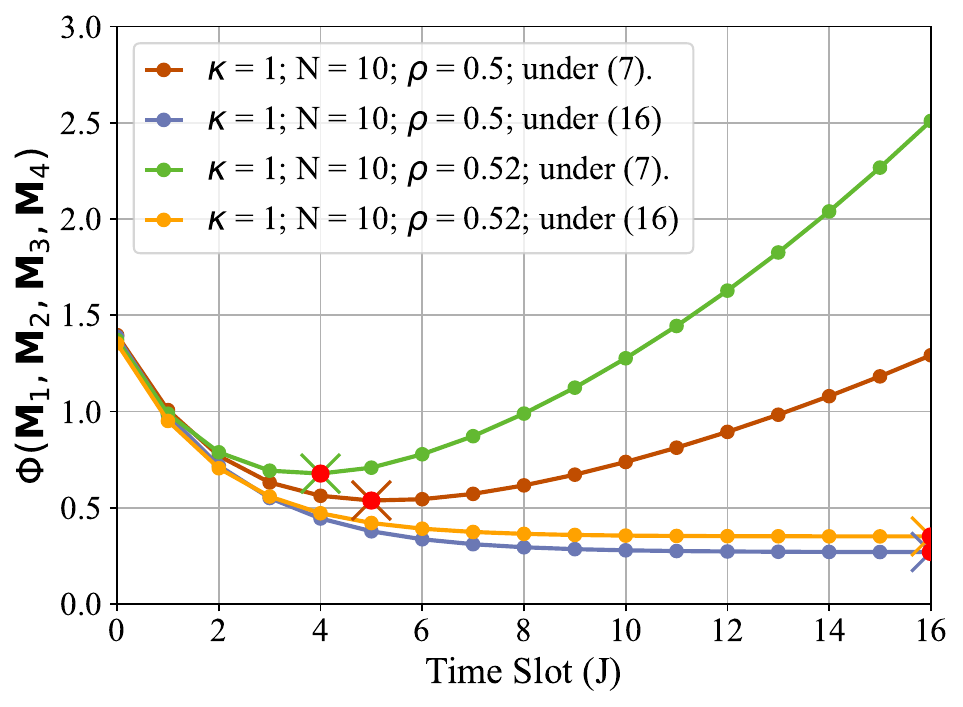}\vspace{-2mm }
		\end{minipage}
	}
	\subfigure[Training accuracy vs. $J$, where D-FL (solid line) and C-FL (dotted line) are compared under the same topology.]
	{
		\begin{minipage}[t]{0.45\textwidth}
			\centering
			\includegraphics[width=7cm]{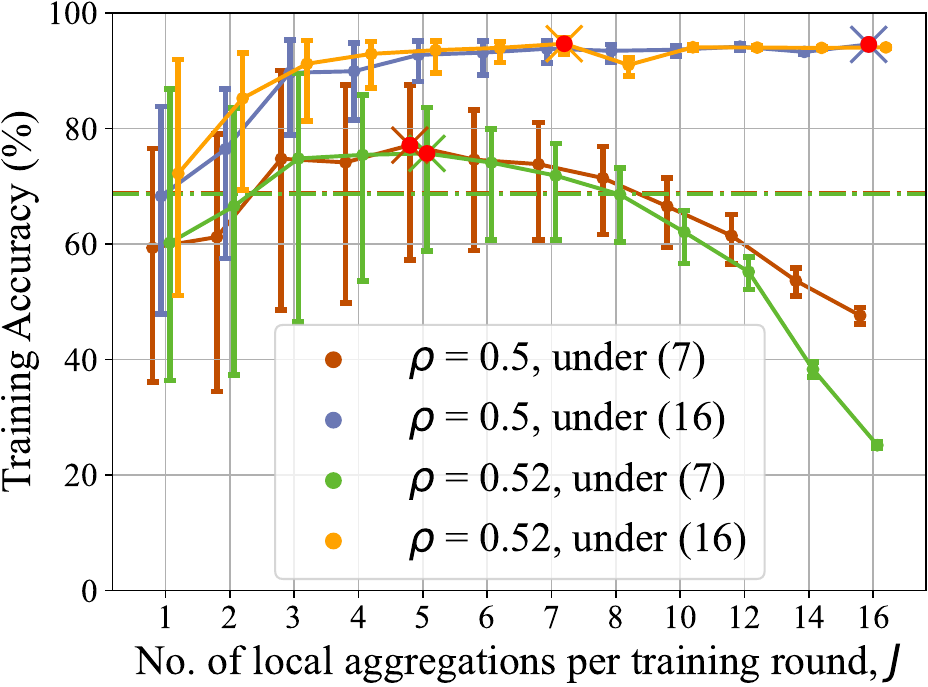}
		\end{minipage}\vspace{-3mm }
	} \vspace{-2mm }\caption{The impact of connection density  under the CNN model.}\label{fig:density}
\end{figure}

Fig. \ref{fig:density}(a) shows that the minimum of $\Phi$ is smaller under a smaller $\rho$, i.e., $J^*$ decreases as $\rho$ increases, when $\mathbf{C}^t$ is designed using \eqref{C_r_definition}. In other words, the reason is that higher connectivity leads to more edges; a device is likely to have more neighbors. This allows each device to receive the local models further away, helping reduce the variance of the locally aggregated models among the devices. Fig. \ref{fig:density}(b) shows that the highest training accuracy of D-FL is achieved at the optimal numbers of local aggregations $J^*$ identified in Fig. \ref{fig:density}(a). The optimization of $J^*$ is important, as it allows D-FL to outperform C-FL, as shown in Fig. \ref{fig:density}(b). When both $J$ and $\rho$ are large, $\Phi$ is large, see Fig. \ref{fig:density}(a), and the training accuracy of D-FL is low and can be worse than that of C-FL, see Fig. \ref{fig:density}(b). The reason is that the growth of $\rho$ leads to an increase of long and weak links, while the growth of $J$ results in accumulations of transmission errors in the locally aggregated models.

\begin{figure}\vspace{-3mm }
	\centering
	\subfigure[$\Phi(\mathbf{M}_1,\mathbf{M}_2,\mathbf{M}_3,\mathbf{M}_4)$ vs. $J$.]
	{
		\begin{minipage}[t]{0.45\textwidth}
			\centering
					\includegraphics[width=7cm]{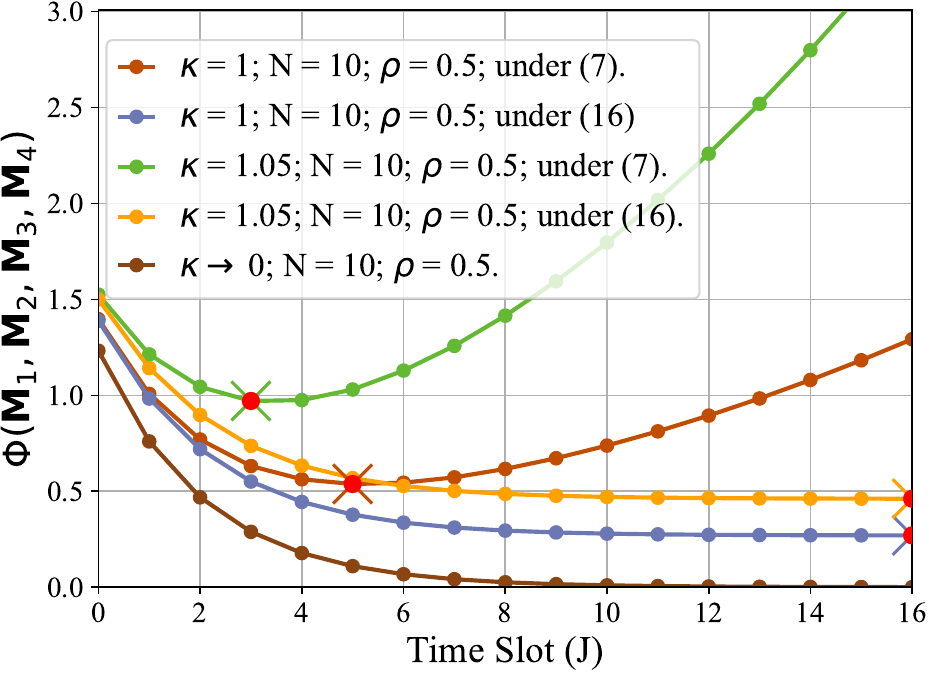}\vspace{-2mm }
		\end{minipage}
	}
	\subfigure[Training accuracy vs. $J$, where D-FL (solid line) and C-FL (dotted line) are compared under the same topology.]
	{
		\begin{minipage}[t]{0.45\textwidth}
			\centering
			\includegraphics[width=7cm]{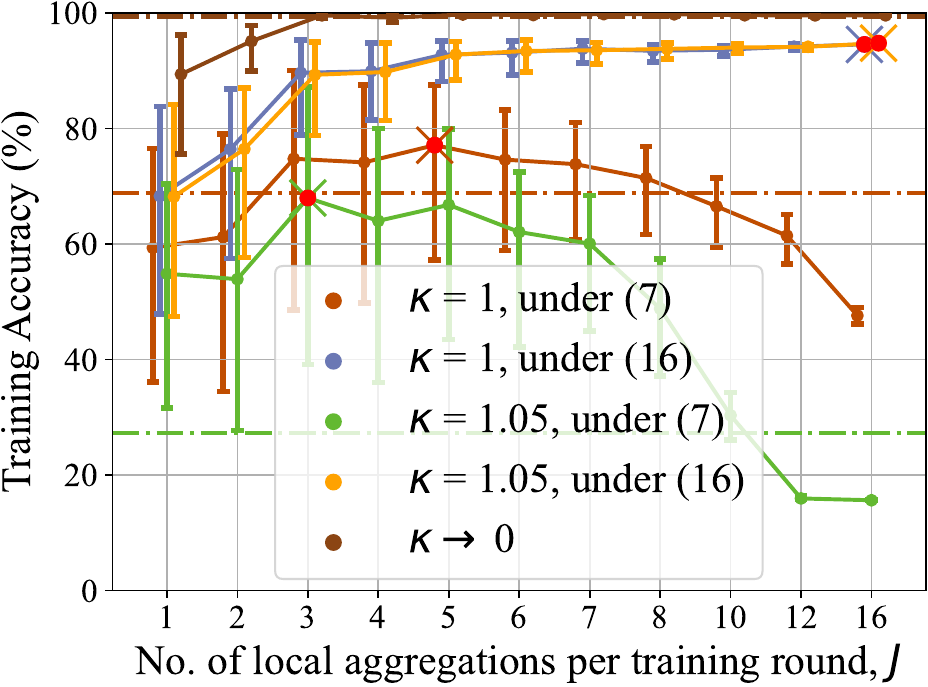}
		\end{minipage}
	}\vspace{-1mm } \caption{The impact of network scale under the CNN model.}\label{fig:scale_changes}
\end{figure}

\begin{figure}[h]
    \centering
    \includegraphics[scale=0.45]{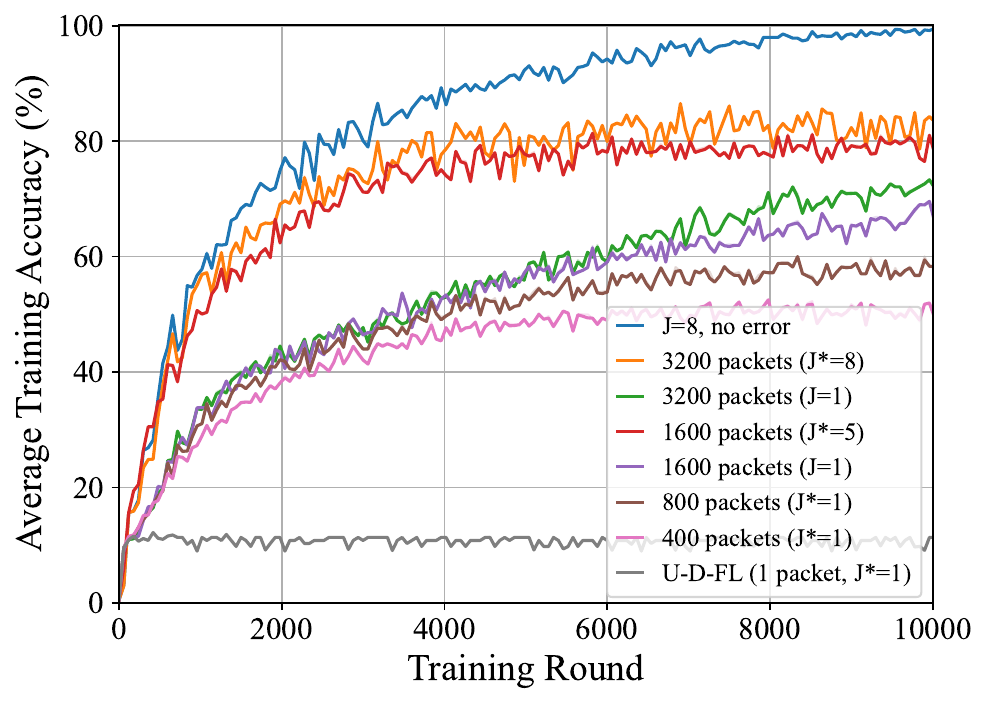}
    \caption{The convergence curves of D-FL under different numbers of packets under the CNN model when $\mathbf{C}^t$ is given by \eqref{C_r_definition}. }
    \label{fig:node_failure}\vspace{-3mm }
\end{figure}

Fig.~\ref{fig:scale_changes}(a) shows that $J^*$ decreases as $\kappa$ increases,  when $\mathbf{C}^t$ is designed using \eqref{C_r_definition}. Fig.~\ref{fig:scale_changes}(b) indicates that the highest training accuracy of D-FL is achieved at $J^*$ identified in Fig. \ref{fig:scale_changes}(a), which is always better than the training accuracy of C-FL. In other words, a larger network scale would require fewer local aggregations per round to achieve its highest training accuracy. This is because as $\kappa$ increases, the link distances of the considered network grow, reducing the SNR and increasing the PER. It is noticed in Fig. \ref{fig:scale_changes} that when $\kappa$ is sufficiently small, e.g., $\kappa \rightarrow 0^+$, the communications between the devices are error-free. Increasing $J$ does no harm to the training accuracy, which can be as high as 100\%, even when $\mathbf{T}^t$ is unknown \textit{a-priori} and $\mathbf{C}^t$ is designed using \eqref{C_r_definition}.

Fig. \ref{fig:node_failure} studies the impact of model segmentation (or packetization) on D-FL when $\mathbf{T}^t$ is unknown \textit{a-priori} and $\mathbf{C}^t$ is designed using \eqref{C_r_definition}. We also plot an ideal situation, where the link conditions are ideal  (i.e., no error), and the number of local aggregations per round is $J=8$. We see that D-FL can benefit from segmenting a model into smaller packets because a larger packet is more prone to errors. Moreover, $J^*$ increases as the packet size $L_{\mathrm{p}}$ decreases, since more aggregations of quality local models improve the convergence of D-FL.

On the other hand, when $\mathbf{T}^t$ is known \textit{a-priori} and $\mathbf{C}^t$ is designed using \eqref{CoT_r_definition}, it is consistently demonstrated in Figs.~\ref{fig:user_num_changes} -- \ref{fig:scale_changes} that the convergence and the training accuracy of D-FL improve with as $J$ increases, which verifies the finding in Fig. 3b. Moreover, given $J$, the convergence and the training accuracy improve with $N$, $\rho$, and $\kappa$ across Figs.~\ref{fig:user_num_changes} -- \ref{fig:scale_changes}, which is consistent with the observations in the scenario with $\mathbf{T}^t$ unknown \textit{a-priori} and $\mathbf{C}^t$ designed using \eqref{C_r_definition}.

\subsection{ResNet-18 on F-CIFAR100}
Next, we consider the ResNet-18 model on the F-CIFAR100 dataset. The experimental setups are the same as the default setups under the CNN model and F-MNIST dataset, except that 100 packets are considered here. As shown in Fig.~\ref{Convergence_CIFAR_100}, the ResNet-18 model displays a consistent convergence trend with the CNN model shown in Fig. \ref{Convergence curves}. Particularly, when $\mathbf{T}^t$ is unknown \textit{a-priori} and $\mathbf{C}^t$ is designed using \eqref{C_r_definition}, the best convergence of D-FL in Fig. 8 is achieved by setting $J=5$ per round. The optimum numbers of local aggregations observed in Fig.~\ref{fig:Fed-cifar100}(b) are consistent with $J^*$ identified analytically in Fig.~\ref{fig:Fed-cifar100}(a). Moreover, the average training accuracy increases when a model is segmented into a larger number of smaller packets, as also shown in Fig.~\ref{fig:node_failure}. On the other hand, when $\mathbf{T}^t$ is known \textit{a-priori} and $\mathbf{C}^t$ is designed using \eqref{CoT_r_definition}, the convergence and the average training accuracy of D-FL are improved by increasing $J$, as shown in Fig. \ref{fig:Fed-cifar100}.

\begin{figure*}[t]\vspace{-3mm }
	\centering
	\subfigure
	{
		\begin{minipage}[t]{0.2\textwidth}
			\centering
				\includegraphics[width=4cm]{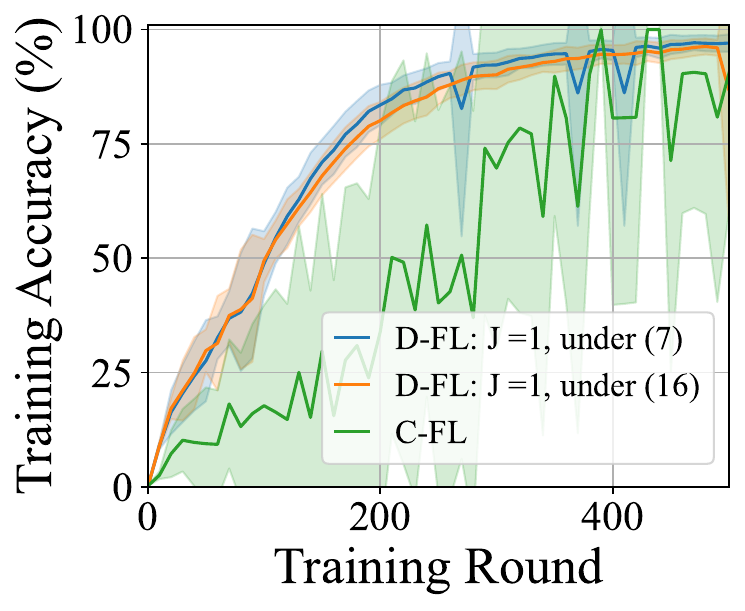}
		\end{minipage}
	}
	\subfigure
	{
		\begin{minipage}[t]{0.2\textwidth}
			\centering
				\includegraphics[width=4cm]{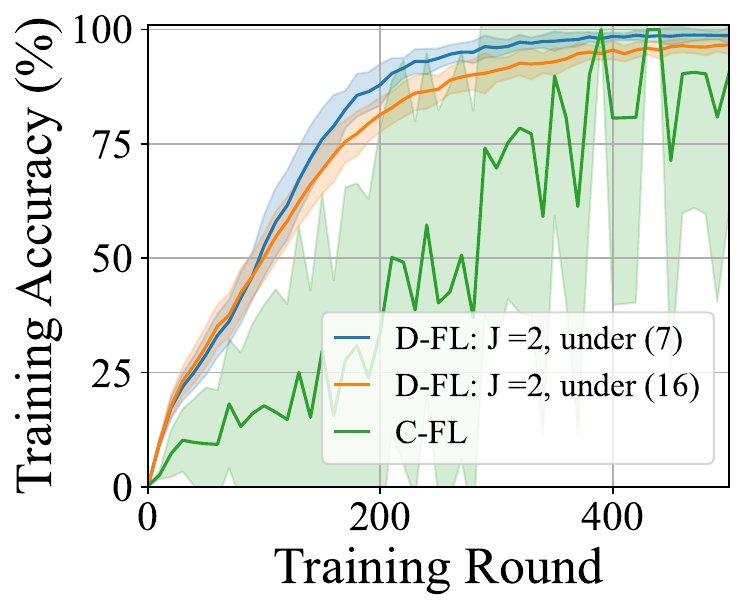}
		\end{minipage}
	} 
 \subfigure
	{
		\begin{minipage}[t]{0.2\textwidth}
			\centering
				\includegraphics[width=4cm]{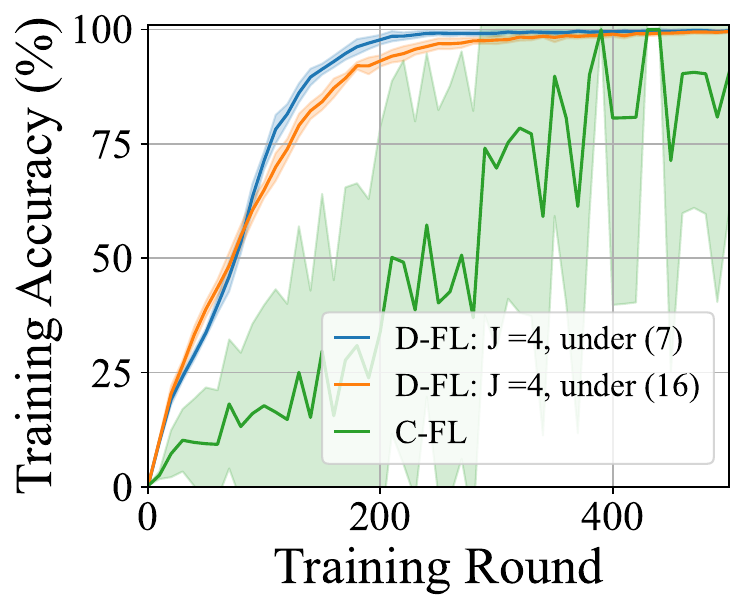}
		\end{minipage}
	}
 \subfigure
	{
		\begin{minipage}[t]{0.2\textwidth}
			\centering
				\includegraphics[width=4cm]{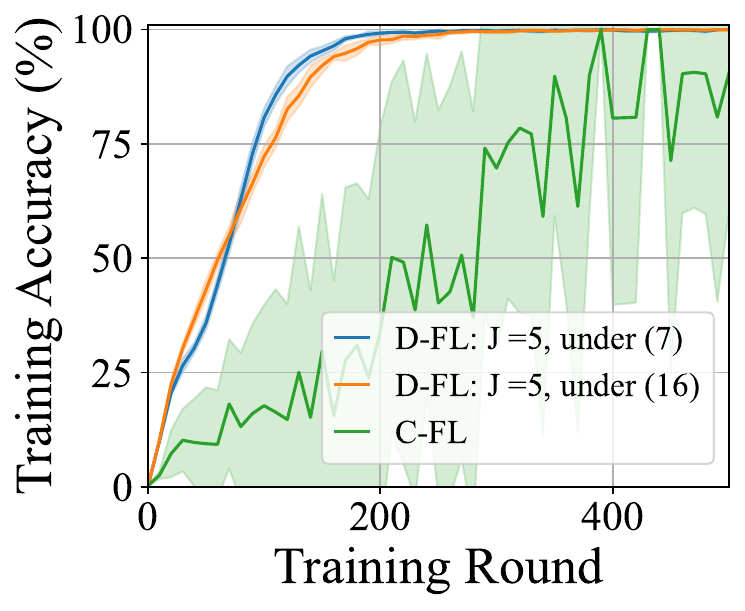}
		\end{minipage}
	}
 \subfigure
	{
		\begin{minipage}[t]{0.2\textwidth}
			\centering
				\includegraphics[width=4cm]{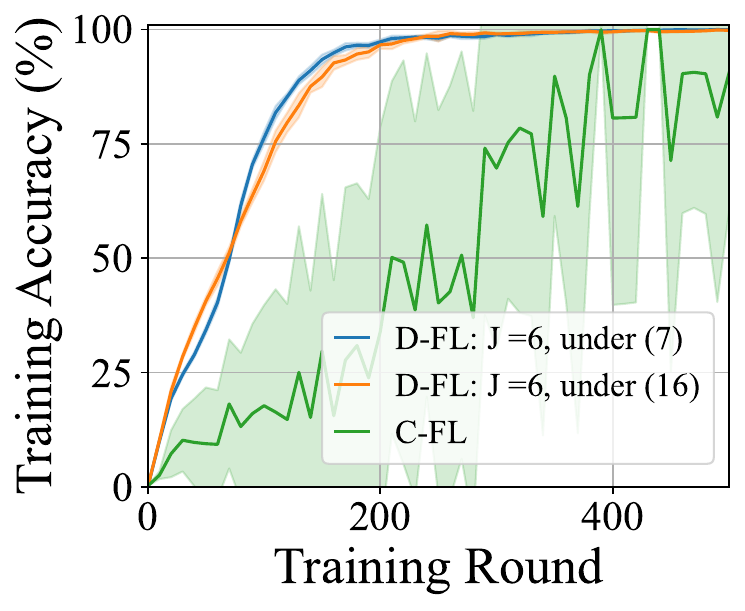}
		\end{minipage}
	}
 \subfigure
	{
		\begin{minipage}[t]{0.2\textwidth}
			\centering
				\includegraphics[width=4cm]{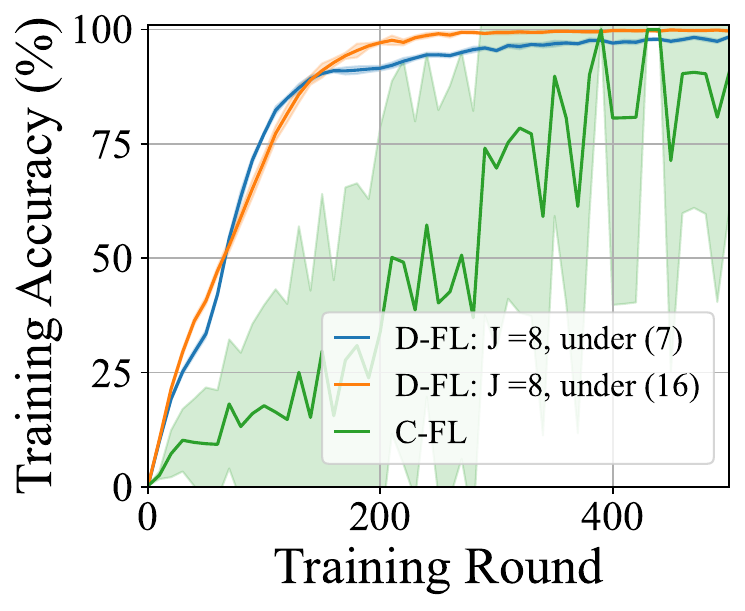}
		\end{minipage}
	}
 \subfigure
	{
		\begin{minipage}[t]{0.2\textwidth}
			\centering
				\includegraphics[width=4cm]{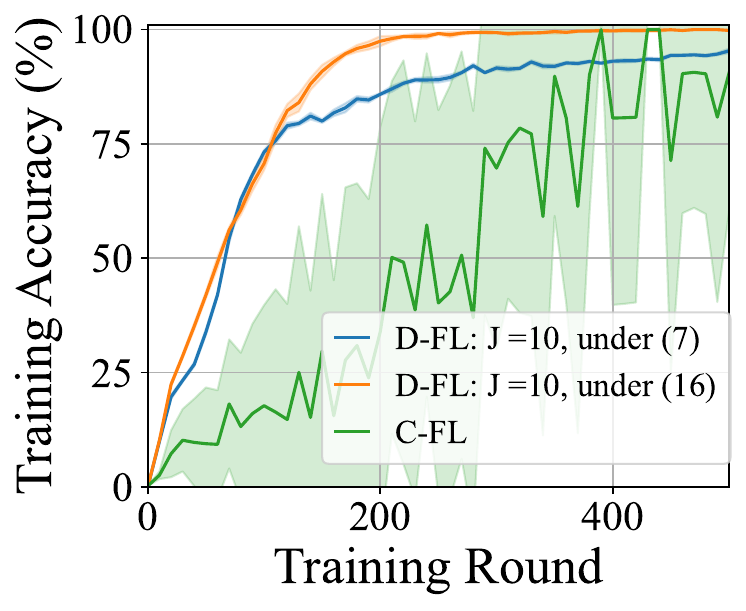}
		\end{minipage}
	}
 \subfigure
	{
		\begin{minipage}[t]{0.2\textwidth}
			\centering
				\includegraphics[width=4cm]{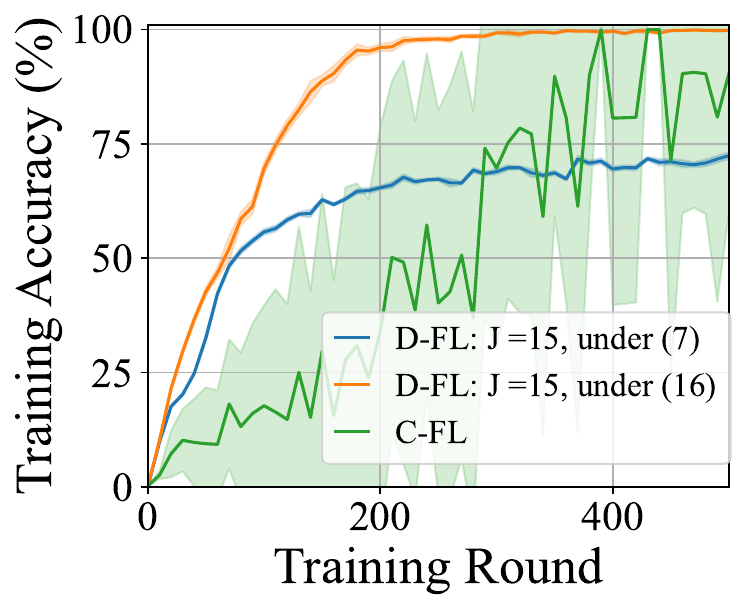}
		\end{minipage}
	}\vspace{-3mm }
\caption{The training accuracy of D-FL versus rounds, where the ResNet-18 model and F-CIFAR100 dataset are considered. $N=10$, $\rho=0.5$, and $\kappa=1$. The C-FL is also considered with device $1$ serving as the central aggregator.}\label{Convergence_CIFAR_100}
\end{figure*}

\begin{figure}[h]\vspace{-3mm }
	\centering
	\subfigure[$\Phi(\mathbf{M}_1,\mathbf{M}_2,\mathbf{M}_3,\mathbf{M}_4)$ vs. $J$.]
	{
		\begin{minipage}[t]{0.4\textwidth}
			\centering
					\includegraphics[width=7cm]{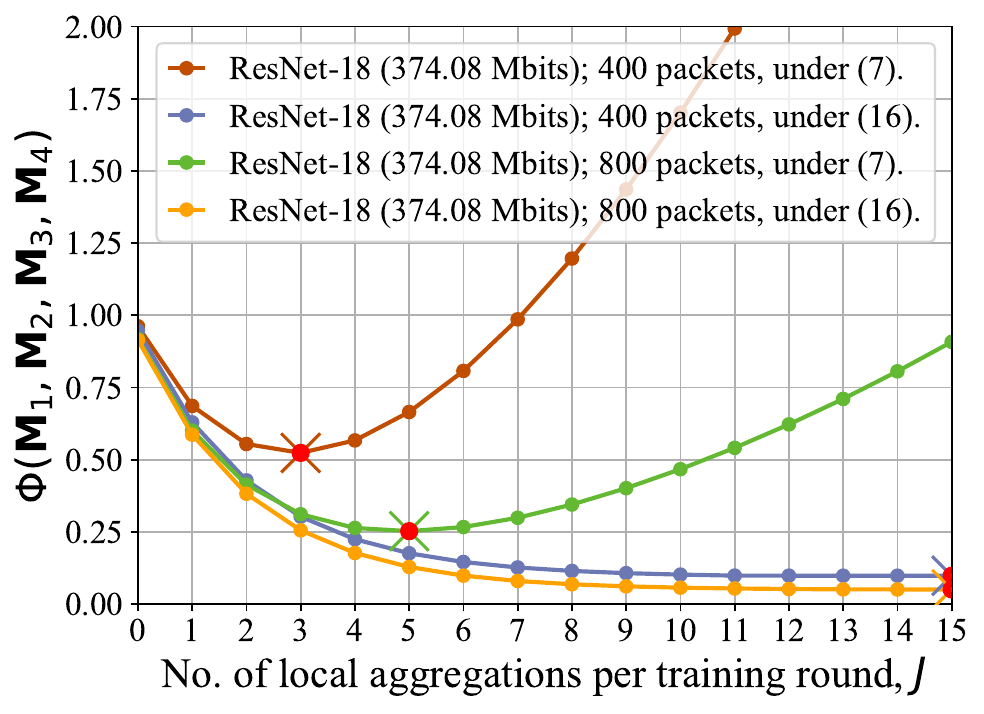}
		\end{minipage}
	}
	\subfigure[Training accuracy vs. $J$.]
	{
		\begin{minipage}[t]{0.4\textwidth}
			\centering
			\includegraphics[width=7cm]{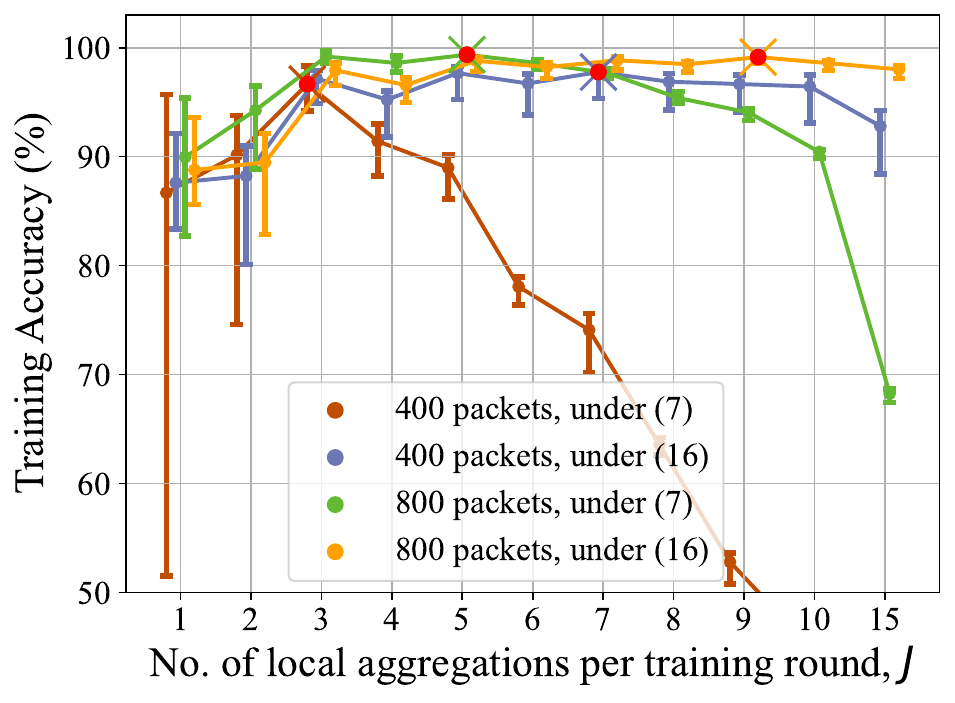}
		\end{minipage}
	} \vspace{-3mm }\caption{Impact of model segmentation under the ResNet-18 model.}\label{fig:Fed-cifar100}\vspace{-2mm }
\end{figure}

In addition to the communication errors and consensus matrix designs, $\Phi$ also depends on the convexity and smoothness of the model (i.e., $\mu$ and $L$), the hyper-parameters (i.e., $\eta$ and $I$), and the packet length $L_{\mathrm{p}}$. For the considered CNN model, $\mu =0.0016$, $L=2$,  $\eta=0.03$, $I=5$, and subsequently, $\Phi=1041\Vert\mathbf{1}_{N}\mathbf{M}_{1}\Vert^{2}+\Vert\mathbf{M}_{3}\Vert^{2}+0.03\Vert\mathbf{M}_{4}\Vert^{2}+1.31\Vert\mathbf{M}_{1}\Vert^{2}+1042\big\Vert\mathbf{M}_{2}\big\Vert.$ For the ResNet-18 model, $\mu =0.005$, $L=3$, $\eta=0.1$, $I=5$, and subsequently, $\Phi=333\Vert\mathbf{1}_{N}\mathbf{M}_{1}\Vert^{2}+\Vert\mathbf{M}_{3}\Vert^{2}+0.09\Vert\mathbf{M}_{4}\Vert^{2}+1.38\Vert\mathbf{M}_{1}\Vert^{2}+335\big\Vert\mathbf{M}_{2}\big\Vert$. When $\mathbf{T}^t$ is unknown \textit{a-priori} and $\mathbf{C}^t$ is designed using~\eqref{C_r_definition}, these parameters lead to different $\Phi$ curves and $J^*$ values between the two models in Fig.~\ref{fig:impact_of_different_models}. Although the ResNet-18 model has a larger packet size and hence a higher PER than the CNN model, $J^*$ is larger for the ResNet-18 model. This is different from the observations in Figs.~\ref{fig:node_failure} and~\ref{fig:Fed-cifar100} that a larger packet size (or a higher PER) results in a smaller $J^*$, when the same model (e.g., either the CNN or ResNet-18) is considered.

\begin{figure}\vspace{-3mm }
	\centering
	\subfigure[$\Phi(\mathbf{M}_1,\mathbf{M}_2,\mathbf{M}_3,\mathbf{M}_4)$ vs. $J$.]
	{
		\begin{minipage}[t]{0.375\textwidth}
			\centering
					\includegraphics[width=7cm]{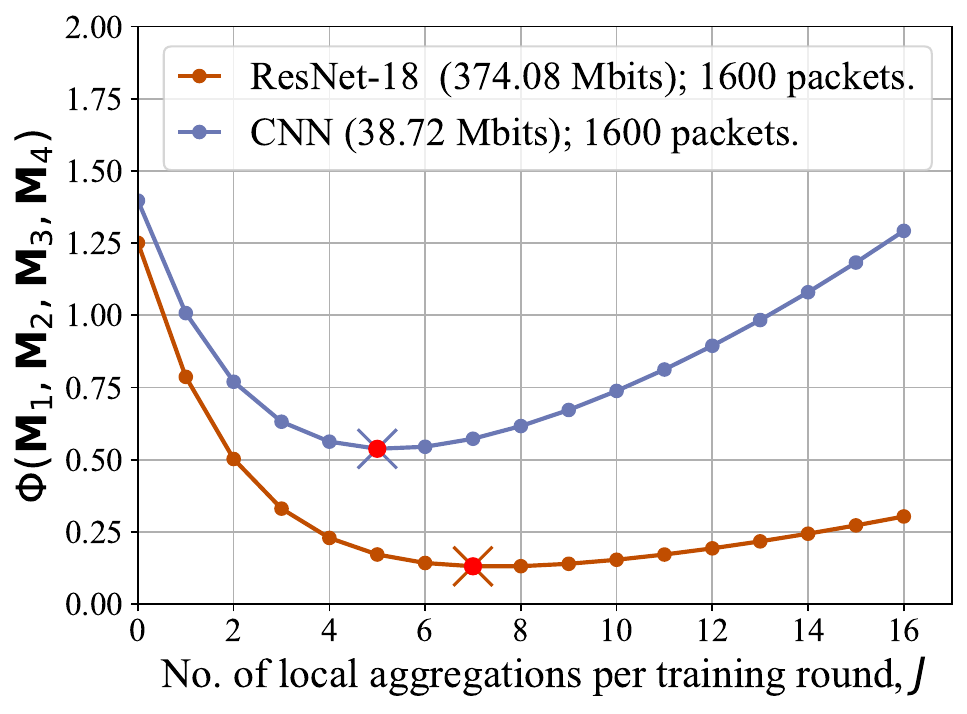}
		\end{minipage}
	}
	\subfigure[Training accuracy vs. $J$.]
	{
		\begin{minipage}[t]{0.375\textwidth}
			\centering
			\includegraphics[width=7cm]{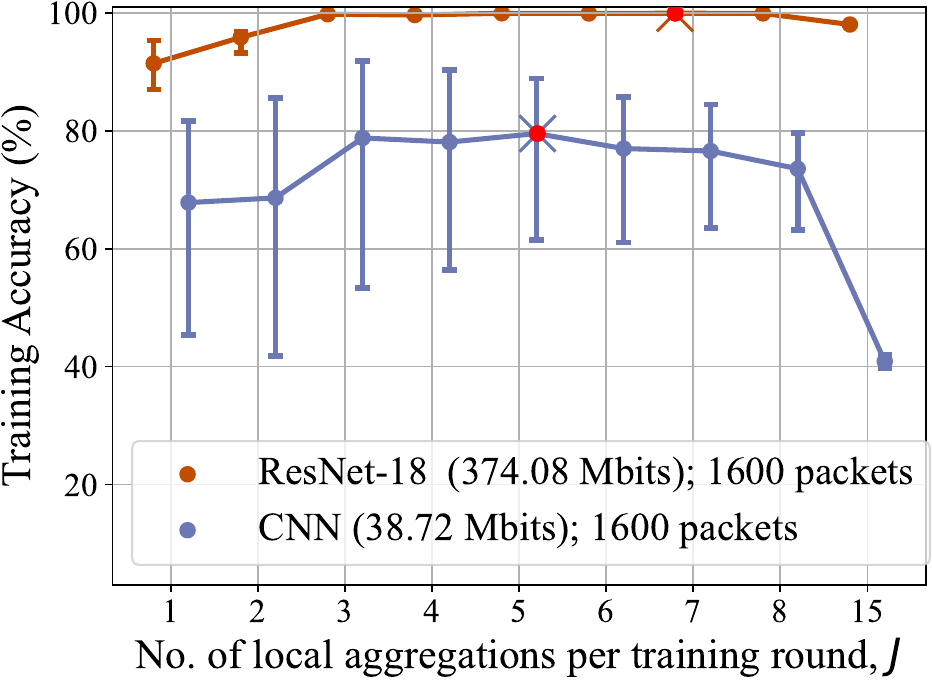}
		\end{minipage}
	}\vspace{-3mm }
 \caption{ Comparison between the CNN and ResNet-18 models, where $\eta=0.03$ for the CNN and $\eta=0.1$ for the ResNet-18. $I=5$. The number of packets is $1600$ per model. $\mathbf{C}^t$ is given by \eqref{C_r_definition}.}\label{fig:impact_of_different_models}\vspace{-2mm }
\end{figure}

 \section{Conclusion}\label{section:con}
This paper has analyzed the impact of imperfect channels on D-FL, and subsequently determined the optimal number of local aggregations. It has revealed that increased local iterations and the \textit{a-priori} knowledge of channels can help alleviate the impact of imperfect channels so that more local aggregations can lead to better convergence. Experiments on the CNN and ResNet-18 models have validated our convergence analysis. It has been shown that D-FL with the optimal number of local aggregations outperforms C-FL and D-FL without optimizing local aggregations by 12.5\% and 10\% in training accuracy, respectively, when the channel conditions are unknown. In the future, we plan to investigate the impact of mobility and non-i.i.d. errors in communication links on the convergence and reliability of D-FL. We also plan to study model selection for D-FL, where each device can adaptively select the models for aggregation under fast-changing fading channels.

\appendix
\subsection{Proof of \textbf{Lemma \ref{theo1}}}\label{appendix:theorem proof}
\begin{lemma}
    \label{minibatch_gradient}
    The variance between the virtually aggregated full-batch gradient and the virtually aggregated stochastic gradient is bounded by
   \begin{small}%
    \begin{align}%
  \notag & \mathbb{E}_{\left\{ \xi_{n,i-1}^{t},{\forall n}\right\}}\Big(\Big\Vert\!{\sum}_{\forall n}p_{n}(\eta\mathbf{g}_{n,i\!-\!1}^{t}\!\!-\!\!\eta\mathbf{g}_{\xi,n,i\!-\!1}^{t})\Big\Vert^{2}\Big)\!\\\notag
  &\;=\!{\sum}_{\forall n}p_{n}^{2}\mathbb{E}_{\xi_{n,i-1}^{t}}\big(\big\Vert\eta\mathbf{g}_{n,i\!-\!1}^{t}\!\!-\!\!\eta\mathbf{g}_{\xi,n,i\!-\!1}^{t}\big\Vert^{2}\big)\\
\notag  & \;+\!\!2\!\!\!\sum_{n\neq m}\!\!\!\big\langle p_{n}\!\eta(\mathbf{g}_{n,i\!-\!1}^{t}\!\!-\!\!\mathbb{E}_{\xi_{n,i\!-\!1}^{t}}\!\!(\mathbf{g}_{\xi,n,i\!-\!1}^{t})\!),p_{m}\!\eta(\mathbf{g}_{m,i\!-\!1}^{t}\!\!-\!\!\mathbb{E}_{\xi_{m,i\!-\!1}^{t}}\!\!(\mathbf{g}_{\xi,m,i\!-\!1}^{t}\!)\!)\!\big\rangle  \! \\
  &     =\!\! {\sum}_{\forall n}\!p_{n}^{2}\!\mathbb{E}_{\xi_{n,i-1}^{t}}\big(\!\big\Vert\eta\mathbf{g}_{n,i\!-\!1}^{t}\!\!-\!\!\eta\mathbf{g}_{\xi,n,i\!-\!1}^{t}\big\Vert^{2}\big)\!\!\leq\!\! {\eta^2}{\sum}_{\forall n}\!p_{n}^{2}\!\sigma^2_n, \!\!\label{a_bound:SGD} 
  \end{align}  
\end{small}
where \eqref{a_bound:SGD} is based on \textbf{Assumption \ref{assumption}-4)}.
\end{lemma}	
 \begin{lemma}\label{appedix_lemma3}
		 The weighted distance between $\bar{\boldsymbol{\omega}}^{t}_i$ and $\boldsymbol{\omega}_{n,i}^{t}, \forall i,n$ is bounded by
	\begin{small}%
	    \begin{align}%
	\notag	&\!\!\!\mathbb{E}_{\{ \xi_{n,0}^{t},\cdots,\xi_{n,i\!-\!1}^{t},\forall n\}}\!\!\left({\sum}_{\forall n}\!p_{n}\!\left\Vert \bar{\boldsymbol{\omega}}_{i}^{t}\!-\!\boldsymbol{\omega}_{n,i}^{t}\right\Vert ^{2}\!\right)\!\!\leq\!\!\left[{(1+\eta)^{i\!+\!1}\!\!-\!\!(1\!+\!\eta)}\right]G^{2}\!\\&\quad\quad+\!(1\!+\!\eta)^{i}{\sum}_{\forall n}\!p_{n}\!\Vert\boldsymbol{\varpi}_{n}^{t\!-\!1}\!\!\!-\!\!\!{\sum}_{\forall m}\!p_{m}\!\boldsymbol{\varpi}_{m}^{t\!-\!1}\!\Vert^{2}.\label{var}%
		\end{align}%
	\end{small}%
	\end{lemma}%
	\begin{proof}%
	    See \textbf{\textbf{Appendix \ref{proof:lemma5}}}.
	\end{proof}
\begin{lemma}\label{appedix_lemma2}Considering the randomness of the $i$-th iteration, the upper bound of $\mathbb{E}_{\left\{ \xi_{n,i-1}^{t},{\forall n}\right\}}(\Delta\boldsymbol{\omega}^{t}_{i})$ is given by
\begin{subequations}\label{delta_i}%
\small
    \begin{align}
&\!\!\!\!\text{When }i=1,\notag\\
&\!\!\!\! \notag \mathbb{E}_{\{ \xi_{n,0}^{t},{\forall n}\}}\! ( \Delta\boldsymbol{\omega}^{t}_{1})\!\!\!\leq\!(1\!+\!\tau_{\epsilon}\!)\!\left(\!1\!-\!\frac{\mu\eta}{2}\right)\!\Delta\boldsymbol{\omega}_{I}^{t\!-\!1}\!\!+\!\!(1\!+\!\tau_{\epsilon}\!)\eta^{2}{\sum}_{\forall n}\!p_{n}^{2}\!\sigma_{n}^{2}\\&\notag+(1\!+\!\tau_{\epsilon}\!)\left(\!2\eta^{2}L^{2}\!+\!(L\!+\!\mu)\eta\right){\sum}_{\forall n}p_{n}\left(\left\Vert \bar{\boldsymbol{\omega}}_{I}^{t\!-\!1}\!-\!\mathbf{x}_{n,J}^{t\!-\!1}\right\Vert ^{2}\!\right)\\&\!+\!\!(1\!\!+\!\!\frac{1}{\tau_{\epsilon}}\!)(\!1\!\!+\!\!\tau_{\eta}\!)\!\left[\!\big\Vert\!\sum_{\forall n}\!p_{n}\!\boldsymbol{\varpi}_{n}^{t\!-\!1}\!\big\Vert^{2}\!\!+\!\!\frac{\eta^{2}\!L^{2}}{\tau_{\eta}}\!\!\sum_{\forall n}p_{n}\big\Vert\mathbf{w}_{n,J}^{t\!-\!1}\!\!-\!\!\mathbf{x}_{n,J}^{t\!-\!1}\!\big\Vert^{2}\!\right]\!.\!\!\label{delta_1}\\
\notag&\!\!\!\!\text{When } i=2,3,\cdots,I,\\\notag
&\!\!\!\!\mathbb{E}_{\left\{ \xi_{n,i-1}^{t},{\forall n}\right\}}\! ( \Delta\boldsymbol{\omega}^{t}_{i})\!\!\leq\!\left(\!1\!-\!\frac{\mu\eta}{2}\right)\!\Delta\boldsymbol{\omega}_{i\!-\!1}^{t}+\!{\sum}_{\forall n}{\!\eta^{2}}p_{n}^{2}\sigma_{n}^{2}\\&+(\!2\eta^{2}L^{2}\!\!+\!(L\!+\!\mu)\eta){\sum}_{\forall n}\!p_{n}\!\left\Vert \bar{\boldsymbol{\omega}}_{i\!-\!1}^{t}\!\!-\!\!\boldsymbol{\omega}_{n,i\!-\!1}^{t}\right\Vert ^{2}.\label{delta_ib}
\end{align}
\end{subequations}
	\end{lemma}
 \begin{proof}
     See \textbf{Appendix \ref{proof:lemma6}}.
 \end{proof}
	By substituting \eqref{var} and \eqref{delta_i} for $i=I-1,I-2, \cdots,0$ into $\mathbb{E}_{\{ \xi_{n,0}^{t},\cdots,\xi_{n,I\!-\!1}^{t},\forall n\}}(\Delta\boldsymbol{\omega}_{I}^{t})$, it follows from the randomness of every iteration of the $t$-th training round, that\begin{small}
	    \begin{align}
	\notag	&\!\!\!{\mathbb{E}_{\{ \xi_{n,0}^{t},\cdots,\xi_{n,I\!-\!1}^{t},\forall n\}}}(\Delta\boldsymbol{\omega}_{I}^{t})\!\leq \!\left(\!1\!-\!\frac{\mu\eta}{2}\right){\mathbb{E}_{\{ \xi_{n,0}^{t},\cdots,\xi_{n,I\!-\!2}^{t},\forall n\}}}(\Delta\boldsymbol{\omega}^{t}_{I\!-\!1})\!\!\\
&\notag+\!(\!2\eta^{2}L^{2}\!\!+\!(L\!\!+\!\!\mu)\eta){\mathbb{E}_{\{ \xi_{n,0}^{t},\cdots,\xi_{n,I\!-\!2}^{t},\forall n\}}}\!\!\left({\sum}_{\forall n}\!p_{n}\!\!\left\Vert \bar{\boldsymbol{\omega}}_{I\!-\!1}^{t}\!\!-\!\!\boldsymbol{\omega}_{n,I\!-\!1}^{t}\right\Vert ^{2}\right)\\
&\quad+\!\!\eta^2
{\sum}_{\forall n}p_{n}^{2}\sigma^2_n\notag
  \\\notag\leq&\!
   (1\!\!+\!\!\tau_{\epsilon})\!\left(1\!\!-\!\!\frac{\mu\eta}{2}\right)^{I}\!\Delta\boldsymbol{\omega}_{I}^{t\!-\!1}\!\!+\!\!\frac{1\!\!+\!\!(\!\frac{(1+\tau_{\epsilon})\mu\eta}{2}\!-\!1)(1\!-\!\frac{\mu\eta}{2})^{I\!-\!1}}{\mu}\sum_{\forall n}\!2\eta p_{n}^{2}\sigma_{n}^{2}
  \\\notag&
    \!\!\!+\!\!(1\!\!+\!\!\frac{1}{\tau_{\epsilon}})\!(\!1\!\!+\!\!\tau_{\eta}\!)\!\!\left(\!1\!\!-\!\!\frac{\mu\eta}{2}\right)^{I\!-\!1}\!\!\Big(\!\big\Vert\!\!\sum_{\forall n}\!p_{n}\!\boldsymbol{\varpi}_{n}^{t\!-\!1}\!\big\Vert^{2}\!\!\!+\!\!\frac{\eta^2 \!L^{2}}{\tau_{\eta}}\!\!\!\!\sum_{\forall n}\!p_{n}\!\big\Vert\mathbf{w}_{n,J}^{t\!-\!1}\!\!-\!\!\mathbf{x}_{n,J}^{t\!-\!1}\!\big\Vert^{2}\!\Big)
    \\\notag
 & \!\!\!\!+\!\! (1\!\!+\!\!\tau_{\epsilon}\!)\!\big(1\!\!-\!\!\frac{\mu\eta}{2}\!\big)^{I-1}\!\big(\!2\eta^{2}L^{2}\!\!+\!\!(\!L\!\!+\!\!\mu\!)\eta\big)\!\!\sum_{\forall n}\!\!p_{n}\!\big(\!\left\Vert \bar{\boldsymbol{\omega}}_{I}^{t\!-\!1}\!-\!\mathbf{x}_{n,J}^{t\!-\!1}\right\Vert ^{2}\!\big)\\\notag
		&\!\!\!+\!\!(2\eta L^{2}\!\!+\!\!(L\!\!+\!\!\mu))(1\!\!+\!\!\eta)^{3}\frac{\!(\!1\!\!+\!\!\eta\!)^{I\!-\!1}\!\!-\!\!(\!1\!\!-\!\!\frac{\mu\eta}{2}\!)^{I\!-\!1}\!}{1\!\!+\!\!\frac{\mu}{2}}\times\\
		& \quad\quad\quad \Big(\!{\sum}_{\forall n}\!p_{n}\!\Vert\boldsymbol{\varpi}_{n}^{t\!-\!1}\!\!\!-\!\!\!{\sum}_{\forall m}\!p_{m}\!\boldsymbol{\varpi}_{m}^{t-\!1}\!\Vert^{2}\!\Big)\!+\!\zeta_{2}G^2, \label{eq:appendix_lemma1}%
		\end{align}
	\end{small}%
           which can be reorganized into
   \eqref{eq:new_expectation}. 

The third term on the RHS of \eqref{eq:new_expectation} can be rewritten as
\begin{subequations}\small\label{eq:new_expectation_third_part}
    \begin{align}
   &\!\!\! \! \mathbb{E}_{\mathbf{e}}\big(\big\Vert{\sum}_{\forall n}p_{n}\boldsymbol{\varpi}_{n}^{t-1}\big\Vert^{2}\big)\notag\\
   &\!\!\! =\!\!\big\Vert\mathbb{E}_{\mathbf{e}}(\sum_{\forall n}\!p_{n}\!\boldsymbol{\varpi}_{n}^{t\!-\!1})\!\big\Vert^{2}\!\!\!+\!\!\mathbb{E}_{\mathbf{e}}\!\big(\big\Vert\!\sum_{\forall n}p_{n}\boldsymbol{\varpi}_{n}^{t\!-\!1}\!\!-\!\!\mathbb{E}_{\mathbf{e}}(\sum_{\forall n}p_{n}\boldsymbol{\varpi}_{n}^{t\!-\!1})\!\big\Vert^{2}\!\big)\label{eq:new_expectation_third_part_a}\\
    &\!\!\! =\!\!\big\Vert\mathbb{E}_{\mathbf{e}}({\sum}_{\forall n}p_{n}\boldsymbol{\varpi}_{n}^{t\!-\!1})\big\Vert^{2}\!\!+\!\!{\sum}_{\forall n}p_{n}^{2}\mathbb{E}_{\mathbf{e}}\big(\Vert\mathbf{w}_{n,J}^{t\!-\!1}\!\!-\!\mathbb{E}_{\mathbf{e}}(\mathbf{w}_{n,J}^{t\!-\!1})\Vert^{2}\big)\label{eq:new_expectation_third_part_b}\\
    &\!\! \!=\!\!\Big\Vert\mathbf{p}\mathbb{E}_{\mathbf{e}}(\boldsymbol{\Pi}^{t-1})\Big\Vert^{2}\!\!+\!\!\mathbf{p}\mathrm{diag}(\mathbf{p})\mathbb{D}_{\mathbf{e}}(\mathbf{W}_{J}^{t-1})\mathbf{1}_M^{\dagger},\label{eq:new_expectation_third_part_c}
\end{align}
\end{subequations}
where \eqref{eq:new_expectation_third_part_a} follows from $\mathbb{E}\big(\left\Vert X\right\Vert ^{2}\big)=\left\Vert \mathbb{E}(X)\right\Vert ^{2}+\mathbb{E}\big(\left\Vert X-\mathbb{E}(X)\right\Vert ^{2}\big)$ with $X=\sum_{\forall n}p_{n}\boldsymbol{\varpi}_{n}^{t-1}$; \eqref{eq:new_expectation_third_part_b} is obtained by substituting $\boldsymbol{\varpi}^{t-1}_{n}=\boldsymbol{\bar{\omega}}_I^{t-1}-\mathbf{w}_{n,J}^{t-1}$
into the second term of \eqref{eq:new_expectation_third_part_a} and 
noting that $\boldsymbol{\bar{\omega}}_I^{t-1}$ is given and has no randomness here;
\eqref{eq:new_expectation_third_part_c} is obtained since $\mathbb{E}_{\mathbf{e}}(\boldsymbol{\Pi}^{t\!-\!1})=\big[\mathbb{E}_{\mathbf{e}}(\boldsymbol{\varpi}^{t-1}_{1}),\cdots,\mathbb{E}_{\mathbf{e}}(\boldsymbol{\varpi}^{t-1}_{N})\big]^{\dagger}$
 and $\mathbb{D}_{\mathbf{e}}(\mathbf{W}_{J}^{t\!-\!1})\mathbf{1}_M^{\dagger}=\big[\mathbb{E}_{\mathbf{e}}\big(\Vert\mathbf{w}_{1,J}^{t\!-\!1}-\mathbb{E}_{\mathbf{e}}(\mathbf{w}_{1,J}^{t-1})\Vert^{2}\big),\cdots,\mathbb{E}_{\mathbf{e}}\big(\Vert\mathbf{w}_{N,J}^{t\!-\!1}-\mathbb{E}_{\mathbf{e}}(\mathbf{w}_{N,J}^{t-1})\Vert^{2}\big)\big]^{\dagger}$. 

The fourth term on the RHS of \eqref{eq:new_expectation} is upper bounded by
\begin{subequations}\small
    \begin{align}
        &\!\!{\sum}_{\forall n}p_{n}\mathbb{E}_{\mathbf{e}}\big(\Vert\boldsymbol{\varpi}_{n}^{t-1}-{\sum}_{\forall m}p_{m}\boldsymbol{\varpi}_{m}^{t-1}\Vert^{2}\big)\nonumber %\label{eq:new_var_fourth_part_a}\nonumber
        \\&={\sum}_{\forall n}p_{n}\mathbb{E}_{\mathbf{e}}\big(\Vert\boldsymbol{\varpi}_{n}^{t-1}\Vert^{2}\big)-\mathbb{E}_{\mathbf{e}}\big(\big\Vert{\sum}_{\forall n}p_{n}\boldsymbol{\varpi}_{n}^{t-1}\big\Vert^{2}\big)\label{eq:new_var_fourth_part_b}
        \\\notag&\!\!=\!\!\left\Vert \mathrm{diag}(\sqrt{\mathbf{p}})\mathbb{E}_{\mathbf{e}}(\boldsymbol{\Pi}^{t\!-\!1})\right\Vert ^{2}-\Vert\mathbf{p}\mathbb{E}_{\mathbf{e}}(\boldsymbol{\Pi}^{t-1})\Vert^{2}\!\!\\&\quad\quad+(\mathbf{p}-\mathbf{p}\mathrm{diag}(\mathbf{p}))\mathbb{D}_{\mathbf{e}}(\mathbf{W}_{J}^{t\!-\!1})\mathbf{1}_M^{\dagger}\label{eq:new_var_fourth_part_c}\\&\!\!
        \leq\!\!\left\Vert \mathrm{diag}(\sqrt{\mathbf{p}}\!\!-\!\!\mathbf{p})\mathbb{E}_{\mathbf{e}}(\boldsymbol{\Pi}^{t\!-\!1})\right\Vert ^{2}\!\!+\!\!(\mathbf{p}\!\!-\!\!\mathbf{p}\mathrm{diag}(\mathbf{p}))\mathbb{D}_{\mathbf{e}}(\mathbf{W}_{J}^{t\!-\!1})\mathbf{1}_M^{\dagger},\label{eq:new_var_fourth_part_d}
    \end{align}
\end{subequations}
where 
\eqref{eq:new_var_fourth_part_b} follows from \eqref{eq:new_expectation_third_part_a};
\eqref{eq:new_var_fourth_part_c} is based on 
\begin{subequations}\small\label{eq:new_expectation_fourth_part}
    \begin{align}
 &\! {\sum}_{\forall n}p_{n}  \mathbb{E}_{\mathbf{e}}\big(\Vert\boldsymbol{\varpi}_{n}^{t-1}\Vert^{2}\big)\notag\\&=\!{\sum}_{\forall n}\!p_{n}\!\left\Vert \!\mathbb{E}_{\mathbf{e}}\!(\boldsymbol{\varpi}_{n}^{t\!-\!1})\right\Vert ^{2}\!+\!{\sum}_{\forall n}\!p_{n}\!\mathbb{E}_{\mathbf{e}}\Vert\boldsymbol{\varpi}_{n}^{t\!-\!1}\!-\!\mathbb{E}_{\mathbf{e}}(\boldsymbol{\varpi}_{n}^{t\!-\!1})\Vert^{2}\label{eq:new_expectation_fourth_part_a}\\
&={\sum}_{\forall n}\!p_{n}\!\left\Vert\! \mathbb{E}_{\mathbf{e}}(\boldsymbol{\varpi}_{n}^{t\!-\!1})\right\Vert ^{2}\!\!+\!\!{\sum}_{\forall n}\!p_{n}\!\mathbb{E}_{\mathbf{e}}\Vert\mathbf{w}_{n,J}^{t\!-\!1}\!-\!\mathbb{E}_{\mathbf{e}}\!(\mathbf{w}_{n,J}^{t\!-\!1})\!\Vert^{2}\label{eq:new_expectation_fourth_part_b}\\
&=\left\Vert \mathrm{diag}(\sqrt{\mathbf{p}})\mathbb{E}_{\mathbf{e}}(\boldsymbol{\Pi}^{t\!-\!1})\right\Vert ^{2}\!+\!\!\mathbf{p}\mathbb{D}_{\mathbf{e}}(\mathbf{W}_{J}^{t-1})\mathbf{1}_M^{\dagger};\label{eq:new_expectation_fourth_part_c}%
\end{align}%
\end{subequations}%
and \eqref{eq:new_var_fourth_part_d} follows from
\begin{subequations}
    \begin{align}
        \notag&\left\Vert \mathrm{diag}(\sqrt{\mathbf{p}})\mathbb{E}_{\mathbf{e}}(\boldsymbol{\Pi}^{t\!-\!1})\right\Vert ^{2}-\Vert\mathbf{p}\mathbb{E}_{\mathbf{e}}(\boldsymbol{\Pi}^{t-1})\Vert^{2}\\
        \leq&\left\Vert \mathrm{diag}(\sqrt{\mathbf{p}})\mathbb{E}_{\mathbf{e}}(\boldsymbol{\Pi}^{t\!-\!1})\right\Vert ^{2}-\Vert\mathrm{diag}(\mathbf{p})\mathbb{E}_{\mathbf{e}}(\boldsymbol{\Pi}^{t-1})\Vert^{2}\label{eq:new_var_fourth_part_e}\\
        \leq&\left\Vert \mathrm{diag}(\sqrt{\mathbf{p}}-\mathbf{p})\mathbb{E}_{\mathbf{e}}(\boldsymbol{\Pi}^{t\!-\!1})\right\Vert ^{2}\label{eq:new_var_fourth_part_f}
    \end{align}
\end{subequations}
where \eqref{eq:new_var_fourth_part_e} is due to $\left\Vert A\right\Vert ^{2}\leq\left\Vert \mathbf{1}_{N}A\right\Vert ^{2}$ with $A=\mathrm{diag}({\mathbf{p}})\mathbb{E}_{\mathbf{e}}(\boldsymbol{\Pi}^{t\!-\!1})$ and $\mathbf{1}_{N}\mathrm{diag}({\mathbf{p}})\mathbb{E}_{\mathbf{e}}(\boldsymbol{\Pi}^{t\!-\!1})={\mathbf{p}}\mathbb{E}_{\mathbf{e}}(\boldsymbol{\Pi}^{t\!-\!1})$, and \eqref{eq:new_var_fourth_part_f} is due to $\left\Vert A\right\Vert ^{2}-\left\Vert B\right\Vert ^{2}\leq \left\Vert A-B\right\Vert ^{2}$.

The fifth term on the RHS of \eqref{eq:new_expectation} can be rewritten as
\begin{subequations}\small\label{eq:new_expectation_fifth_part}
    \begin{align}
 &\! {\sum}_{\forall n}p_{n}  \mathbb{E}_{\mathbf{e}}\big(\Vert\mathbf{w}_{n,J}^{t\!-\!1}\!\!-\!\!\mathbf{x}_{n,J}^{t\!-\!1}\Vert^{2}\big)\notag\\
&\!\!=\!\!{\sum}_{\forall n}\!\!p_{n}\!\!\left\Vert\! \mathbb{E}_{\mathbf{e}}(\mathbf{w}_{n,J}^{t\!-\!1}\!\!-\!\!\mathbf{x}_{n,J}^{t\!-\!1})\right\Vert ^{2}\!\!+\!\!{\sum}_{\forall n}\!\!p_{n}\!\mathbb{E}_{\mathbf{e}}\Vert\mathbf{w}_{n,J}^{t\!-\!1}\!\!-\!\!\mathbb{E}_{\mathbf{e}}\!(\mathbf{w}_{n,J}^{t\!-\!1})\!\Vert^{2}\label{eq:new_expectation_fifth_part_a}\\
&\!\!=\!\!\left\Vert \mathrm{diag}(\sqrt{\mathbf{p}})\mathbb{E}_{\mathbf{e}}\!\!\left(\mathbf{W}_{J}^{t\!-\!1}-\mathbf{X}_{J}^{t\!-\!1}\right)\!\right\Vert ^{2}\!\!\!+\!\!\mathbf{p}\mathbb{D}_{\mathbf{e}}(\mathbf{W}_{J}^{t\!-\!1})\mathbf{1}_M^{\dagger}.\!\label{eq:new_expectation_fifth_part_b}%
\end{align}%
\end{subequations}%
The seventh term on the RHS of \eqref{eq:new_expectation} is rewritten as
    {\footnotesize\begin{align}
 &\!\!\!\!\!\!\! \sum_{\forall n}\!p_{n}\!  \Vert\bar{\boldsymbol{\omega}}_{I}^{t\!-\!1}\!\!-\!\!\mathbf{x}_{n,J}^{t\!-\!1}\Vert^{2}\!\!=\!\!\Big\Vert \mathrm{diag}(\!\sqrt{\mathbf{p}})\!\Big(\!\frac{\mathbf{1}_N^{\dagger}\mathbf{1}_N}{N}N\mathrm{diag}(\mathbf{p})\boldsymbol{\Omega}_I^{t\!-\!1}\!\!-\!\!{\mathbf{X}}_{J}^{t\!-\!1}\Big)\!\Big\Vert ^{2}\!\!.\!\!\label{eq:new_expectation_six_part_a}
\end{align}}
Substituting \eqref{eq:new_expectation_third_part_c},  \eqref{eq:new_var_fourth_part_d}, \eqref{eq:new_expectation_fifth_part_b} and \eqref{eq:new_expectation_six_part_a} yields \eqref{eq:new_expectationb}.

\subsection{Proof of \textbf{Lemma \ref{lemma1}}}\label{proof of lemma1}
	The expectation of $\mathbf{e}_{n,m,j}^{t}\!\circ\!{\mathbf{w}}_{m,j\!-\!1}^{t}$ in \eqref{imperfect_wic} is rewritten as
\begin{subequations}
	\begin{align}
		\mathbb{E}_{\mathbf{e}}(\mathbf{e}_{n,m,j}^{t}\!\circ\!{\mathbf{w}}_{m,j-1}^{t})&=\!\!\mathbb{E}_{\mathbf{e}}(\mathbf{e}_{n,m,j}^{t})\circ\mathbb{E}_{\mathbf{e}}({\mathbf{w}}_{m,j-1}^{t})\label{E:ca}\\
		&=\!\!\left((1\!\!-\!\!\epsilon_{\mathrm{P},n,m}^{t})\mathbf{1}_{  M}\right)\!\circ\!\mathbb{E}_{\mathbf{e}}({\mathbf{w}}_{m,j-1}^{t})\label{E:cb}\\
		&=\!\!(1-\epsilon_{\mathrm{P},n,m}^{t})\mathbb{E}_{\mathbf{e}}({\mathbf{w}}_{m,j-1}^{t}),\label{E:c}
	\end{align}%
	\end{subequations}
	where \eqref{E:ca} is due to the fact that $\mathbf{e}_{n,m,j}^{t}$ and ${\mathbf{w}}_{m,j-1}^{t}$ are independent; \eqref{E:cb} is because $\Pr(\mathbf{e}^{t}_{n,m,j}[l]=\mathbf{1}_{ L_{\mathrm{p}}})=1-\epsilon_{\mathrm{P},n,m}^{t}$, see \eqref{error1}; \eqref{E:c} is from $\mathbf{1}_ M\circ{\mathbf{w}}_{m,j-1}^{t}={\mathbf{w}}_{m,j-1}^{t}$.	
	% Based on \eqref{E:c}, 
 Hence, the expectation of the locally aggregated model in \eqref{imperfect_wic} is 
            \begin{align}
		\notag\mathbb{E}_{\mathbf{e}}({\mathbf{w}}_{n,j}^{t})=&{\sum}_{\forall m}c_{n,m}^{t}\mathbb{E}_{\mathbf{e}}(\mathbf{e}_{n,m,j}^{t}\circ{\mathbf{w}}_{m,j-1}^{t})
		\\=&{\sum}_{\forall m}c_{n,m}^{t}(1-\epsilon_{\mathrm{P},n,m}^{t})\mathbb{E}_{\mathbf{e}}({\mathbf{w}}_{m,j-1}^{t}),\label{Expectation_w_nj}
	\end{align}%
which can be rewritten in a matrix form, as given by
	\begin{align}
	\mathbb{E}_{\mathbf{e}}({\mathbf{W}}_{j}^{t})=(\mathbf{C}^{t}\circ\mathbf{T}^{t})\mathbb{E}_{\mathbf{e}}({\mathbf{W}}_{j-1}^{t}).\label{eq:wrc}
	\end{align}
	Let ${\mathbf{W}}^{t}_{J}$ collect the locally aggregated models of all devices (after $J$ local aggregations) at the end of the $t$-th training round,
 with $\mathbf{W}^{t}_{0}$ collecting the initial models for aggregation of the $t$-th round.
	The expectation of ${\mathbf{W}}^{t}_{J}$ is given by
        \begin{align}
            \mathbb{E}_{\mathbf{e}} ({\mathbf{W}}^{t}_{J}) &=(\mathbf{C}^{t}\circ\mathbf{T}^{t})^J\mathbf{W}^{t}_{0},\label{eq:wrC b}
        \end{align}%
where \eqref{eq:wrC b} is obtained by substituting \eqref{eq:wrc} recursively $J$ times. Combining \eqref{noise_consensus:c} and \eqref{eq:wrC b}, we obtain~\eqref{Expectation_consensus}.

\vspace{-3mm }
\subsection{Proof of \textbf{Lemma \ref{lemma_variance}}}\label{proof_of_lemma_variance}
From \eqref{imperfect_wic} and \eqref{Expectation_w_nj}, $\mathbf{e}^{t}_{n,m,j}$ and $\mathbf{w}_{m,j-1}^{t}$ are independent random variables. $\mathbf{e}^{t}_{n,m,j}$ obeys the Bernoulli distribution. $\mathbb{D}_{\mathbf{e}}(\mathbf{e}^{t}_{n,m,j})\!=\!(1\!-\!\epsilon_{\mathrm{P},n,m}^{t})\epsilon_{\mathrm{P},n,m}^{t}$; $\mathbb{E}_{\mathbf{e}}(\mathbf{e}^{t}_{n,m,j})\!=\!1\!-\!\epsilon_{\mathrm{P},n,m}^{t}$.~So,
 \begin{align}
     \notag \mathbb{D}_{\mathbf{e}}(\mathbf{w}_{n,j}^{t})={\sum}_{\forall m}\mathbb{D}_{\mathbf{e}}(c_{n,m}^{t}(\mathbf{e}_{n,m,j}^{t}\circ{\mathbf{w}}_{m,j-1}^{t}))\\
     ={\sum}_{\forall m}(c_{n,m}^{t})^{2}\mathbb{D}_{\mathbf{e}}(\mathbf{e}_{n,m,j}^{t}\circ{\mathbf{w}}_{m,j-1}^{t}).
 \end{align} 
For any two random variables $X$ and $Y$, the variance of $XY$~is  
$\mathbb{D}(XY)=\left(\mathbb{D}(X)+\mathbb{E}^2(X)\right)\mathbb{D}(Y)+\mathbb{D}(X)\mathbb{E}^{2}(Y)$.
By replacing $X$ with the elements of $\mathbf{e}_{n,m,j}^{t}$ and $Y$ with the elements of $\mathbf{w}^{t}_{m,j}$, it readily follows that
     \begin{align}
    \notag\mathbb{D}_{\mathbf{e}}&(\mathbf{w}_{n,j}^{t})
   \!\!=\!\!{\sum}_{\forall m}\!(c_{n,m}^{t})^{2}\!(1\!-\!\epsilon_{\mathrm{P},n,m}^{t})\\&\!\!\!\!\!\!\times\left(\mathbb{D}_{\mathbf{e}}(\mathbf{w}_{m,j\!-\!1}^{t})\!+\!\epsilon_{\mathrm{P},n,m}^{t}\mathbb{E}_{\mathbf{e}}(\mathbf{w}_{m,j\!-\!1}^{t})\!\circ\!\mathbb{E}_{\mathbf{e}}(\mathbf{w}_{m,j\!-\!1}^{t})\!\right).\label{Variance_wnj}
\end{align}

Note that $(c_{n,m}^{t})^{2}\!(1\!-\!\epsilon_{\mathrm{P},n,m}^{t})$ and $(c_{n,m}^{t})^{2}\!(1\!-\!\epsilon_{\mathrm{P},n,m}^{t})\epsilon_{\mathrm{P},n,m}^{t}$ are the $n$-th row and the $m$-th columns of the matrices $(\mathbf{C}^{t}\circ \mathbf{C}^{t}\circ \mathbf{T}^{t})$ and $(\mathbf{C}^{t}\circ \mathbf{C}^{t}\circ \mathbf{T}^{t}\circ ({\mathbf{1}_N^{\dagger}\!\mathbf{1}_N}\!-\!\mathbf{T}^{t})$, respectively. The matrix form of the variance is obtained as
\begin{small}
    \begin{align}
 \notag    &\mathbb{D}_{\mathbf{e}}(\mathbf{W}_j^{t})=(\mathbf{C}^{t}\circ \mathbf{C}^{t}\circ \mathbf{T}^{t})\mathbb{D}_{\mathbf{e}}(\mathbf{W}_{j-1}^{t})\\&+\big(\mathbf{C}^{t}\circ\!\!\mathbf{C}^{t}\!\!\circ\!\mathbf{T}^{t}\!\!\circ\!\!({\mathbf{1}_{N}^{\dagger}\!\mathbf{1}_{N}}\!-\!\mathbf{T}^{t})\big)\mathbb{E}_{\mathbf{e}}\!(\mathbf{W}_{j\!-\!1}^{t})\!\circ\!\mathbb{E}_{\mathbf{e}}\!(\mathbf{W}_{j\!-\!1}^{t})\notag\\&
=\mathbf{A}_{1}^{t}\mathbb{D}_{\mathbf{e}}(\mathbf{W}_{j\!-\!1}^{t})+\!\mathbf{A}_{2}^{t}\mathbf{A}_{3}^{t}(j)(\mathbf{W}_{0}^{t}\!\circ\!\mathbf{W}_{0}^{t}),\label{D_j_1}%
\end{align}%
\end{small}%
       where $\mathbb{E}_{\mathbf{e}}({\mathbf{W}}^{t}_{j-1})=(\mathbf{C}^{t}\circ\mathbf{T}^{t})^{j-1}\mathbf{W}^{t}_{0}$ from \eqref{eq:wrC b}, $\mathbb{E}_{\mathbf{e}}\!(\mathbf{W}_{j-1}^{t})\!\circ\!\mathbb{E}_{\mathbf{e}}\!(\mathbf{W}_{j\!-\!1}^{t})=\left((\mathbf{C}^{t}\circ\mathbf{T}^{t})^{j\!-\!1}\circ (\mathbf{C}^{t}\circ\mathbf{T}^{t})^{j\!-\!1}\right)(\mathbf{W}_{0}^{t}\!\circ\!\mathbf{W}_{0}^{t})$, $\mathbb{D}_{\mathbf{e}}(\mathbf{W}_0^{t})=\mathbf{0}_{N\times M}$, and $\mathbb{E}_{\mathbf{e}}(\mathbf{W}_0^{t})=\mathbf{W}_0^{t}$.        
As a result, 
       \begin{align}
       \!\!  \!\!  \mathbb{D}_{\mathbf{e}}(\!\mathbf{W}_{j}^{t}\!)\!\!=\!\!  {\sum}_{k\!=\!1}^{j}(\!\mathbf{A}_{1}^{t}\!)^{k\!-\!1}\!\mathbf{A}_{2}^{t}\mathbf{A}_{3}^{t}\!(j\!\!-\!\!k)
  {(\mathbf{W}_{0}^{t}\!\circ\!\mathbf{W}_{0}^{t})}.&\label{D_j_3}
       \end{align}
       By taking $j=J$, we obtain \eqref{D_J}.

	\subsection{Bounds of $\left\Vert\mathbf{M}_1\right\Vert$, $\left\Vert\mathbf{M}_3\right\Vert$, and $\left\Vert\mathbf{M}_4\right\Vert$}\label{proof:B}

\subsubsection{Bounds of $\left\Vert\mathbf{M}_1\right\Vert$}

When $\mathbf{C}^t$ is designed using~\eqref{C_r_definition}, the lower bound of $\left\Vert\mathbf{M}_1\right\Vert$ is obtained based on the triangle inequality of $2$-norm:
\begin{align}
\notag\left\Vert\mathbf{M}_1\!\right\Vert&\!\!=\!\!\bigg\Vert\frac{{\mathbf{1}_N^{\dagger}\mathbf{1}_N}}{N}\!\!-\!\!(\mathbf{C}^{t}\!\!\circ \!\!\mathbf{T}^{t})^J\bigg\Vert\!\!
\geq\!\!\bigg\Vert\frac{\mathbf{1}_N^{\dagger}\mathbf{1}_N}{N}\bigg\Vert\!\!-\!\!\left\Vert(\mathbf{C}^{t}\!\!\circ\!\!\mathbf{T}^{t})^{J}\right\Vert\\&=  1-\left\Vert\mathbf{C}^{t}\!\circ\!\mathbf{T}^{t}\right\Vert^J=1-\beta_{1}^J.\label{m1_lower}
\end{align}
The upper bound of $\left\Vert\mathbf{M}_1\right\Vert$ is given by
\begin{subequations}\small
    \begin{align}
\Vert\mathbf{M}_1\Vert&=\bigg\Vert\frac{{\mathbf{1}_N^{\dagger}\mathbf{1}_N}}{N}-(\mathbf{C}^{t}\!\circ \!\mathbf{T}^{t})^J\bigg\Vert\notag\\
\label{bound_of_Ma}
   &\!\!\! \leq\left\Vert (\mathbf{C}^{t})^{J}\!-\!(\mathbf{C}^{t}\!\circ\!\mathbf{T}^{t})^{J}\right\Vert +\bigg\Vert \frac{\mathbf{1}_N^{\dagger}\mathbf{1}_N}{N}-(\mathbf{C}^{t})^{J}\bigg\Vert \\\label{bound_of_Mb}
    &\!\!\!\leq\frac{\beta_{2}^J-\beta_1^J}{\beta_{2}-\beta_1}\beta_{3}+\beta_{4}^{J},
\end{align}%
\end{subequations}
where \eqref{bound_of_Ma} is based on the triangle inequality of norms, and \eqref{bound_of_Mb} is obtained by plugging \eqref{up-low-of-M4a} and \eqref{lemma-for-P3} into \eqref{bound_of_Ma}.

When $\mathbf{C}^t$ is designed using~\eqref{CoT_r_definition}, $\left\Vert\mathbf{M}_1\right\Vert=\beta_5^J$ due to $\Vert A^J\Vert=\Vert A\Vert^J$ for 2-norm.

 \subsubsection{Bounds of $\left\Vert\mathbf{M}_3\right\Vert$}
When $\mathbf{C}^t$ is designed using either~\eqref{C_r_definition} or~\eqref{CoT_r_definition}, by replacing $\frac{{\mathbf{1}_N^{\dagger}\mathbf{1}_N}}{N}$ with $(\mathbf{C}^{t})^J$ in \eqref{m1_lower} and using $\Vert(\mathbf{C}^{t})^J\Vert=\Vert\mathbf{C}^{t}\Vert^J$, the lower bound of $\left\Vert\mathbf{M}_3\right\Vert$ is  
 $$\left\Vert\mathbf{M}_3\right\Vert
		\geq \Vert\mathbf{C}^{t}\Vert^J-\Vert\mathbf{C}^{t}\circ\mathbf{T}^{t}\Vert^J=\beta_{2}^J-\beta_{1}^J.$$ 
  
By mathematical induction, the upper bound of $\left\Vert\mathbf{M}_3\right\Vert$ is 
     \begin{align} 
 	\left\Vert\mathbf{M}_3\right\Vert\!=\!\left\Vert (\mathbf{C}^{t})^{J}\!-\!(\mathbf{C}^{t}\!\circ\!\mathbf{T}^{t})^{J}\right\Vert \leq\frac{\beta_{2}^J-\beta_1^J}{\beta_{2}-\beta_1}\beta_{3}.\label{lemma-for-P3}
		\end{align}
Specifically, when $J=1$, $\left\Vert \mathbf{C}^{t}-(\mathbf{C}^{t}\!\circ\!\mathbf{T}^{t})\right\Vert =\beta_{3}$. 	  
Suppose that when $J=j$,
$\left\Vert \left(\mathbf{C}^{t}\right)^{j}-(\mathbf{C}^{t}\!\circ\!\mathbf{T}^{t})^{j}\right\Vert \leq\frac{\beta_{2}^j-\beta_1^j}{\beta_{2}-\beta_1}\beta_{3}$.
Then, we can prove that when $J=j+1$,
    \begin{subequations}\small%
      \begin{align}%
		&\!\left\Vert \left(\mathbf{C}^{t}\right)^{j+1}\!-\!(\mathbf{C}^{t}\!\circ\!\mathbf{T}^{t})^{j+\!1}\right\Vert\notag\\
		\leq&\!\left\Vert \!(\mathbf{C}^{t})^{j\!+\!1}\!\!\!-\!\!\mathbf{C}^{t}(\mathbf{C}^{t}\!\!\circ\!\!\mathbf{T}^{t}\!)^{j}\right\Vert \!\!+\!\!\left\Vert \mathbf{C}^{t}\!(\mathbf{C}^{t}\!\!\circ\!\!\mathbf{T}^{t}\!)^{j}\!\!-\!\!(\!\mathbf{C}^{t}\!\!\circ\!\!\mathbf{T}^{t}\!)^{j\!+\!1}\!\right\Vert \!\label{cot_a+1:c}\\
		\leq&\!\left\Vert(\mathbf{C}^{t})^j\right\Vert \!\left\Vert\! (\mathbf{C}^{t})^{j}\!\!-\!\!(\mathbf{C}^{t}\!\!\circ\!\!\mathbf{T}^{t})^{j}\right\Vert \!\!+\!\!\left\Vert (\mathbf{C}^{t}\!\!\circ\!\!\mathbf{T}^{t})^{j}\right\Vert \cdot\left\Vert \mathbf{C}^{t}\!\!-\!\!\mathbf{C}^{t}\!\circ\!\mathbf{T}^{t}\right\Vert\label{cot_a+1:d} \\
		\leq&\frac{\beta_{2}^{j+1}-\beta_1^{j+1}}{\beta_{2}-\beta_1}\beta_{3}.\label{cot_a+1:f}
		\end{align}%
  \end{subequations}%

\subsubsection{Bounds of $\left\Vert\mathbf{M}_4\right\Vert$}
When $\mathbf{C}^t$ is designed using~\eqref{C_r_definition}, $\left\Vert\mathbf{M}_4\right\Vert=\beta_4^J$ due to $\Vert A^J\Vert=\Vert A\Vert^J$ for 2-norm.

When $\mathbf{C}^t$ is designed using~\eqref{CoT_r_definition},  by replacing $(\mathbf{C}^{t}\circ\mathbf{T}^{t})^J$ with $(\mathbf{C}^{t})^J$ in \eqref{m1_lower}, 
the lower bound of $\left\Vert\mathbf{M}_4\right\Vert$
is given by $$\left\Vert\mathbf{M}_4\right\Vert
		\geq \left\Vert\mathbf{C}^{t}\right\Vert^J-\left\Vert\mathbf{C}^{t}\circ\mathbf{T}^{t}\right\Vert^J=\beta_{2}^J-1.$$   
The upper bound of $\left\Vert\mathbf{M}_4\right\Vert$ is given by
\begin{subequations}\small
    \begin{align}
\Vert\mathbf{M}_4\Vert&=\bigg\Vert \frac{\mathbf{1}_N^{\dagger}\mathbf{1}_N}{N}-(\mathbf{C}^{t})^{J}\bigg\Vert \notag\\
\label{bound_of_M4a}
   &\!\!\! \leq\left\Vert (\mathbf{C}^{t})^{J}\!-\!(\mathbf{C}^{t}\!\circ\!\mathbf{T}^{t})^{J}\right\Vert +\bigg\Vert\frac{{\mathbf{1}_N^{\dagger}\mathbf{1}_N}}{N}-(\mathbf{C}^{t}\!\circ \!\mathbf{T}^{t})^J\bigg\Vert\\\label{bound_of_M4b}
    &\!\!\!\leq\frac{\beta_{2}^J-\beta_1^J}{\beta_{2}-\beta_1}\beta_{3}+\beta_{5}^{J},
\end{align}%
\end{subequations}
where \eqref{bound_of_M4a} is based on the triangle inequality of norms, and \eqref{bound_of_M4b} is obtained by plugging \eqref{up-low-of-Ma} and \eqref{lemma-for-P3} into \eqref{bound_of_M4a}.

\subsection{Proof of the upper and lower bounds of $\big\Vert\mathbf{1}_N\mathbf{M}_1\big\Vert$}\label{proof:C}

When $\mathbf{C}^t$ is designed using~\eqref{C_r_definition}, the lower bound of $\big\Vert\mathbf{1}_N\mathbf{M}_1\big\Vert$ is obtained as
\begin{align}\notag&\big\Vert\mathbf{1}_N\mathbf{M}_1\big\Vert=\big\Vert\mathbf{1}_N\big(\frac{\mathbf{1}_N^{\dagger}\mathbf{1}_N}{N}-(\mathbf{C}^{t}\circ\mathbf{T}^{t})^{J}\big)\big\Vert\\\notag&\geq\big\Vert\mathbf{1}_N\big\Vert-\big\Vert\mathbf{1}_N(\mathbf{C}^{t}\circ\mathbf{T}^{t})^{J}\big\Vert\\
 &\geq\sqrt{N}(1-\big\Vert(\mathbf{C}^{t}\circ\mathbf{T}^{t})^{J}\big\Vert)
 \geq\sqrt{N}(1-\big\Vert\mathbf{C}^{t}\circ\mathbf{T}^{t}\big\Vert^{J})\notag\\\notag&=\sqrt{N}(1- \beta_{1}^{J}). \end{align}%
The upper bound of $\big\Vert\mathbf{1}_N\mathbf{M}_1\big\Vert$ is given by
\begin{subequations}\label{proof_1M}
    \begin{align}
    \big\Vert\mathbf{1}_N\mathbf{M}_1\big\Vert
    \leq&\big\Vert\mathbf{1}_N-\beta_6\mathbf{1}_{N}(\mathbf{C}^{t}\circ\mathbf{T}^{t})^{J-1}\big\Vert\label{proof_1Ma}
    \\\leq&\big\Vert\mathbf{1}_N-\beta_6^{J}\cdot\mathbf{1}_N\big\Vert\label{proof 1M c}=
    \sqrt{N}\left(1-\beta_6^{J}\right),
\end{align}%
\end{subequations}
where \eqref{proof_1Ma} is due to
 \begin{align}
\mathbf{1}_N(\mathbf{C}^{t}\circ\mathbf{T}^{t})
    \succeq\,\underset{\forall m}{\min}\;\Big\{{\sum}_{\forall n}c_{n,m}^{t}(1\!-\!\epsilon_{\mathrm{P},n,m}^{t})\Big\}\mathbf{1}_{N},\label{proof 1M d}
 \end{align} 
and \eqref{proof 1M c} is obtained by applying \eqref{proof 1M d} for $J$ times.

When $\mathbf{C}^t$ is designed using~\eqref{CoT_r_definition}, $\left\Vert\mathbf{1}_N \mathbf{M}_1\right\Vert=\big\Vert\mathbf{1}_N-\mathbf{1}_N(\mathbf{C}^{t}\circ\mathbf{T}^{t})^{J}\big\Vert=0$ since $\mathbf{1}_N(\mathbf{C}^{t}\circ\mathbf{T}^{t})^{J}=\mathbf{1}_N$.
 
\subsection{Proof of the upper and lower bounds of $\Vert\mathbf{M}_2\Vert$}\label{proof:D}
We first use mathematical induction to prove that 
\begin{align}
\mathbb{D}_{\mathbf{e}}(\mathbf{w}_{n,j}^{t})\preceq\left(\mathbf{y}_{n,j}^{t}-\mathbb{E}_{\mathbf{e}}(\mathbf{w}_{n,j}^{t})\right)\circ\mathbb{E}_{\mathbf{e}}(\mathbf{w}_{n,j}^{t}).\label{D_w_n_j}    
\end{align}
When $J=1$, \eqref{D_w_n_j} holds:
\begin{align}
     \mathbb{D}_{\mathbf{e}}(\mathbf{w}_{n,1}^{t})\preceq\left(\mathbf{y}_{n,0}^{t}-\mathbb{E}_{\mathbf{e}}(\mathbf{w}_{n,0}^{t})\right)\circ\mathbb{E}_{\mathbf{e}}(\mathbf{w}_{n,0}^{t}).
\end{align}
Suppose that when $J=j-1$, \eqref{D_w_n_j} holds:
\begin{align}
    \mathbb{D}_{\mathbf{e}}(\mathbf{w}_{n,j-1}^{t})\preceq\left(\mathbf{y}_{n,j-1}^{t}-\mathbb{E}_{\mathbf{e}}(\mathbf{w}_{n,j-1}^{t})\right)\circ\mathbb{E}_{\mathbf{e}}(\mathbf{w}_{n,j-1}^{t}).\label{Variance_wnj-1}
\end{align}
Next, we prove \eqref{D_w_n_j} when $J=j$. 
By plugging \eqref{Variance_wnj-1} into \eqref{Variance_wnj} and then simplifying \eqref{Variance_wnj}, it readily follows that
\begin{subequations}\footnotesize
    \begin{align}
  \notag &\!\!\mathbb{D}_{\mathbf{e}}(\mathbf{w}_{n,j}^{t}) \!\! \\&\notag\!\!\preceq\!\!\sum_{\forall m}\!\!c_{n,m}^{t}\!(\mathbf{y}_{m,j\!-\!1}^{t}\!\!-\!\!(1\!\!-\!\!\epsilon_{\mathrm{P},n,m}^{t})\mathbb{E}_{\mathbf{e}}\!(\!\mathbf{w}_{m,j\!-\!1}^{t}))\!\!\circ\!c_{n,m}^{t}(1\!\!-\!\!\epsilon_{\mathrm{P},n,m}^{t})\mathbb{E}_{\mathbf{e}}(\!\mathbf{w}_{m,j\!-\!1}^{t}\!)\!\\\notag
   &\!\!\preceq\!\!\sum_{\forall m}\!\!c_{n,m}^{t}\!\!\left(\mathbf{y}_{m,j\!-\!1}^{t}\!\!-\!\!(1\!\!-\!\!\epsilon_{\mathrm{P},n,m}^{t})\mathbb{E}_{\mathbf{e}}(\mathbf{w}_{m,j\!-\!1}^{t})\!\right)\!\!\circ\!\!\sum_{\forall m}c_{n,m}^{t}(1\!\!-\!\!\epsilon_{\mathrm{P},n,m}^{t})\mathbb{E}_{\mathbf{e}}(\mathbf{w}_{m,j\!-\!1}^{t})\\
    \notag&\!\!=\!\!\left(\mathbf{y}_{n,j}^{t}-\mathbb{E}_{\mathbf{e}}(\mathbf{w}_{n,j}^{t})\right)\circ\mathbb{E}_{\mathbf{e}}(\mathbf{w}_{n,j}^{t}),
\end{align}
\end{subequations}
where $\mathbf{y}_{n,j}^{t}={\sum}_{\forall m}(c_{n,m}^{t})\mathbf{y}_{m,j-1}^{t}$ is the locally aggregated, error-free model of device $n$,
% {\color{red}using the given $\mathbf{C}^t$}, 
which is the $n$-th element of $(\mathbf{C}^{t})^j\mathbf{W}^{t}_{0}$; and the erroneous locally aggregated model of device $n$ is $\mathbb{E}_{\mathbf{e}}(\mathbf{w}_{n,j}^{t})={\sum}_{\forall m}(c_{n,m}^{t})(1-\epsilon_{\mathrm{P},n,m}^{t})\mathbb{E}_{\mathbf{e}}(\mathbf{w}_{m,j-1}^{t})$.
As a result, for any $J\in \mathbb{N}_+$, it follows that $\mathbb{D}_{\mathbf{e}}(\mathbf{w}_{n,j}^{t})\preceq\left(\mathbf{y}_{n,j}^{t}-\mathbb{E}_{\mathbf{e}}(\mathbf{w}_{n,j}^{t})\right)\circ\mathbb{E}_{\mathbf{e}}(\mathbf{w}_{n,j}^{t})$. 
Applying \eqref{D_J}, we can rewrite the upper bound of $\mathbb{D}_{\mathbf{e}}\!(\mathbf{w}_{n,J}^{t}),\forall n$ in a matrix form:
\begin{align}
\notag&\!\!\mathbf{M}_2(\mathbf{W}_{0}^{t}\circ\mathbf{W}_{0}^{t})  =\mathbb{D}_{\mathbf{e}}\!(\mathbf{W}_{J}^{t}) \!\\
&\preceq\!(\mathbf{C}^{t}\!\circ\!\mathbf{T}^{t})^{\!J}\mathbf{W}_{0}^{t}\!\circ\!\big((\mathbf{C}^{t})^{J}-\!(\mathbf{C}^{t}\!\circ\!\mathbf{T}^{t})^{J}\big)\mathbf{W}_{0}^{t}\notag\\
&\!\!  =\left((\mathbf{C}^{t}\!\circ\!\mathbf{T}^{t})^{J}\!\circ\!\big(\!(\mathbf{C}^{t})^{J}\!-\!(\mathbf{C}^{t}\!\circ\!\mathbf{T}^{t})^{J}\big)\right)(\mathbf{W}_{0}^{t}\circ\mathbf{W}_{0}^{t}).\label{matrix_variance_J}
\end{align}
 As a result, $\mathbf{M}_2\preceq\!(\mathbf{C}^{t}\!\!\circ\!\!\mathbf{T}^{t})^{\!J}\circ\!\!\big(\!(\mathbf{C}^{t})^{J}\!-\!\!(\mathbf{C}^{t}\!\!\circ\!\!\mathbf{T}^{t})^{\!J}\!\big)$. The upper bound of $\Vert\mathbf{M}_2\Vert$ is given by
\begin{subequations}
 \begin{align}
\notag\Vert\mathbf{M}_{2}\Vert\leq&\left\Vert \!(\mathbf{C}^{t} \circ\!\!\mathbf{T}^{t})^{\!J}\circ\!\!\big(\!(\mathbf{C}^{t})^{J}\!-\!\!(\mathbf{C}^{t}\!\!\circ\!\!\mathbf{T}^{t})^{\!J}\!\!\right\Vert \\
\leq &\left\Vert \!(\mathbf{C}^{t}\!\!\circ\!\!\mathbf{T}^{t})^{\!J}\!\right\Vert \left\Vert \big(\!(\mathbf{C}^{t})^{J}\!-\!\!(\mathbf{C}^{t}\!\!\circ\!\!\mathbf{T}^{t})^{\!J}\!\right\Vert \label{M2_a}
\\\leq&\frac{\beta_{2}^J-\beta_1^J}{\beta_{2}-\beta_1}\beta_{3}\beta_1^J,\label{M2_b}
\end{align}
\end{subequations}
 where \eqref{M2_a} is based on $\Vert A\circ B\Vert \leq\Vert A\Vert \Vert B\Vert$, and \eqref{M2_b} is based on \eqref{lemma-for-P3}. 

The lower bound of $ \mathbb{D}_{\mathbf{e}}\!(\mathbf{w}_{n,J}^{t}),\forall n$ can also be written in a matrix form, as given by
\begin{align}
 \notag   \notag&\!\!\!\mathbf{M}_2 (\mathbf{W}_{0}^{t}\circ\mathbf{W}_{0}^{t})=  \mathbb{D}_{\mathbf{e}}(\mathbf{W}_J^{t})\notag\succeq(\mathbf{C}^{t}\!\circ\!\mathbf{C}^{t}\!\circ\!\mathbf{T}^{t}\!\circ\!({\mathbf{1}_{N}^{\dagger}\!\mathbf{1}_{N}}\!-\!\mathbf{T}^{t}))\!\\&\quad\times\!\left(\!(\mathbf{C}^{t}\!\circ\!\mathbf{T}^{t})^{J}\!\circ\!(\mathbf{C}^{t}\circ\mathbf{T}^{t})^{J}\right)(\mathbf{W}_{0}^{t}\circ\mathbf{W}_{0}^{t}),\label{M_2_lower bound}
\end{align}
since the first term in the RHS of \eqref{D_j_1} is always positive.

Since all elements on the RHS of \eqref{M_2_lower bound} are positive, we have
\begin{align}
\notag\mathbf{M}_2\succeq&(\mathbf{C}^{t}\!\circ\!\mathbf{C}^{t}\!\circ\!\mathbf{T}^{t}\!\circ\!({\mathbf{1}_{N}^{\dagger}\mathbf{1}_{N}}\!-\!\mathbf{T}^{t}))\!\left(\!(\mathbf{C}^{t}\!\circ\!\mathbf{T}^{t})^{J}\!\circ\!(\mathbf{C}^{t}\!\circ\!\mathbf{T}^{t})^{J}\right)\\\succeq&\beta_{7}^{2j+1}(1-\beta_{8})(\mathbf{C}^{t}\!\circ\!\mathbf{C}^{t}\!)\!\left(\!(\mathbf{C}^{t}\!)^{J}\!\circ\!(\mathbf{C}^{t})^{J}\right),\label{M_2_lowerbound:a}
\end{align}
where \eqref{M_2_lowerbound:a} is obtained because $\beta_{7}$ and $1-\beta_{8}$ are the minimum of $\mathbf{T}^{t}$ and ${\mathbf{1}_{N}^{\dagger}\!\mathbf{1}_{N}}\!\!-\!\!\mathbf{T}^{t}$, respectively.

The lower bound of the 2-norm of $\mathbf{M}_2$ is given by
\begin{subequations}
    \begin{align}
\Vert\mathbf{M}_2\Vert\!&\geq\!\beta_{7}^{2J\!+\!1}(1\!-\!\beta_{8})\left\Vert(\mathbf{C}^{t}\!\circ\!\mathbf{C}^{t}\!)\!\left(\!(\mathbf{C}^{t}\!)^{J}\!\circ\!(\mathbf{C}^{t})^{J}\right)\right\Vert\label{2_norm_M2_lowerbound:b}\\
&\geq \beta_{7}^{2J+1}(1-\beta_{8})\frac{\Vert\mathbf{C}^{t}\!(\mathbf{C}^{t}\!)^{J}\Vert^{2}}{N}\label{2_norm_M2_lowerbound:c}\\
&=\frac{\beta_{7}^{2J+1}(1-\beta_{8})\beta_{2}^{2J+2}}{N},\label{2_norm_M2_lowerbound:d}
\end{align}
\end{subequations}
where \eqref{2_norm_M2_lowerbound:b} is obtained by taking the 2-norm on both sides of \eqref{M_2_lowerbound:a}. Applying Cauchy's inequality to the corresponding elements of $(\mathbf{C}^{t}\!\circ\!\mathbf{C}^{t}\!)\!\left(\!(\mathbf{C}^{t}\!)^{J}\!\circ\!(\mathbf{C}^{t})^{J}\right)$ and $\mathbf{C}^{t}\!(\mathbf{C}^{t}\!)^{J}$ yields \eqref{2_norm_M2_lowerbound:c}.

\subsection{Proof of \textbf{Lemma \ref{appedix_lemma3}}}\label{proof:lemma5}%
An upper bound on ${\mathbb{E}_{\{ \xi_{n,0}^{t},\cdots,\xi_{n,i\!-\!1}^{t},\forall n\}}}{\sum}_{\forall n}p_{n}\!\left\Vert \bar{\boldsymbol{\omega}}_{i}^{t}\!-\!\boldsymbol{\omega}_{n,i}^{t}\!\right\Vert ^{2}$~is 
 \begin{subequations}
	   \small \begin{align}%
&\!\! {\mathbb{E}_{\{ \xi_{n,0}^{t},\cdots,\xi_{n,i\!-\!1}^{t},\forall n\}}}\!\Big({\sum}_{\forall n}\!p_{n}\!\Vert \bar{\boldsymbol{\omega}}_{i}^{t}\!\!-\!\!\boldsymbol{\omega}_{n,i}^{t}\Vert ^{2}\!\Big)\!\!\notag
\\
&\!\!=\!\!{\sum}_{\forall n}\!p_{n}{\mathbb{E}_{\{ \xi_{n,0}^{t},\cdots,\xi_{n,i\!-\!1}^{t}\}}}\!\!\!\left(\Vert \bar{\boldsymbol{\omega}}_{i}^{t}\!\!-\!\!\bar{\boldsymbol{\omega}}_{i\!-\!1}^{t}\!\!+\!\!\bar{\boldsymbol{\omega}}_{i\!-\!1}^{t}\!\!-\!\!\boldsymbol{\omega}_{n,i}^{t}\Vert ^{2}\right)\notag
\\
&\!\!\leq\!\!{\sum}_{\forall n}\!p_{n}{\mathbb{E}_{\{ \xi_{n,0}^{t},\cdots,\xi_{n,i\!-\!1}^{t}\}}}\left(\Vert \bar{\boldsymbol{\omega}}_{i\!-\!1}^{t}\!-\!\boldsymbol{\omega}_{n,i}^{t}\Vert ^{2}\right)\label{lemma3:b}\\
&\!\!\leq\!\!\sum_{\forall n}\!p_{n}\!(\!1\!\!+\!\!\eta\!)\!\Big(\!{\mathbb{E}_{\{ \xi_{n,0}^{t},\cdots,\xi_{n,i\!-\!2}^{t}\}}}\!\Vert\bar{\boldsymbol{\omega}}_{i\!-\!1}^{t}\!\!-\!\!\boldsymbol{\omega}_{n,i\!-\!1}^{t}\!\Vert^{2}\!\!+\!\!\eta{\mathbb{E}_{\xi_{n,i\!-\!1}^{t}}}\!\!\big(\!\Vert\!\mathbf{g}_{\xi,n,i\!-\!1}^{t}\!\Vert^{2}\!\big)\!\!\Big)\label{lemma3:d}\\
&\!\!\leq\!\!(1\!\!+\!\!\eta){\sum}_{\forall n}p_{n}\Big({\mathbb{E}_{\{ \xi_{n,0}^{t},\cdots,\xi_{n,i\!-\!2}^{t}\}}}\left(\Vert\bar{\boldsymbol{\omega}}_{i\!-\!1}^{t}\!\!-\!\!\boldsymbol{\omega}_{n,i\!-\!\!1}^{t}\Vert^{2}\right)\!\!+\!\!\eta G^{2}\Big)\label{lemma3:e}\\
&\!\!\leq\!\!(1\!+\!\eta)\!^{i}{\sum}_{\forall n}\!p_{n}\!\left\Vert \bar{\boldsymbol{\omega}}_{0}^{t}\!\!-\!\!\boldsymbol{\omega}_{n,0}^{t}\right\Vert^{2}\!\!\!+\!\frac{\!(1\!+\!\eta)^{i+1}\!\!-\!\!(1\!+\!\eta)\!}{\eta} \eta G^{2}\!\!\!\label{lemma3:f}\\
&\!\!=\!\!(1\!\!+\!\!\eta)^{i}\!{\sum}_{\forall n}\!p_{n}\!\Vert\boldsymbol{\varpi}_{n}^{t\!-\!1}\!\!\!-\!\!\!{\sum}_{\forall m}\!\!p_{m}\!\boldsymbol{\varpi}_{m}^{t\!-\!1}\!\Vert^{2}\!\!+\!\!\big[{(1\!\!+\!\!\eta)^{i\!+\!1}\!\!-\!\!(1\!\!+\!\!\eta)\!}\big]\!G^{2}\!,\!\!\!\label{lemma3:g}\!%
		\end{align}	    
		\end{subequations}
where \eqref{lemma3:b} is due to the fact that ${\mathbb{E}\big(\big\Vert X-\mathbb{E}(X)\big\Vert ^{2}\big)=\mathbb{E}\big(\big\Vert X\big\Vert ^{2}\big)-\big\Vert \mathbb{E}(X)\big\Vert ^{2}}\leq\mathbb{E}\big(\big\Vert X\big\Vert ^{2}\big)$ and here $X=\bar{\boldsymbol{\omega}}_{i-1}^{t}-\boldsymbol{\omega}_{n,i}^{t}$;  \eqref{lemma3:d} is obtained by substituting \eqref{epoch_imperfect} and $\mathbf{g}^{t}_{\xi,n,i-1}$ into \eqref{lemma3:b} and then exploiting  {$\Vert\mathbf{a}+\mathbf{b}\Vert^{2}\leq(1+\eta)\Vert\mathbf{a}\Vert^{2}+(1+\frac{1}{\eta})\Vert\mathbf{b}\Vert^{2}$}, here, $\mathbf{a}=\bar{\boldsymbol{\omega}}_{i-1}^{t}-\boldsymbol{\omega}_{n,i-1}^{t}$ and $\mathbf{b}=\eta\mathbf{g}_{\xi,n,i-1}^{t}$; and \eqref{lemma3:e} is based on \textbf{Assumption \ref{assumption}}-3).
By repeating \eqref{lemma3:e} for $i$ times, we obtain \eqref{lemma3:f}. 

Finally, \eqref{lemma3:g} is obtained by substituting \eqref{eq:omega_varpi} into \eqref{lemma3:f}:%
 \begin{align}
\!\!\bar{\boldsymbol{\omega}}_{0}^{t}\!-\!\boldsymbol{\omega}_{n,0}^{t}&\!=\!{\sum}_{\forall m}\!p_{m}\mathbf{w}_{m,J}^{t\!-\!1}\!-\!\mathbf{w}_{n,J}^{t\!-\!1}\!\notag\\&=\!{\sum}_{\forall m}\!p_{m}\!\left(\boldsymbol{\bar{\omega}}_{I}^{t\!-\!1}\!-\!\boldsymbol{\varpi}_{m}^{t\!-\!1}\right)\!-\!\left(\boldsymbol{\bar{\omega}}_{I}^{t\!-\!1}\!-\!\boldsymbol{\varpi}_{n}^{t\!-\!1}\right) \!\notag\\&=\!\boldsymbol{\varpi}_{n}^{t\!-\!1}\!-\!{\sum}_{\forall m}\!p_{m}\!\boldsymbol{\varpi}_{m}^{t\!-\!1}.\label{eq:omega_varpi}
  \end{align}
Here, \eqref{eq:omega_varpi} is based on the definition in \eqref{Eq:consensus error noise} and $\boldsymbol{\omega}_{n,0}^{t}=\mathbf{w}_{n,J}^{t-1}$.

\subsection{Proof of \textbf{Lemma \ref{appedix_lemma2}}}\label{proof:lemma6}
 \subsubsection{Proof of \eqref{delta_1}}

We assume an ideal scenario where the transmissions are error-free and full-batch gradient is considered. After $J$ (error-free) local aggregations, each device could obtain ${\mathbf{x}^{t-1}_{n,J}}$ in the $(t-1)$-th training round, see \eqref{eq: consensus aggregation}, to start its local training in the $t$-th training round. We further introduce an auxiliary variable $\boldsymbol{\bar{\varphi}}^{t}_{1}$, which virtually aggregates the local models that all devices obtain after the first iteration of the $t$-th training round, i.e., 
\begin{align}\label{yr1}
\boldsymbol{\bar{\varphi}}^{t}_{1}={\sum}_{\forall n}p_n\big({\mathbf{x}^{t-1}_{n,J}}-\eta{\nabla F_{n}(\mathbf{x}^{t-1}_{n,J},\xi_{n,0}^{t})}\big).
\end{align}

Using $\boldsymbol{\bar{\varphi}}^{t}_{1}$, we bound the gaps between $\bar{\boldsymbol{\omega}}^{t}_{1}$ and $\boldsymbol{\bar{\varphi}}^{t}_{1}$,  between $\boldsymbol{\bar{\varphi}}^{t}_{1}$ and $\boldsymbol{\omega}^{*}$, and hence between $\bar{\boldsymbol{\omega}}^{t}_{1}$ and $\boldsymbol{\omega}^{*}$, i.e., $\Delta\boldsymbol{\omega}^{t}_{1}$. Specifically,
When $i\!\!=\!\!1$, we expand $\Delta\boldsymbol{\omega}^{t}_{1}$, and then apply the inequality of arithmetic and geometric means, i.e.,
{\small\begin{align}
\!\!\!\!\!\mathbb{E}_{\{ \!\xi_{n,0}^{t},{\forall n}\!\}}\!(\!\Delta\boldsymbol{\omega}_{1}^{t}\!)\!\!\leq\!\!\mathbb{E}_{\{ \!\xi_{n,0}^{t},{\forall n} \!\}}\!\!\left(\!\!(1\!\!+\!\!\frac{1}{\tau_{\epsilon}}\!)\!\Vert\bar{\boldsymbol{\omega}}_{1}^{t}\!\!-\!\!\boldsymbol{\bar{\varphi}}_{1}^{t}\!\Vert^{2}\!\!\!+\!\!(\!1\!\!+\!\!\tau_{\epsilon}\!)\!\Vert\!\boldsymbol{\bar{\varphi}}_{1}^{t}\!\!-\!\!\boldsymbol{\omega}^{*}\!\Vert^{2}\!\!\right)\!,\label{delta_r1}
\end{align}}where $\tau_{\epsilon}\in\mathbb{R}_+$ can be any positive values\footnote{
If the channel is perfect, then $\epsilon_{\mathrm{P},n,m}^t \rightarrow 0,\forall m,n,t$ and $\tau_{\epsilon}\rightarrow0$. Meanwhile, $\mathbb{E}_{\{ \xi_{n,0}^{t},{\forall n}\}}\!\big(\big\Vert \bar{\boldsymbol{\omega}}^{t}_{1}\!\!-\!\!\boldsymbol{\bar{\varphi}}^{t}_{1}\big\Vert ^{2}\big)\rightarrow 0$ since  $\mathbf{w}_{n,J}^{t-1}\rightarrow\mathbf{x}_{n,J}^{t-1}$. As a result, the RHS of \eqref{delta_r1} tends to $\mathbb{E}_{\{\xi_{n,0}^{t},{\forall n}\}}(\left\Vert \boldsymbol{\bar{\varphi}}^{t}_{1}\!\!-\!\!\boldsymbol{\omega}^{*}\right\Vert ^{2})$.}. By plugging \eqref{imprecise x} and \eqref{yr1} in $\mathbb{E}_{\{ \xi_{n,0}^{t}\}_{\forall n}}\!\big(\big\Vert \bar{\boldsymbol{\omega}}^{t}_{1}\!\!-\!\!\boldsymbol{\bar{\varphi}}^{t}_{1}\big\Vert ^{2}\big)$, we have
\begin{subequations}\label{a_bound:eq}%
    \small\begin{align}&\!\!\!\mathbb{E}_{\{ \xi_{n,0}^{t},{\forall n}}\}\!\big(\big\Vert \bar{\boldsymbol{\omega}}^{t}_{1}\!\!-\!\!\boldsymbol{\bar{\varphi}}^{t}_{1}\big\Vert ^{2}\big)\!\notag\\
    &\!\!\!\!\!\!\!\!=\!\!\mathbb{E}_{\{ \xi_{n,0}^{t},{\forall n}\}}\!\big\Vert\!
    {\sum}_{\forall n}\!\!p_{n}\!\big(\!\mathbf{w}_{n,J}^{t\!-\!1}\!\!-\!\!\eta\mathbf{g}_{\xi,n,0}^{t}\!\!-\!\!{\mathbf{x}_{n,J}^{t\!-\!1}}\!\!+\!\!\eta\nabla F_{n}\!\big({\mathbf{x}_{n,J}^{t\!-\!1}},\xi_{n,0}^{t}\big)\!\big\Vert^{2}\!\!\label{a_bound:eqa}\\
&\!\!\!\!\!\!\!\!\leq\!\!(1\!\!+\!\!\tau_{\eta})\!\big\Vert\!{\sum}_{\forall n}p_{n}(\mathbf{w}_{n,J}^{t\!-\!1}-\bar{\boldsymbol{\omega}}_{I}^{t-1})\!\big\Vert^{2}\notag\\
&\!\!\!\!\!+\!(1\!\!+\!\!\frac{1}{\tau_{\eta}})\mathbb{E}_{\{ \xi_{n,0}^{t},{\forall n}\}}\!\big\Vert\!{\sum}_{\forall n}\!\!p_{n}\!\left(\!\eta\mathbf{g}_{\xi,n,0}^{t}\!\!-\!\!\eta\!\nabla F_{n}\!\big({\mathbf{x}_{n,J}^{t-1}},\xi_{n,0}^{t}\big)\!\right)\!\!\big\Vert^{2}\label{a_bound:a}
   \\&\!\!\!\!\!\!\!\!\leq\!\! (\!1\!\!+\!\!\tau_{\eta}\!)\!\big\Vert\! {\sum}_{\forall n}\!p_{n}\boldsymbol{\varpi}_{n}^{t\!-\!1}\!\big\Vert ^{2}\!\!\!+\!\!(1\!\!+\!\!\frac{1}{\tau_{\eta}}\!)\eta^2 L^2\!{\sum}_{\forall n}\!\!p_{n}\!\big\Vert\mathbf{x}_{n,J}^{t\!-\!1}\!\!-\!\!\mathbf{w}_{n,J}^{t\!-\!1}\!\big\Vert^2\!\label{a_bound:eqc},
		\end{align}%
\end{subequations}%
where \eqref{a_bound:a} is based on $\Vert x_1+x_2\Vert^2\leq(1+\tau_{\eta})\Vert x_1\Vert^2+(1+\frac{1}{\tau_{\eta}})\Vert x_2\Vert^2$ and  ${\sum_{\forall n}\!p_{n}{\mathbf{x}_{n,J}^{t-1}}}=\bar{\boldsymbol{\omega}}_{I}^{t-1}$. \eqref{a_bound:eqc} follows the definition in \eqref{Eq:consensus error noise} and \eqref{a_bound:eqc} is based on $L$-smoothness and Cauchy's inequality when $\sum_{\forall n}p_n=1$, i.e., $\big\Vert\!\sum_{\forall n}p_{n}\big(\eta\mathbf{g}_{\xi,n,0}^{t}-\eta\nabla F_{n}({\mathbf{x}_{n,J}^{t-1}},\xi_{n,0}^{t})\big)\big\Vert^{2} \leq \eta^{2}L^{2}{\sum}_{\forall n}p_{n}\big\Vert\mathbf{w}_{n,J}^{t-1}-\mathbf{x}_{n,J}^{t-1}\big\Vert^{2}, \forall \xi_{n,0}^{t}$. Here, $\tau_{\eta}$ and $\eta$ are positively and linearly related. When $\eta=0$ and $\tau_{\eta}=0$, the RHS of \eqref{a_bound:eqc} tends to $\big\Vert {\sum}_{\forall n}p_{n}\boldsymbol{\varpi}_{n}^{t-1}\!\big\Vert ^{2}$.
By replacing the local training inputs $\{\boldsymbol{\omega}_{n,i-1}^t\}$ with $\{\mathbf{x}_{n,J}^{t-1}\}$ in \eqref{delta_ib}, and then replacing the output $\bar{\boldsymbol{\omega}}^{t}_{i}$ with  $\boldsymbol{\bar{\varphi}}^{t}_{1}$, we have%
    \begin{align}%
   \notag  &\mathbb{E}_{\{ \xi_{n,0}^{t}\} _{\forall n}}\!\!\!\left(\Vert \boldsymbol{\bar{\varphi}}^{t}_{1}\!\!-\!\!\boldsymbol{\omega}^{*}\Vert ^{2}\right)\leq \!\left(\!1\!-\!\frac{\mu\eta}{2}\right)\Delta\boldsymbol{\omega}_{I}^{t-1}+\!\!{\sum}_{\forall n}\!\eta^2p_{n}^{2}\sigma^2_n\!\!\\
    &\quad+\left(\!2\eta^{2}L^{2}\!+\!(L\!+\!\mu)\eta\right){\sum}_{\forall n}p_{n}\big(\Vert \bar{\boldsymbol{\omega}}_{I}^{t\!-\!1}\!-\!\mathbf{x}_{n,J}^{t\!-\!1}\Vert ^{2}\!\big),\label{b}\end{align}%
where $\bar{\boldsymbol{\omega}}_{I}^{t-1}={\sum}_{\forall n}p_{n}\mathbf{x}_{n,J}^{t-1}$.
Substituting \eqref{a_bound:eqc} and \eqref{b} into \eqref{delta_r1}, we finally obtain \eqref{delta_1}.

\subsubsection{Proof of \eqref{delta_ib}} 
When $i=2,\cdots,I$, we expand $\Delta\boldsymbol{\omega}^{t}_{i}$ by applying \eqref{imprecise x} and \eqref{delta_r}, i.e.,
\begin{small}%
   { \begin{align}
\Delta&\boldsymbol{\omega}^{t}_{i}\!\!=\!\!\big\Vert\!\sum_{\forall n}\!p_{n}\!(\boldsymbol{\omega}_{n,i\!-\!1}^{t}\!\!-\!\!\eta\mathbf{g}_{n,i\!-\!1}^{t})\!\!-\!\!\boldsymbol{\omega}^{*}\!\big\Vert^{2}\!\!+\!\!\big\Vert\!\sum_{\forall n}\!\!p_{n}\!(\eta\mathbf{g}_{n,i\!-\!1}^{t}\!\!-\!\!\eta\mathbf{g}_{\xi,n,i\!-\!1}^{t})\big\Vert^{2}\!\notag\\\notag
&+\!\!2\Big\langle \!{\sum}_{\forall n}\!p_{n}\!\left(\boldsymbol{\omega}_{n,i\!-\!1}^{t}\!\!-\!\!\eta\mathbf{g}_{n,i\!-\!1}^{t}\right)\!\!-\!\!\boldsymbol{\omega}^*,{\sum}_{\forall n}\!p_{n}\!(\eta\mathbf{g}_{n,i\!-\!1}^{t}\!\!-\!\!\eta\mathbf{g}_{\xi,n,i\!-\!1}^{t})\!\Big\rangle.
		\end{align}}%
\end{small}%
 Based on $\mathbb{E}_{\xi_{n,i-1}^{t}}(\mathbf{g}_{\xi,n,i-1}^{t}-\mathbf{g}_{n,i-1}^{t})=0$ and \eqref{a_bound:SGD}, we have
        \begin{align}
\notag&\mathbb{E}_{\left\{ \xi_{n,i-1}^{t},{\forall n}\right\}}\! ( \Delta\boldsymbol{\omega}^{t}_{i})\! \! =\!\big\Vert{\sum}_{\forall n}p_{n}\left(\boldsymbol{\omega}^{t}_{n,i\!-\!1}\!-\!\eta\mathbf{g}^{t}_{n,i\!-\!1}\right)\!\!-\!\!\boldsymbol{\omega}^{*}\big\Vert^{2}\\\notag&\quad\quad+\mathbb{E}_{\left\{ \xi_{n,i\!-\!1}^{t},{\forall n}\right\}}\Big(\big\Vert{\sum}_{\forall n}p_{n}\left(\eta\mathbf{g}_{n,i\!-\!1}^{t}\!\!-\!\!\eta\mathbf{g}_{\xi,n,i\!-\!1}^{t}\right)\big\Vert^{2}\Big)\\
&\quad\leq\!\big\Vert\!{\sum}_{\forall n}\!p_{n}\!\left(\!\boldsymbol{\omega}^{t}_{n,i\!-\!1}\!\!-\!\!\eta\mathbf{g}^{t}_{n,i\!-\!1}\!\right)\!\!-\!\!\boldsymbol{\omega}^{*}\big\Vert^{2}\!\!\!+\!\!{\sum}_{\forall n}\eta^2p_{n}^{2}\sigma^2_n.\label{delta_i_sgd}%
\end{align}%
   Let $\bar{\mathbf{g}}^{t}_{i-1}={\sum}_{\forall n}p_{n}\mathbf{g}^{t}_{n,i-1}$ define the virtually aggregated full-batch gradient in the $(i-1)$-th iteration of the $t$-th round. 
   The first term on the RHS of \eqref{delta_i_sgd} can be expanded, as given by
\begin{align}
\notag&\big\Vert{\sum}_{\forall n}p_{n}\left(\boldsymbol{\omega}^{t}_{n,i\!-\!1}\!-\!\eta\mathbf{g}^{t}_{n,i\!-\!1}\right)\!-\!\boldsymbol{\omega}^{*}\big\Vert^{2}\\
&=\!\!\left\Vert\bar{\boldsymbol{\omega}}^{t}_{i\!-\!1}\!\! -\!\! \boldsymbol{\omega}^{*}\!\right\Vert ^{2}\!+\!\left\Vert \eta\bar{\mathbf{g}}^{t}_{i\!-\!1}\right\Vert ^{2}\!-\!2\eta\!\left\langle \bar{\boldsymbol{\omega}}^{t}_{i\!-\!1}\!\!-\!\boldsymbol{\omega}^{*}\!,\bar{\mathbf{g}}^{t}_{i\!-\!1}\!\right\rangle .\label{triangle i}
		\end{align}
		The upper bound of $\left\Vert \eta\bar{\mathbf{g}}^{t}_{i-1}\right\Vert ^{2}$ is given by
	\begin{subequations}\small\label{g_bound}
	    \begin{align}
		&\!\!
  \left\Vert \! \eta\bar{\mathbf{g}}_{i\!-\!1}^{t}\right\Vert ^{2}\!\!=\!\!\eta^{2}\!\Big\Vert \!{\sum}_{\forall n}\!p_{n}\!\left(\mathbf{g}_{n,i\!-\!1}^{t}\!\!-\!\!\nabla F_{n}(\bar{\boldsymbol{\omega}}_{i\!-\!1}^{t})\!\right)\!\!+\!\!\nabla F(\bar{\boldsymbol{\omega}}_{i\!-\!1}^{t})\Big\Vert ^{2}\notag\\
  &\leq\!2\eta^{2}\Big\Vert \!{\sum}_{\forall n}p_{n}\!\left(\mathbf{g}_{n,i\!-\!1}^{t}\!\!-\!\!\nabla F_{n}(\bar{\boldsymbol{\omega}}_{i\!-\!1}^{t})\!\right)\Big\Vert ^{2}\!\!\!\!+\!\!2\eta^{2}\left\Vert \nabla F(\bar{\boldsymbol{\omega}}_{i\!-\!1}^{t})\right\Vert ^{2}\label{g_bound:b}\\
 & \leq2\eta^{2}L^{2}\!{\sum}_{\forall n}p_{n}\!\left\Vert \boldsymbol{\omega}_{n,i\!-\!1}^{t}\!\!-\!\!\bar{\boldsymbol{\omega}}_{i\!-\!1}^{t}\right\Vert ^{2}\!\!\!+\!\!4L\eta^{2}\left(F(\bar{\boldsymbol{\omega}}_{i\!-\!1}^{t})\!\!-\!\!F^{*}\right)\label{g_bound:c},
		\end{align}
	\end{subequations}where 
\eqref{g_bound:b} is based on $\Vert x_1+x_2\Vert^2\leq 2\Vert x_1\Vert^2+2\Vert x_2\Vert^2$; and \eqref{g_bound:c} is due to the $L$-smoothness of $F_n$.

Next, we expand $-2\eta\left\langle \bar{\boldsymbol{\omega}}^{t}_{i-1}-\boldsymbol{\omega}^{*},\bar{\mathbf{g}}^{t}_{i-1}\right\rangle$ as
\begin{small} \begin{align}
 -\!2\eta\left\langle \bar{\boldsymbol{\omega}}_{i-\!1}^{t}\!\!-\!\boldsymbol{\omega}^{*}\!,\bar{\mathbf{g}}_{i-\!1}^{t}\right\rangle \!=&\!-\!2\eta{\sum}_{\forall n}p_{n}\left\langle \bar{\boldsymbol{\omega}}_{i-\!1}^{t}\!-\!\boldsymbol{\omega}_{n,i-\!1}^{t},\mathbf{g}_{n,i-\!1}^{t}\right\rangle \!\!\notag\\&\!\!\!\!\!\!\!\!\!\!\!\!\!\!\!\!\!\!\!\!-\!2\eta{\sum}_{\forall n}p_{n}\left\langle \boldsymbol{\omega}_{n,i-\!1}^{t}\!-\!\boldsymbol{\omega}^{*},\mathbf{g}_{n,i-\!1}^{t}\right\rangle .\label{A_2}
\end{align}\end{small}%
Since $F_n$ is $L$-smooth, the first term on the RHS of \eqref{A_2} is upper bounded by%
 \begin{small}\begin{align}
		\!\!\!\!\!\!-\!\!\left\langle \!\bar{\boldsymbol{\omega}}_{i\!-\!1}^{t}\!\!-\!\!\boldsymbol{\omega}_{n,i\!-\!1}^{t},\mathbf{g}_{n,i\!-\!1}^{t}\!\right\rangle \!\!\leq\!\!F_{n}\!(\boldsymbol{\omega}_{n,i\!-\!1}^{t})\!\!-\!\!F_{n}\!(\bar{\boldsymbol{\omega}}_{i\!-\!1}^{t}\!)\!\!+\!\!\frac{L}{2}\!\Vert \bar{\boldsymbol{\omega}}_{i\!-\!1}^{t}\!\!-\!\!\boldsymbol{\omega}_{n,i\!-\!1}^{t}\!\Vert ^{2}\!\!.%
  \end{align}%
  \end{small}%
Since $F_n$ is $\mu$-strongly convex, the second term is upper bounded by
       \begin{small}
            \begin{align}	\!\!\!\!\!-\! \!\left\langle\!\boldsymbol{\omega}^{t}_{n,i\!-\!1}\!\!-\!\!\boldsymbol{\omega}^{*}\!,\mathbf{g}^{t}_{n,i\!-\!1}\!\right\rangle \!\!\leq\!\!-\!F_{n}(\boldsymbol{\omega}^{t}_{n,i\!-\!1})\!\!+\!\!F_{n}(\boldsymbol{\omega}^{*})\!\!-\!\!\frac{\mu}{2}\!\Vert \boldsymbol{\omega}^{t}_{n,i\!-\!1}\!\!-\!\!\boldsymbol{\omega}^{*}\Vert ^{2}.\label{A_2-2}
		\end{align}
       \end{small}
  
  By substituting \eqref{g_bound} -- \eqref{A_2-2} into \eqref{triangle i}, it follows that%
  \begin{small}
      \begin{subequations}
    \begin{align}	\notag&\!\!\!\!\!\!\!\!\big\Vert{\sum}_{\forall n}p_{n}\left(\boldsymbol{\omega}^{t}_{n,i\!-\!1}\!-\!\eta\mathbf{g}^{t}_{n,i\!-\!1}\right)\!-\!\boldsymbol{\omega}^{*}\big\Vert^{2}\\\notag
    \leq&\Delta\boldsymbol{\omega}^{t}_{i\!-\!1}\!\!+\!\!2\eta^{2}L^{2}\!{\sum}_{\forall n}p_{n}\!\left\Vert \boldsymbol{\omega}_{n,i\!-\!1}^{t}\!\!-\!\!\bar{\boldsymbol{\omega}}_{i\!-\!1}^{t}\right\Vert ^{2}\!\!\!+\!\!4L\eta^{2}\left(F(\bar{\boldsymbol{\omega}}_{i\!-\!1}^{t})\!\!-\!\!F^{*}\right)\notag\\\notag&\!\!\!\!-\!2\eta{\sum}_{\forall n}p_{n}\!\left(\!F_{n}\!\left(\boldsymbol{\omega}^{t}_{n,i\!-\!1}\right)\!\!-\!\!F_{n}\!(\boldsymbol{\omega}^{*})\!\!+\!\!\frac{\mu}{2}\!\left\Vert \boldsymbol{\omega}^{t}_{n,i\!-\!1}\!\!-\!\!\boldsymbol{\omega}^{*}\right\Vert ^{2}\!\right)\!\\\label{delta rea}&\!\!\!\!-2\!\eta{\sum}_{\forall n}p_{n}\!\big(F_{n}(\bar{\boldsymbol{\omega}}_{i\!-\!1}^{t}\!)\!\!-\!\!F_{n}(\boldsymbol{\omega}_{n,i\!-\!1}^{t})\!\!-\!\!\frac{L}{2}\left\Vert \bar{\boldsymbol{\omega}}_{i\!-\!1}^{t}\!\!-\!\!\boldsymbol{\omega}_{n,i\!-\!1}^{t}\right\Vert ^{2}\big)
			\\
			=&2\eta^{2}L^{2}\!{\sum}_{\forall n}p_{n}\!\Vert \bar{\boldsymbol{\omega}}_{i\!-\!1}^{t}\!\!-\!\!\boldsymbol{\omega}_{n,i\!-\!1}^{t}\Vert ^{2}\!\!\!+\!\!(4L\eta^{2}-2\eta)\left(F(\bar{\boldsymbol{\omega}}_{i\!-\!1}^{t})\!\!-\!\!F^{*}\right)\!\notag\\&+\!(L\!\!+\!\!\mu)\eta{\sum}_{\forall n}\!p_{n}\!\!\left\Vert \bar{\boldsymbol{\omega}}_{i\!-\!1}^{t}\!\!-\!\!\boldsymbol{\omega}_{n,i\!-\!1}^{t}\right\Vert ^{2}\!\!+\!\left(\!1\!-\!\frac{\mu\eta}{2}\right)\Delta\boldsymbol{\omega}_{i-1}^{t}\label{delta reb}\\
  \leq& (\!2\eta^{2}L^{2}\!\!+\!\!(L\!\!+\!\!\mu)\eta)\!{\sum}_{\forall n}\!p_{n}\!\left\Vert \bar{\boldsymbol{\omega}}_{i\!-\!1}^{t}\!\!-\!\!\boldsymbol{\omega}_{n,i\!-\!1}^{t}\right\Vert ^{2}\!\!\!+\!\!\left(\!1\!\!-\!\!\frac{\mu\eta}{2}\!\right)\!\Delta\boldsymbol{\omega}_{i\!-\!1}^{t},\label{delta rec}
\end{align}%
\end{subequations}%
  \end{small}%
where \eqref{delta rec} is based on $4L\eta^{2}\!-\!2\eta<0$ under the assumption that $0<\eta<\frac{1}{2L}$ and ${\sum}_{\forall n}p_{n}\left(F_{n}(\bar{\boldsymbol{\omega}}^{t}_{i-1})-F_{n}^{*}\right)>0$ with $F_n^*=\underset{\boldsymbol{\omega}}{\min} F_n(\boldsymbol{\omega})\leq F_{n}(\bar{\boldsymbol{\omega}}_{i-1}^{t})$.

Substituting \eqref{delta rec} into \eqref{delta_i_sgd} yields \eqref{delta_ib}.

	\bibliographystyle{IEEEtran}
	\bibliography{ciations}
\end{document}